\documentclass[11pt,letterpaper]{article}
\usepackage{style}
\usepackage{csquotes}
\usepackage{todonotes}
\usepackage{subcaption}
\usepackage{float}
\usepackage{enumitem}
\usepackage{hyperref}
\usepackage{authblk}  
\usepackage{natbib}
\usetikzlibrary{backgrounds,fit,shapes,positioning}

\title{Approximability Landscape of Welfare Maximization within Fair Allocations\thanks{The paper is accepted by the twenty-sixth ACM conference on Economics and Computation (EC'25).}}

\date{}
\author[1]{Xiaolin Bu}
\author[2]{Zihao Li}
\author[3]{Shengxin Liu}
\author[4]{Jiaxin Song}
\author[1]{Biaoshuai Tao}

\affil[1]{Shanghai Jiao Tong University, \texttt{lin\_bu@sjtu.edu.cn}, \texttt{bstao@sjtu.edu.cn}}
\affil[2]{Nanyang Technological University, \texttt{zihao004@e.ntu.edu.sg}}
\affil[3]{Harbin Institute of Technology, Shenzhen, \texttt{sxliu@hit.edu.cn}}
\affil[4]{University of Illinois, Urbana-Champaign, \texttt{jiaxins8@illinois.edu}}

\begin{document}
\maketitle
\begin{abstract}
The problem of fair allocation of indivisible goods studies allocating a set of $m$ goods among $n$ agents in a fair manner. 
While fairness is a fundamental requirement in many real-world applications, it often conflicts with (economic) efficiency.
This raises a natural and important question: How can we identify the most welfare-efficient allocation among all fair allocations? 
This paper gives an answer from the perspective of computational complexity.
Specifically, we study the problem of maximizing \emph{utilitarian social welfare} (the sum of agents' utilities) under two widely studied fairness criteria: \emph{envy-freeness up to any item} (EFX) and \emph{envy-freeness up to one item} (EF1).
We examine both normalized and unnormalized valuations, where normalized valuations require that each agent's total utility for all items is identical.

The key contributions of this paper can be summarized as follows: (i) we sketch the complete complexity landscape of welfare maximization subject to fair allocation constraints; and (ii) we provide interesting bounds on the price of fairness for both EFX and EF1.
Specifically: 
(1) For $n=2$ agents, we develop polynomial-time approximation schemes (PTAS) and provide $\classNP$-hardness results for EFX and EF1 constraints.
(2) For $n>2$ agents, under EFX constraints, we design algorithms that achieve approximation ratios of $O(n)$ and $O(\sqrt{n})$ for unnormalized and normalized valuations, respectively.
These results are complemented by asymptotically tight inapproximability results. We also obtain similar results for EF1 constraints.
(3) When the number of agents is a fixed constant, we show that the optimal solution can be computed in polynomial time by slightly relaxing the fairness constraints, whereas exact fairness leads to strong inapproximability.
(4) Furthermore, our results imply the price of EFX is $\Theta(\sqrt{n})$ for normalized valuations, which is unknown in the literature. 
\end{abstract}

\newpage
\tableofcontents

\newpage
\section{Introduction}
\emph{Fair allocation} is a classical topic studied in economics, social science, and computer science, and has garnered significant attention due to its wide range of applications (e.g., school choices~\citep{abdulkadirouglu2005new}, course allocations~\citep{budish2012multi}, paper review assignments~\citep{lian2018conference}, computational resource allocation~\citep{ghodsi2011dominant}, etc.). 
Generally speaking, fair allocation studies how to fairly allocate a given set of resources to agents with respect to their heterogeneous preferences. 
Among all fair allocations, an optimal allocation can naturally be defined as the one that maximizes \emph{social welfare}. 
The problem of finding such an optimal fair allocation has been extensively studied in the \emph{cake-cutting} literature, where resources are assumed to be (infinitely) divisible~\citep{aumann2015efficiency,aumann2012computing,brams2012maxsum,caragiannis2012efficiency,cohler2011optimal,bei2012optimal}.
In contrast, much less attention has been given to the corresponding problem for \emph{indivisible goods}.
In this paper, we study the computational complexity and approximability of optimal fair allocation for indivisible goods.

\paragraph{Fairness: fundamental desideratum.} In the context of fair allocation, as the name suggests, fairness is a \emph{fundamental desideratum}.
In many societal scenarios, fairness plays a pivotal role in resource allocation.
For example, when the government allocates educational resources, efforts should be made to ensure equitable access for all areas, despite variations in resource utilization across different areas.
Moreover, fairness is also a basic requirement in many computer systems.
In a shared computer system, multiple users utilize the same computational resources. 
The resource allocator should fairly distribute these resources to prevent monopolization by any single user. 
A substantial body of work has explored the design of fair allocation policies~\citep{ghodsi2011dominant,vuppalapati2023karmaresourceallocationdynamic,bhardwaj2023cilantro}.

\paragraph{Fairness notions.}
Different application scenarios give rise to various fairness notions to cope with agents' heterogeneous preferences. 
Among these notions, \emph{envy-freeness} (EF) is arguably the most natural fairness criterion, stating that each agent values her bundle at least as much as the bundle allocated to any other agent.
In other words, no agent should envy any other agent in the allocation.
However, exact fairness may not be achievable when allocating indivisible items (e.g., when the number of items is smaller than the number of agents).
To address this issue, \emph{envy-freeness up to one item} (EF1) is proposed~\cite{Lipton04onapproximately,budish2011combinatorial}, which requires that, for any pair of agents $i$ and $j$, $i$ does not envy $j$ if \emph{an} item is removed from $j$'s bundle.
A stricter notion \emph{envy-freeness up to any item} (EFX) is later introduced~\cite{CKMP+19,gourves2014near}, which requires that, for any pair of agents $i$ and $j$, $i$ does not envy $j$ after the removal of \emph{any} item from $j$'s bundle.
Since then, EF1 and EFX allocations have attracted significant attention. 
The existence of EF1 allocations is well established~\citep{Lipton04onapproximately,CKMP+19}, but the existence of EFX allocations still remains one of the most important open problems in fair allocation~\citep{AmanatidisAzBi23}.
In this work, we mainly focus on these two notions, EF1 and EFX.

\paragraph{Finding the welfare-maximizing allocation.}
While fairness is a fundamental desideratum in fair allocation, (economic) \emph{efficiency} is a crucial desideratum.
Among all fair allocations, it is desirable to select one with higher social welfare.
This leads to an interesting problem: identifying the welfare-maximizing allocation among all fair allocations, which can be formulated as a \emph{constrained optimization} problem.
There are multiple ways to define social welfare, with \emph{utilitarian social welfare} --- the sum of all agents' utilities for the allocation --- being the most natural and widely used metric.
In the cake-cutting setting, utilitarian social welfare has been extensively studied~\citep{aumann2015efficiency,aumann2012computing,SS19,CKMP+19,BKV18,garg2021fair}, including work on maximizing utilitarian social welfare subject to fairness constraints such as envy-freeness or proportionality~\citep{cohler2011optimal,brams2012maxsum,bei2012optimal}.
In other domains of social choice theory, utilitarian social welfare has also attracted great interest, e.g., the Vickrey-Clarke-Groves (VCG) family~\citep{vickrey1961counterspeculation,clarke1971multipart,groves1973incentives} of mechanisms is designed to maximize utilitarian social welfare.

Unfortunately, compared with other domains (especially the cake-cutting domain), this problem is far less understood in the context of indivisible items.
Specifically, the study of this constrained optimization problem in the EFX setting is even rarer, as finding an EFX allocation is already challenging (see Section~\ref{sect:relatedwork} for details).
Recent studies (e.g., \cite{aziz2020computing,barman2020optimal,10.5555/3306127.3331927}) have explored the computational complexity of finding the optimal social welfare\footnote{We will use ``social welfare'' to refer to ``utilitarian social welfare'' unless otherwise specified.} subject to various fairness constraints, such as EF1 and EFX, or the computational complexity of deciding whether there exists an allocation maximizing social welfare (without being subject to a fairness constraint) while being fair (see Section~\ref{sect:relatedwork} for details).
However, these studies mostly deal with \emph{exact} optimal social welfare.
The design of \emph{approximation algorithms} and the study of \emph{approximability} for these constrained optimization problems are mostly absent from the previous literature.
Motivated by the above reasons, we aim to answer the following question in our paper:
\begin{center}
\emph{
What is the computational complexity of finding an allocation that (approximately)\\
maximizes {social welfare} subject to EFX/EF1 constraints?
}
\end{center}

\section{Our Results}
\label{sec:our_results}
Our main contribution is providing a \emph{complete landscape} of the complexity, specifically, the \emph{approximability}, of the problems of optimizing social welfare subject to the EFX and EF1 constraints, respectively.
We represent these two problems as \mswx and \msw.

As a preliminary step, we first present the concept of \emph{resource monotonicity} for EF1 and EFX allocations in \Cref{sec:resource_mono}.
A fairness notion $\Gamma$ satisfies resource monotonicity with respect to social welfare if, for any partial allocation (where some items are unallocated) adhering to $\Gamma$, there always exists a complete allocation that achieves weakly higher social welfare while maintaining fairness $\Gamma$.
We demonstrate that resource monotonicity does not hold for EFX in Lemma~\ref{thm:monotonicityEFX}, even in instances where complete EFX allocations are known to exist.
Thus, for \mswx, we should not exclude partial allocations from our consideration.
However, we find that resource monotonicity holds for EFX allocations with two agents.
Some insights (shown in \Cref{prop:monotonicity}) behind the proof of \Cref{thm:monotonicityEFX} are crucial in designing our algorithm for \mswx with two agents (Section~\ref{sect:PTAS}).
For EF1, resource monotonicity holds trivially by the envy-cycle procedure.
Thus, when dealing with \msw, we always consider complete allocations.

\begin{table}[t]
    \centering
    \small
    \begin{threeparttable}
    \renewcommand{\arraystretch}{1.2} 
      \centering
        \begin{tabular}{cl|rr}
        \hline 
    
        \hline
        \multicolumn{4}{c}{Approximation ratios and inapproximabilities of \mswx.} \\
        \hline 
        {\bf $n$ } & {\bf utility}  & {\bf Positive results} & {\bf Negative results}\\
        \hline
       \multirow{2}{*}{2} & normalized & PTAS (Sect~\ref{sect:PTAS}, Thm~\ref{thm:PTAS-overall}) &  $\classNP$-hard \citep{aziz2020computing} \\
        \cline{3-4}
        & unnormalized & 
        PTAS (Sect~\ref{sect:PTAS}, Thm~\ref{thm:PTAS-overall}) &  $\classNP$-hard\citep{aziz2020computing} \\
        \hline
        \multirow{4}{*}{constant, $\ge 3$} & \multirow{2}{*}{normalized} & $\pseudopoly$ $c\sqrt{n}$ (Thm~\ref{thm:mswx_pos_overall})  & 
        \multirow{2}{*}{$(\sqrt{8n+1}-1)/8$ (Thm~\ref{thm:mswx_neg_overall})} \\ 
        & & {$(1,1-\epsilon)$-BC (Thm~\ref{thm:bicriteria})}
        &  \\ 
        \cline{3-4} 
        & \multirow{2}{*}{unnormalized} & $\pseudopoly$ $(2n+1)$ (Thm~\ref{thm:mswx_pos_overall})  & 
        \multirow{2}{*}{$(n+1)/2$ (Thm~\ref{thm:mswx_neg_overall})} \\ 
        &  & 
        {$(1,1-\epsilon)$-BC} 
        (Thm~\ref{thm:bicriteria}) &  \\ 
        \cline{3-4} 
        \hline 
        \multirow{3}{*}{general} & normalized & $\pseudopoly$ $O(\sqrt{n})$ (Thm~\ref{thm:mswx_pos_overall}) & $\pseudopoly$  $n^{0.5-\epsilon}$ (Thm~\ref{thm:mswx_neg_overall})  \\
        \cline{3-4}
        & \multirow{2}{*}{unnormalized}& \multirow{2}{*}{$\pseudopoly$ $O(n)$ (Thm~\ref{thm:mswx_pos_overall})}  & $\pseudopoly$ $n^{1-\epsilon}$ (Thm~\ref{thm:mswx_neg_overall}) \\
        &  &  &  $(n^{1-\epsilon},0.5+\epsilon)$-BC (Thm~\ref{thm:hardness-bicriteria-overall})     \\
        
        \hline 
        
        \hline 
        \multicolumn{4}{c}{Approximation ratios and inapproximabilities of \msw.} \\
        \hline 
        \multirow{2}{*}{2} & normalized & FPTAS (Sect~\ref{sect:FPTAS}, Thm~\ref{thm:PTAS-overall}) &  
        $\classNP$-hard(Thm~\ref{thm:PTAS-overall}) \\
        \cline{3-4}
        & unnormalized & 
        FPTAS (Sect~\ref{sect:FPTAS}, Thm~\ref{thm:PTAS-overall}) &  $\classNP$-hard(Thm~\ref{thm:PTAS-overall}, \citep{aziz2020computing}) \\
        \hline
            \multirow{4}{*}{constant, $\ge 3$} & \multirow{2}{*}{normalized} & 
        $\approx12\sqrt{n}$~\citep{barman2020optimal}
        & \multirow{2}{*}{$4n/(3n+1)$, $n^{\frac14}/4.5$ (Thm~\ref{thm:msw_neg_overall})} \\
        & & $(1,1-\epsilon)$-BC 
        (Thm~\ref{thm:bicriteria})
        &  
         \\ 
        \cline{3-4}
        & unnormalized & $(1,1-\epsilon)$-BC 
        (Thm~\ref{thm:bicriteria}) 
        &  
        $\lfloor(1 + \sqrt{4n-3}) /2\rfloor $ (Thm~\ref{thm:msw_neg_overall}) \\ 
        \hline 
        \multirow{3}{*}{general} & normalized & $O(\sqrt{n})$~\citep{barman2020optimal} & $n^{\frac{1}{3}-\epsilon}$, $m^{\frac{1}{2}-\epsilon}$ (Thm~\ref{thm:msw_neg_overall})  \\
        \cline{3-4}
        & unnormalized & $n$ (Thm~\ref{thm:1/n}) & 
        \begin{tabular}{r}
             $m^{1-\epsilon}$~\citep{10.5555/3306127.3331927} \\
             $(n^{\frac12-\epsilon},\epsilon)$-BC,
             $(m^{1-\epsilon},\epsilon)$-BC (Thm~\ref{thm:hardness-bicriteria-overall})
        \end{tabular} \\

        \hline
      \end{tabular}
      \begin{tablenotes}
          \footnotesize
          \item $\circ$ $n$ and $m$ represent the numbers of agents and items respectively and $c$ represents a constant number.
          \item $\circ$ $\pseudopoly$ marks for positive/negative results for pseudo-polynomial time algorithms.
          \item $\circ$ $(\alpha,\beta)$-BC: bi-criteria optimization with $\alpha$-approximation on social welfare and $\beta$-approximation on EFX (\Cref{def:EFX_approx}) or EF1 (\Cref{def:EF1_approx}).
        \end{tablenotes}      \caption{\label{tab:results}Appropximability landscape of \mswx and \msw.}
    \end{threeparttable}
    \end{table}
\subsection{Approximability for EFX/EF1 + MSW}
This part presents our main contribution---giving a complete landscape of the approximability of \mswx and \msw.
The results are summarized for different cases based on the number of agents (see \Cref{tab:results}).
Note that we examine both normalized and unnormalized valuations, where normalized valuations require that each agent's total utility for all items is identical.
In the main body of the paper, our results for \mswx are presented in Section~\ref{sect:PTAS} and Section~\ref{sec:mswxgeneraln}, and our results for \msw are presented in Section~\ref{sect:FPTAS} and Section~\ref{sec:mswgeneraln}, with most of the proofs deferred to the appendix.
We note that the results for bi-criteria optimization are discussed in \Cref{sec:bi-critera}.
Below, we highlight the key interesting features and observations from our results for different cases.

\paragraph{Two agents.}
It is known that the prices of EF1 and EFX for two agents are bounded away from $1$~\citep{bei2021price,LiLiLu24}.
Consequently, the maximum social welfare, $\MSW = \sum_{g\in [m]} \max_i v_i(g)$, cannot be used as a target to achieve any approximation ratio of the optimal values for our problems.
However, even without relying on the maximum social welfare as the target value, we find that an FPTAS and a PTAS can be achieved for \msw and \mswx, respectively. This complements the results on prices of fairness -- even if the optimal welfare of EFX/EF1 allocations may be constantly bounded away from the maximum social welfare, it can still be efficiently approximated within any given ratio.
Meanwhile, we also show the $\classNP$-hardness for this case, which successfully bridges the gaps between the approximation and inapproximability ratios.
In particular, \Cref{thm:PTAS-overall} resolves an open problem raised by~\cite{aziz2020computing}.

\paragraph{General number of agents.}
When the number of agents is generalized, for \mswx, we asymptotically close the gap between the approximation and the inapproximability ratios, achieving a near-complete understanding of the problem.
The ratios are $\Theta(\sqrt{n})$ and $\Theta(n)$ for normalized and unnormalized cases, respectively (\Cref{thm:mswx_pos_overall}).
Additionally, we provide stronger inapproximability results for \msw.
Our intractability results (\Cref{thm:hardness-bicriteria-overall})
demonstrate polynomial inapproximability factors, even when EFX and EF1 constraints are significantly relaxed.

\paragraph{Constant number of agents.}
We also consider a notable case where the number of agents, $n$, is a fixed constant integer larger than $2$.
By slightly relaxing EFX/EF1 by a factor of $(1-\epsilon)$, we find that the optimal solution can be computed in polynomial time, as indicated by our results in bi-criteria optimization (\Cref{thm:bicriteria}).
However, if exact fairness is required, we show that strong inapproximability ratios still hold for the two problems and become increasingly severe as $n$ increases (\Cref{thm:mswx_neg_overall} and \Cref{thm:msw_neg_overall}).

\subsection{Price of EFX and EF1}
Our results imply that the price of EFX is $\Theta(n)$ and $\Theta(\sqrt{n})$ for unnormalized and normalized valuations, respectively.
The known results for the price of EFX/EF1 are presented in \Cref{tab:priceoffairness}.
\begin{table}[h!]
    \centering
    \small
    \setlength{\tabcolsep}{3pt}
    \begin{tabular}{cl|cc}
    \hline 
    
    \hline
      & {\bf utility} & {\bf $n=2$} & {\bf general $n$} \\
      \hline
        \multirow{2}{*}{EFX} & normalized  & $1.5$ \citep{bei2021price} & $\Theta(\sqrt{n})$ (Thm~\ref{thm:price_of_fairness}) \\
         & unnormalized  & $2$ (\citep{LiLiLu24}, Thm~\ref{thm:price_of_fairness}) & $\Theta(n)$ (Thm~\ref{thm:price_of_fairness})\\
         \hline
         \multirow{2}{*}{EF1} & normalized  & $8/7$ \citep{bei2021price,LiLiLu24}\hspace{0.3cm} & $\Theta(\sqrt{n})$ \citep{bei2021price,barman2020optimal} \\
         & unnormalized  & $2$ (Thm~\ref{thm:price_of_fairness}) & $n$ (Thm~\ref{thm:price_of_fairness}) \\
         \hline

         \hline
    \end{tabular}
    \caption{Known results and our results for the price of fairness.}
    \label{tab:priceoffairness}
\end{table}

We remark that, following the convention in the previous work (e.g., \cite{LiLiLu24}), we have defined the price of EF1/EFX such that the allocation space includes both complete and partial allocations.
We will elaborate on this later in Remark~\ref{rmk:pof}.

\section{Related Work}
\label{sect:relatedwork}
\paragraph{Fair and efficient allocation.}
The study of fair and efficient allocation in the context of (in)divisible resources is extensive~\citep{aumann2015efficiency,aumann2012computing,SS19,CKMP+19,BKV18,garg2021fair}.
Several studies considered the problem of maximizing social welfare within fair allocations, which is the focus of this paper. 
It is worth noting that, for indivisible items, most papers~\citep{CKMP+19,BKV18,garg2021fair} focused on Nash social welfare, while exploration of utilitarian social welfare is comparatively limited.
For divisible goods, the complexity of the problem is well understood~\citep{cohler2011optimal,brams2012maxsum,bei2012optimal}.
Specifically, maximizing social welfare within envy-free/proportional allocations for piecewise-constant valuations can be optimally solved~\citep{cohler2011optimal}.
However, if each agent is required to receive a connected piece, this problem admits a polynomial inapproximability factor~\citep{bei2012optimal}.

On the other hand, there are some recent papers that consider the problem of maximizing social welfare within EF1 allocations for indivisible goods. 
\cite{10.5555/3306127.3331927} proved that the problem is $\classNP$-hardfor $2$ agents, even when the valuation of one agent is a scaled version of the other (\Cref{thm:NPhardTwo}). The problem is $\classNP$-hardto approximate to within a factor of $m^{1-\epsilon}$ for any $\epsilon>0$ for general numbers of agents $n$ and items $m$ (\Cref{thm:msw_neg_overall}).
They also presented a $1/2$-approximation algorithm for this problem with the dominant-strategy incentive-compatible (DSIC) property, under the special setting where all agents' valuations are scaled versions of a single valuation function (referred to as the ``single-parameter setting'').
\cite{barman2020optimal} presented an algorithm that outputs an EF1 allocation achieving social welfare of at least $O(\OPT/\sqrt{n})$, with $\OPT$ representing the optimal social welfare of all feasible allocations.
\cite{aziz2020computing} showed that the problem is still $\classNP$-hard for $n\ge 3$, even when the utility function is normalized (we improve it in Lemma~\ref{thm:NPhardGeqTwo}), and left it as an open problem when there are $n=2$ agents\footnote{Nevertheless, \cite{aziz2020computing} showed that the problem is $\classNP$-hardfor two agents with \emph{unnormalized} valuations.} (we resolve it in \Cref{thm:PTAS-overall}).
In addition, they proposed a pseudo-polynomial time algorithm when the number of agents is fixed.

Most previous studies focused on EF1, with a limited understanding of the relationship between EFX and social welfare.
\cite{aziz2020computing} showed that the problem of maximizing social welfare subject to the EFX constraint is $\classNP$-hard, and the $\classNP$-hardness result continues to hold for two agents with normalized valuations.
Other than this, to the best of our knowledge, there are no known (in)approximability results for this problem.

Even for EF1, previous studies focused on exact optimal social welfare, rather than its approximability (with the two exceptions of the $m^{1-\epsilon}$ inapproximability~\cite{10.5555/3306127.3331927} and the $O(\sqrt{n})$-approximation algorithm~\cite{barman2020optimal} previously mentioned).
A complete landscape of the approximability for this constrained optimization problem is still missing before this paper.

\paragraph{Resource monotonicity.}
Our paper studies the resource monotonicity with respect to the social welfare.
The resource monotonicity for other efficiency notions like Pareto-optimality and Nash social welfare has been studied by~\cite{caragiannis2019envy} and \cite{ChaudhuryGaMe24}.
\cite{ChaudhuryGaMe24} showed that there exist partial EFX allocations that are not Pareto-dominated by any complete EFX allocations, and there exist partial EFX allocations with Nash social welfare higher than any complete EFX allocations.

\paragraph{Price of fairness.}
The price of fairness measures the loss in social welfare if a fairness constraint is imposed.
The price of EF1/EFX is defined by the supremum (among all fair division instances) of the ratio of the maximum social welfare over the maximum social welfare within EF1/EFX allocations.
The price of EF1 under normalized valuations is $\Theta(\sqrt{n})$:
\cite{barman2020optimal} showed that the price of EF1 is $O(\sqrt{n})$, and \cite{bei2021price} showed that the price of EF1 is $\Omega(\sqrt{n})$.
\cite{bei2021price} and \cite{LiLiLu24} showed that the prices of EF1 and EFX for two agents with normalized valuations are $8/7$ and $1.5$, respectively.
Further, \cite{LiLiLu24} proved that the price of EFX for two agents with unnormalized valuations is at most 2, which complements the result in \Cref{thm:price_of_fairness}.
The price of EFX for a general number of agents is still an open problem (we resolve it in \Cref{thm:price_of_fairness}).

\paragraph{EFX existence, partial EFX allocations, and approximately EFX allocations.}
As mentioned, finding an EFX allocation is challenging and hard to tackle.
The existence of an EFX allocation with additive valuations is known for a small number of agents, e.g., two agents~\citep{plaut2020almost} and three agents~\citep{ChaudhuryGaMe24,AkramiAlCh23,BergerCoFe22}, or for any number of agents with restricted utility functions, e.g., bi-valued \citep{AmanatidisBiFi21,GargMu23}.
The existence of EFX allocations is also known for restricted classes of non-additive utility functions for any number of agents, e.g., identical valuations~\citep{plaut2020almost}, binary submodular valuations~\citep{BabaioffEzFe21}, general binary valuations~\citep{BuSoYu23}, graph valuations~\citep{christodoulou2023graph}, and three types of MMS-feasible utility functions~\citep{ghosal2023four}, or for a small number of items~\citep{mahara2024EFX}.
Some papers also focus on relaxations of EFX, e.g., epistemic EFX~\citep{caragiannis2023EEFX,akrami2024EEFX}, approximately-EFX allocations~\citep{AkramiAlCh23,AmanatidisMaNt20,amanatidis2024graph}, and partial EFX allocations~\citep{BergerCoFe22,caragiannis2019envy,chaudhury2021little,ChaudhuryGaMe23-MOOR}.
The existence and approximation of EFX are also studied in the chore setting~\citep{zhou2024chore,christoforidis2024chore,garg2024chore}.

\paragraph{Fairness notions.}
There are also other criteria, which are not based on envy-freeness, that have attracted interest in the previous literature.
In the context of cake-cutting, there are other common fairness criteria including \emph{proportionality} (PROP), \emph{equitability} (EQ), etc. (see, e.g.,~\cite{BramsTa96,BrandtCoEn16,Endriss17,Moulin19} for a survey).

Since exact fairness cannot be achieved for indivisible items, relaxed versions of the above notions have been proposed, such as PROP1/PROPX and EQ1/EQX, which adapt PROP and EQ to the indivisible resource setting in the same "fairness up to one item" or "fairness up to any item" manner (see ~\cite{Lipton04onapproximately,gourves2014near,CFS17,FSVX19,AZIZ2020573}).

\paragraph{Welfare maximization.}
There has been a series of works focused on maximizing utilitarian social welfare in economics and computer science. 
When the valuations are additive, this can be simply achieved by allocating each item to the agent who favors it. 
However, the problem becomes challenging when there are additional constraints on the allocation.
In addition to fairness constraints, matroid constraints are also studied in a recent work~\citep{dror2023fair}, which showed maximizing social welfare subject to heterogeneous matroid constraints can be done in polynomial time.
Regarding submodular valuations, maximizing the social welfare of indivisible items is shown to be $\classNP$-hard and has $\left(\frac{e}{e-1}\right)$-inapproximability ratio by~\cite{khot2005inapproximability}.
\cite{lehmann2001combinatorial} showed that that problem is a special case of the problem of maximizing a submodular function $f:2^X\rightarrow \mathbb{R}^+$ subject to a matroid $\mathcal{M} = (X, \mathcal{I})$.
There has been a line of papers working on the latter problem~\citep{vondrak2008optimal,calinescu2011maximizing}.

\section{Preliminaries}
\label{sec:prelim}
Let $[k]=\{1,\ldots,k\}$.
Denote by $N = [n]$ the set of agents and $M = [m]$ the set of indivisible items. 
Each agent $i$ has a nonnegative \emph{utility function} $v_i:\{0,1\}^M\to\mathbb{R}_{\geq0}$.
We assume each agent's utility function is \emph{additive}:
$v_i(S)=\sum_{g\in S}v_i(g)$ for every $i\in N$ and $S\subseteq M$, and we denote $v_i(\{g\})$ by $v_{ig}$ or $v_i(g)$ for notation simplicity.
Further, a utility function $v_i$ is said to be \emph{normalized} if $v_i(M)=\sum_{g\in M}v_{ig}=1$, i.e., agent $i$ values exactly 1 for the set of indivisible items $M$.
We will use the two phrases \emph{utility function} and \emph{valuation} interchangeably in this paper.

An \emph{allocation} of the items is the collection of the $n$ item sets $\A = (A_1, A_2, \dots, A_n)$ satisfying $A_i\cap A_j=\emptyset$ for any $i,j\in [n]$, where $A_i$ is the bundle of items allocated to agent $i$.
An allocation is \emph{complete} if $\bigcup_{i=1}^nA_i=M$, i.e., $\A = (A_1, \dots, A_n)$ is a partition of $M$.
We say an allocation is \emph{partial} if it is not complete.
An allocation is \emph{envy-free} if $v_{i}(A_{i}) \geq v_{i}(A_{j})$ for any two agents $i$ and $j$ in $N$.
That is, according to agent $i$'s utility function, agent $i$ does not envy any other agent $j$'s allocation.
An envy-free allocation may not exist in the problem of allocating indivisible items (e.g. when $m<n$). 
We consider two well-known common relaxations of envy-freeness, \emph{envy-freeness up to any item} (EFX) and \emph{envy-freeness up to one item} (EF1), defined below.

\begin{definition}\label{def:EFX}
An allocation $\A = (A_1, \dots, A_n)$ is said to satisfy \emph{envy-freeness up to any item} (EFX), if for any two agents $i$ and $j$, $v_{i}(A_{i})\geq v_{i}(A_{j}\setminus\{g\})$ holds for any $g\in A_j$.
\end{definition}

\begin{definition}\label{def:EF1}
An allocation $\A = (A_1, \dots, A_n)$ is said to satisfy \emph{envy-freeness up to one item} (EF1), if for any two agents $i$ and $j$, there exists an item $g \in A_{j}$ such that $v_{i}(A_{i})\geq v_{i}(A_{j}\setminus\{g\})$.
\end{definition}

For a verbal description, in an EF1 allocation, after removing \emph{some} item $g$ from agent $j$'s bundle, agent $i$ will no longer envy agent $j$.
For EFX, the quantifier \emph{some} is changed to \emph{any}.
Given an allocation $(A_1,\ldots,A_n)$, we say that agent $i$ \emph{envies} agent $j$ if $v_i(A_i)<v_i(A_j)$.
We say that agent $i$ \emph{EFX-envies} agent $j$ if $v_i(A_i)<v_i(A_j\setminus\{g\})$ for \emph{some} $g\in A_j$; agent $i$ \emph{EF1-envies} agent $j$ if $v_i(A_i)<v_i(A_j\setminus\{g\})$ for \emph{every} $g\in A_j$.
By our definition, an allocation is EFX/EF1 if and only if $i$ does not EFX/EF1-envy $j$ for every pair $(i,j)$ of agents.

\begin{remark}\label{remark:EFX1}
Consider an allocation $\A=(A_1,\ldots,A_n)$ and two agents $i$ and $j$.
Since we are considering additive valuations, EFX requires that $v_i(A_i)\geq v_i(A_j\setminus\{g\})$ where $g$ is an item in $A_j$ with \emph{minimum} $v_i(g)$, and EF1 requires that $v_i(A_i)\geq v_i(A_j\setminus\{g\})$ where $g$ is an item in $A_j$ with \emph{maximum} $v_i(g)$.
\end{remark}

It is well-known that a complete EF1 allocation always exists for general additive utility functions, and it can be computed in polynomial time \citep{Lipton04onapproximately,budish2011combinatorial}.
However, the existence of a complete EFX allocation is an open problem for $n\geq 4$.

Another critical issue is economic efficiency, where we consider \emph{social welfare} as defined below.
\begin{definition}\label{def:SW}
The \emph{social welfare} of an allocation $\A=(A_1,\ldots,A_n)$, denoted by $\SW(\A)$, is the sum of the utilities of all the agents
$\SW(\A)=\sum_{i=1}^nv_i(A_i)$.
\end{definition}

In this paper, we focus on the problem of maximizing social welfare subject to the EFX/EF1 constraint. More formally, we have the following two constrained optimization problems.

\begin{problem}[{\sc MSWwithinEFX, MSWwithinEF1}]\label{prob:efx}
Given a set of indivisible items $M=[m]$ and a set of agents $N=[n]$ with their utility functions $(v_1,\ldots,v_n)$, the problem of
\begin{itemize}[leftmargin=0.5cm]
\item \emph{maximizing social welfare within EFX allocations (\mswx)} aims to find an allocation $\A=(A_1,\ldots,A_n)$ that maximizes social welfare $\SW(\A)$ subject to that $\A$ is EFX;
\item \emph{maximizing social welfare within EF1 allocations (\msw)} aims to find an allocation $\A=(A_1,\ldots,A_n)$ that maximizes social welfare $\SW(\A)$ subject to that $\A$ is EF1.
\end{itemize}
\end{problem}

As discussed before, for \mswx, we do not restrict to complete allocations.
However, for \msw, as for any partial allocation, there is a complete allocation with weakly higher social welfare, we can focus exclusively on complete allocations for \msw without loss of generality.
Below we provide more details on the resource monotonicity of EFX/EF1.

\subsection{Resource Monotonicity} 
\label{sec:resource_mono}
The concept of \emph{resource monotonicity} for EF1 and EFX allocations questions whether it is always possible to find a complete EF1/EFX allocation that achieves weakly higher social welfare compared to a given partial EF1/EFX allocation.
Given a partial EF1 allocation, we can apply the well-known \emph{envy-cycle procedure} by \cite{Lipton04onapproximately} to obtain a complete allocation, and each agent's received value is non-decreasing throughout the procedure.
This immediately implies the resource monotonicity of EF1 allocations (described in \Cref{thm:monotonicityEF1}).

\begin{restatable}{lemma}{monotonicityEFOne}
\label{thm:monotonicityEF1}
(Resource monotonicity for EF1)
For any partial EF1 allocation $\A$, we can compute a complete EF1 allocation $\A^*$ in polynomial time such that $\SW(\A^*)\geq\SW(\A)$.
\end{restatable}

Regarding EFX, as shown in \Cref{thm:monotonicityEFX}, the resource monotonicity holds for two agents.
However, when the number of agents is larger, it fails even for the more restrictive normalized valuations. 
The proof of \Cref{thm:monotonicityEFX} is deferred to \Cref{append:monotonicityEFX}.
Here we mainly give a high-level sketch.
The proof of the case of two agents primarily relies on the subroutine within \Cref{prop:monotonicity}, which will be detailed in Section~\ref{sec:two}.
Given a partial allocation $(A, B)$ and an unallocated item $g$, if allocating $g$ to each bundle will cause the recipient to be envied by the other agent,  
the subroutine will compute a new allocation $(A', B')$ such that $A' \cup B' = A\cup B\cup \{g\}$ and the new allocation is EFX and that each one's utility does not decrease.
Thus, we can iteratively apply this subroutine to a given partial allocation until obtaining a complete allocation.
For more than two agents, the counter-example is based on the instance introduced by~\cite{ChaudhuryGaMe24}.

\begin{restatable}{lemma}{efxResourceMono}
\label{thm:monotonicityEFX}
(Resource monotonicity for EFX)
When $n=2$, for any partial EFX allocation $\A$, we can compute a complete EFX allocation $\A'$ such that $\SW(\A') \ge \SW(\A)$ in polynomial time.
However, when $n>2$, there exists an instance $\mathcal{I} = (N,M,(v_1,\ldots,v_n))$ with normalized $v_1,\ldots,v_n$ where (1) complete EFX allocations exist, and (2) there is a partial EFX allocation $\A$ such that $\SW(\A)>\SW(\A')$ for any complete EFX allocation $\A'$.
\end{restatable}

\subsection{Price of Fairness}
The price of fairness measures the loss in social welfare when a fairness constraint is imposed.

\begin{definition}\label{def:maxsocialwelfare}
    Given a valuation profile $(v_1,\ldots,v_n)$, let $\MSW(v_1,\ldots,v_n)=\sum_{g=1}^m\max_{i\in[n]}v_i(g)$ be the maximum social welfare among all allocations (without any fairness constraint).
    We simply write $\MSW$ when the valuation profile is clear from the context.
\end{definition}

\begin{definition}[Price of EFX]
\label{def:priceofEFX}
The \emph{price of EFX} is defined by 
$$
\sup\limits_{(v_1,\ldots,v_n)}\frac{\MSW(v_1,\ldots,v_n)}{\SW(\A^\ast)},
$$
where $\A^\ast$ is an EFX allocation (which is allowed to be partial) with maximum social welfare.
\end{definition}

\begin{definition}[Price of EF1]\label{def:priceofEF1}
    The \emph{price of EF1} is defined by $$
    \sup\limits_{(v_1,\ldots,v_n)}\frac{\MSW(v_1,\ldots,v_n)}{\SW(\A^\ast)},$$
    where $\A^\ast$ is an EF1 allocation with maximum social welfare.
\end{definition}

\begin{remark}\label{rmk:pof}
    When defining the price of fairness (EF1 or EFX in this paper), we follow the convention in the previous literature (e.g., \cite{LiLiLu24}) that allows partial allocations.
    Another way to define the price of fairness only considers complete allocations.
    Both definitions are equivalent in the case of EF1 where resource monotonicity holds (Lemma~\ref{thm:monotonicityEF1}).
    For other fairness notions that fail the resource monotonicity property (or the satisfiability of the resource monotonicity property is unknown), if allocations with the fairness notion are guaranteed to exist, the two definitions become different.
    The one allowing partial allocations is more commonly used.
    For example, the existence of EFM allocations for mixed divisible and indivisible goods is guaranteed~\cite{bei2021fair} while it is unknown whether the resource monotonicity property holds for EFM; the price of EFM was defined in the first way in which partial allocations are allowed~\cite{LiLiLu24}.
    We follow the same convention for EFX in this paper.
    It should also be remarked that the second definition for the price of EFX (that only considers complete allocations) only makes sense if EFX allocations are guaranteed to exist (which is unknown for now).\footnote{If it turns out that there are instances where EFX allocations do not exist, and if we insist on defining the price of EFX with complete allocations, we have to restrict the space of instances to only include those admitting EFX allocations. A much more natural way to define the price of EFX is to include partial allocations, as it is defined in this paper.}
\end{remark}

\subsection{NP-hardness and NP-completeness Results}
\label{sect:np-hardness}

Our hardness proofs are mainly built upon the classical partition and independent set problems.
The partition problem decides whether a given set of positive integers can be partitioned into two subsets with equal subset-sum, which is known to be $\classNP$-complete~\cite{karp2010reducibility}.
\begin{definition}[Partition]\label{def:partition}
Given a set $S$ of positive integers, the \emph{partition problem} decides whether $S$ can be partitioned into subsets $S_1$ and $S_2$ such that the sums of the numbers in the subsets are equal.
\end{definition}
Meanwhile, we use two variants of the independent set problems. 
The first is the original one, which decides whether a given graph $G$ contains an independent set of at least $k$ vertices for a given input integer $k$ and is also known to be $\classNP$-complete~\cite{karp2010reducibility}.
The second one is the optimization version, which finds the maximum size of the independent set of $G$ and is shown to have $n^{1-\epsilon}$ inapproximability by~\cite{hastad1996clique,khot2001improved,zuckerman2006linear}.\footnote{\cite{hastad1996clique} showed the inapproximability for the ratio $n^{1-\epsilon}$ under the stronger assumption $\classZPP\neq\classNP$. \cite{khot2001improved} described more precisely the inapproximability ratio, but with an even stronger assumption. \cite{zuckerman2006linear} derandomized the construction weakening its assumption to $\classP\neq\classNP$.}
\begin{definition}[Independent Set and Maximum Independent Set]\label{def:indset}
We are given a graph $G=(V,E)$.
The \emph{independent set problem} decides whether $G$ contains an independent set of at least $x$ vertices for a given input integer $x$.
The \emph{maximum independent set problem} seeks for maximum $x$ such that $G$ contains an independent set of $x$ vertices.
\end{definition}
\begin{lemma}
[Inapproximability of Maximum Independent Set]\label{thm:indset}
For any $\epsilon>0$, it is $\classNP$-hard to approximate the maximum independent set within a factor of $n^{1-\epsilon}$, where $n$ is the number of vertices.
\end{lemma}

\section{Maximizing Social Welfare with Two Agents}
\label{sec:two}
This section focuses on the case of two agents.
We begin by presenting a fully polynomial-time approximation scheme (FPTAS) for \msw in Section~\ref{sect:FPTAS}.
Based on the techniques used in Section~\ref{sect:FPTAS} and \Cref{prop:monotonicity} in addition, we present a polynomial-time approximation scheme (PTAS) for \mswx in \Cref{sect:PTAS}.
Finally, in \Cref{sect:hardness2}, we complement these positive results by proving that both \msw and \mswx are $\classNP$-hard.
Notice that, for both \msw and \mswx, our FPTAS and PTAS are applicable to unnormalized valuations, and our $\classNP$-hardness results hold even for normalized valuations.
In addition, the $\classNP$-hardness result for \msw with normalized valuations resolves the open problem raised by~\cite{aziz2020computing}.
We summarize the results for two agents as follows.
\begin{restatable}{theorem}{PTASOverall}
\label{thm:PTAS-overall}
For two agents, both \msw and \mswx are $\classNP$-hard even for normalized valuations.
On the positive side, Algorithm~\ref{alg:fptas} provides an FPTAS for \msw, and Algorithm~\ref{alg:ptas} provides a PTAS for \mswx. 
\end{restatable}

\subsection{A Fully Polynomial-Time Approximation Scheme for \msw}
\label{sect:FPTAS}
This part presents our FPTAS for two agents.
We first consider a natural allocation \( (O_1, O_2) \): \( O_1 = \{ g \in [m] \mid v_{1}(g) \geq v_{2}(g) \} \) and \( O_2 = [m]\setminus O_1\), that maximizes social welfare.
If the allocation satisfies EF1, we have already found the $\OPT$.
Otherwise, below we assume that $(O_1, O_2)$ is not EF1 and agent 2 EF1-envies agent 1 without loss of generality, and present an FPTAS for that case.

\begin{algorithm}[t]
\caption{An FPTAS for \msw}\label{alg:fptas}
\KwInput{two utility functions $v_1(\cdot)$ and $v_2(\cdot)$, and the parameter $\epsilon>0$} 
\KwOutput{an EF1 allocation}
Let $O_1 \leftarrow \{g\in \left[m\right]\ |\  v_{1}(g) \geq v_{2}(g) \}$ and $O_2 \leftarrow  \{g \in \left[m\right]\ |\ v_1(g) < v_2(g) \}$\;
Let $\Pi\leftarrow \emptyset$ be the set of all the candidate allocations\;
\For{each item $g \in O_1$}{
    \For{each item $o \in O_1\setminus \{g\}$}{
        $v(o) \leftarrow v_{1o} - v_{2o}$ and $w(o) \leftarrow v_{2o}$\;
    }
    $A_1'\leftarrow$ the output of FPTAS for Knapsack with parameter $\epsilon$ and item set $[m]\setminus(O_2\cup\{g\})$, value function $v$, weight function $w$ and capacity constraint $(v_2([m])-v_{2}(g))/ 2$\;
    Let $A_2'=[m]\setminus(A_1'\cup\{g\})$\;
    $(A_1, A_2) \leftarrow$ {\sc LocalSearchEF1}($ A_1'\cup\{g\}, A_2'$) \;
    $\Pi \leftarrow \Pi \cup \left\{\left(A_1, A_2\right)\right\}$\;
}
\Return{the allocation with the largest social welfare in $\Pi$} \label{line:choose_largest}\;

\SetKwProg{Fn}{Function}{:}{}
\Fn{\FLocalSearchEF{$A_1$, $A_2$}}{
\While{agent $1$ EF1-envies agent $2$}{
Find an arbitrary item $g \in A_2\setminus O_2$ and $A_2\leftarrow A_2\setminus\{g\}$ \label{line:find_g}\;
\eIf{agent $2$ envies agent $1$ under the partial allocation $(A_1,A_2)$}{
    $(A_1, A_2)\leftarrow (A_2\cup \{g\}, A_1)$ \label{line:add_g_to_agent1}\;
}
{ $(A_1, A_2)\leftarrow (A_1\cup\{g\}, A_2)$\;}
}
\Return{$(A_1, A_2)$} 
}
\end{algorithm}

The FPTAS is shown in Algorithm~\ref{alg:fptas} and works as follows: The bundle $O_2$ is fixed to be allocated to agent 2 and an item $g \in O_1$ is fixed to be given to agent 1.
This item $g$ will be a ``guess'' of the item whose removal ensures agent $2$ does not envy agent $1$ in Definition~\ref{def:EF1}, and we will enumerate all possibilities of $g\in O_1$.
Next, we decide the allocation of the remaining items.
To ensure agent $2$ does not envy agent $1$ after removing $g$ from agent $1$'s bundle, the allocation $(A_1,A_2)$ must satisfy $v_2(A_1\setminus\{g\})\leq v_2(A_2)$, which is equivalent to $v_2(A_1\setminus\{g\})\leq v_2([m]\setminus\{g\})/2$.
Therefore, the problem can be viewed as a classical Knapsack problem, where the capacity of the knapsack is $v_2([m]\setminus\{g\})/2$, the weight of item $o$ is $v_{2}(o)$, and the value is $v_{1}(o)-v_{2}(o)$. 
We then use the well-known FPTAS algorithm for the Knapsack problem (see, e.g., Chapter~8 in the textbook~\citep{vazirani2001approximation}). 
After this step, we have an allocation with almost optimal social welfare while ensuring agent 2 does not EF1-envy agent 1.
However, agent 1 may envy agent 2.
If so, we use a local search algorithm to make the allocation EF1 with improved social welfare.

\begin{lemma} \label{thm:FPTAS}
Algorithm~\ref{alg:fptas} is an FPTAS for \msw with two agents.
\end{lemma}

Before proving Lemma~\ref{thm:FPTAS}, we define some additional notations.
Let $\Phi_{1\not\rightsquigarrow 2}$ (resp. $\Phi_{2\not\rightsquigarrow 1}$) be the constraint ensuring that agent 1 (resp. agent 2) does not EF1-envy agent 2 (resp. agent 1).
Hence, the EF1 constraint is the conjunction of $\Phi_{1\not\rightsquigarrow 2}$ and $\Phi_{2\not\rightsquigarrow 1}$.
Let $\OPT$ be the optimal social welfare under the EF1 constraint and $\OPT(\Phi_{2\not\rightsquigarrow 1})$ be the optimal social welfare solely subject to $\Phi_{2\not\rightsquigarrow 1}$.
Let $\ALG$ be the social welfare of the allocation output by Algorithm~\ref{alg:fptas}.

Next, we sketch the proof of \Cref{thm:FPTAS} and highlight only the key components.
First, since the initial allocation $(O_1, O_2)$ maximizes social welfare, agent 1 must not envy agent 2.
Otherwise, swapping the two bundles results in higher social welfare, which contradicts the optimality of $(O_1, O_2)$.
Since $\Phi_{2\not\rightsquigarrow 1}$ is a part of the EF1 constraint and the additional constraint $\Phi_{1\not\rightsquigarrow 2}$ can only possibly reduce the optimal social welfare, we have $\OPT \le \OPT(\Phi_{2\not\rightsquigarrow 1})$.
Despite this, with a bit of counterintuition, we show that the two optimal values are essentially the same under the assumption about $(O_1, O_2)$ in \Cref{lem:optcTwo}.
In other words, there exists one among all the allocations maximizing the social welfare subject to $\Phi_{2\not\rightsquigarrow 1}$ that satisfies $\Phi_{1\not\rightsquigarrow 2}$ as well (and is thus EF1).

Afterward, as described above, constraint $\Phi_{2\not\rightsquigarrow 1}$ can be converted to the capacity constraint of the knapsack problem if the item $g$ to be removed from agent 1's bundle is known.
Then we show that the allocation after invoking FPTAS of the knapsack problem already achieves $(1-\epsilon)\cdot \OPT(\Phi_{2\not\rightsquigarrow 1})$ by the correctness of FPTAS of the knapsack problem.

Finally, \Cref{lem:localsearch} demonstrates that the local search subroutine transforms the allocation into an EF1 allocation without decreasing the social welfare, which means the final allocation $(1-\epsilon)$ approximates $\OPT$ (\Cref{lem:Oneepsilon}).
With the guarantee of approximation proved in \Cref{lem:Oneepsilon}, \Cref{thm:FPTAS} holds immediately as we can verify that the algorithm's running time is polynomial in $m$ and $1/\epsilon$ by the property of FPTAS for Knapsack.
The non-trivial part of the time complexity analysis is the running time for the local search, which is analyzed in~\Cref{lem:localsearch}.

\begin{restatable}{proposition}{optcTwo}
\label{lem:optcTwo}
$\OPT=\OPT(\Phi_{2\not\rightsquigarrow 1})$ when agent 2 EF1-envies agent 1 in the initial allocation $(O_1, O_2)$.
\end{restatable}
\begin{proof}
For the sake of contradiction, suppose every allocation satisfying $\Phi_{2\not\rightsquigarrow 1}$ with social welfare $\OPT(\Phi_{2\not\rightsquigarrow 1})$ violates $\Phi_{1\not\rightsquigarrow 2}$, which means agent 1 always EF1-envies agent 2 in those allocations.
Let $(A_1,A_2)$ be such an allocation among those that agent $1$ envies agent $2$ the least, i.e., with $v_1(A_2)-v_1(A_1)$ being minimized.
Then $A_2$ must contain at least one item $g$ of $O_1$.
Otherwise, $O_1\subseteq A_1$, which means agent 1 cannot envy agent 2 by the property of the initial allocation $(O_1, O_2)$.

Consider the partial allocation $(A_1,A_2\setminus \{g\})$ after removing that item from agent 2's bundle.
Since agent 1 envies agent 2 in $(A_1, A_2)$, agent 1 will still envy agent 2 in the partial allocation. 
We then derive contradictions by respectively discussing the cases whether agent 2 envies agent 1 in the partial allocation $(A_1,A_2\setminus \{g\})$.

\begin{itemize}[leftmargin=0.5cm]
    \item \textbf{Case 1:} agent 2 envies agent 1. 
    If we exchange the two bundles, the social welfare will increase, and neither of them will envy the other. 
    We further give item $g$ to agent 1 and consider the allocation $(A_2,A_1)$. 
    It can be observed that $(A_2,A_1)$ has higher social welfare than $(A_1, A_2)$ (exchanging $A_1$ and $A_2\setminus\{g\}$ increases the social welfare, and the reallocation of item $g$ from agent 2 to agent 1 weakly increases the social welfare by our definition of the set $O_1$ where $g$ belongs to) and constraint $\Phi_{2\not\rightsquigarrow 1}$ is still satisfied (agent $2$ does not envy agent $1$ after the exchange, so she does not EF1-envy agent $1$ if $g$ is additionally given to $1$). 
    This violates the assumption that $(A_1, A_2)$ is an optimal solution.
    \item \textbf{Case 2:} agent 2 does not envy agent 1. 
    Then the new allocation $(A_1\cup\{g\},A_2\setminus\{g\})$ has a weakly higher social welfare than $(A_1, A_2)$ (since $g\in O_1$) while still meets $\Phi_{2\not\rightsquigarrow 1}$.
    It violates our assumption that $(A_1,A_2)$ minimizes the amount of envy, as reallocating item $g$ from agent $2$ to agent $1$ reduces the amount of envy while still keeping social welfare optimal subject to $\Phi_{2\not\rightsquigarrow 1}$.
    \qedhere
\end{itemize}
\end{proof}

\begin{restatable}{proposition}{localsearch}
\label{lem:localsearch}
\FLocalSearchEF outputs an EF1 allocation $(A_1^o,A_2^o)$ that has weakly higher social welfare than that of the input $(A_1^i,A_2^i)$, and it terminates after at most $m$ while-loop iterations.
\end{restatable}
\vspace{-3mm}
\begin{proof}
The first paragraph in the proof of \Cref{lem:optcTwo} shows the existence of item $g$ at Line~\ref{line:find_g} of \FLocalSearchEF.
The analysis of the two cases in the proof of \Cref{lem:optcTwo} shows that the social welfare weakly increases after each while-loop iteration.
It remains to show that the while-loop terminates after at most $m$ iterations.

In Case 1, where agent $2$ envies agent $1$ after removing $g$, the algorithm terminates immediately after exchanging two agents' bundles. 
Although the algorithm may not terminate immediately in Case 2, we observe that the size of $A_1$ is increased by $1$ in each iteration corresponding to Case 2.
If Case 1 happens after many iterations corresponding to Case 2, we know the algorithm will terminate after one more iteration.
The increasing size of $A_1$ ensures that the algorithm terminates after at most $m$ iterations in this scenario.
If Case 1 never happens, agent 1 will not envy agent 2 at some stage when more and more items in $O_1$ are added to $A_1$ since agent 1 does not envy agent 2 in the initial allocation $(O_1, O_2)$.
Hence, the algorithm will terminate with an EF1 allocation.
Again, the increasing size of $A_1$ ensures that the algorithm terminates after at most $m$ iterations.
\end{proof}

\begin{restatable}{proposition}{Oneepsilon} 
\label{lem:Oneepsilon}
$\ALG \ge (1-\epsilon)\OPT$.        
\end{restatable}
\begin{proof}
Let $(S_1,S_2)$ be an allocation corresponding to both $\OPT$ and $\OPT(\Phi_{2\not\rightsquigarrow 1})$ (see \Cref{lem:optcTwo}).
First, it is easy to see that $O_2\subseteq S_2$, and so $S_1\subseteq O_1$.
Otherwise, if an item in $O_2$ is allocated to agent~$1$, reallocating it to agent $2$ strictly increases the social welfare while $\Phi_{2\not\rightsquigarrow 1}$ is still satisfied.

Since $(S_1,S_2)$ is EF1, there exists $g\in S_1$ such that $v_2(S_2)\geq v_2(S_1\setminus\{g\})$.
For each $o\in O_1$, let $v(o)=v_{1}(o)-v_{2}(o)$ as it is in Line~5 of Algorithm~\ref{alg:fptas}.
We can write $\OPT$ and $\OPT(\Phi_{2\not\rightsquigarrow 1})$ as
\begin{equation*}
    \OPT=\OPT(\Phi_{2\not\rightsquigarrow 1})=\sum_{o\in S_1}v_{1o}+\sum_{o\in S_2}v_{2o}=\sum_{o\in S_1}v(o)+\sum_{o=1}^mv_{2o}=v_2([m])+v(g)+\sum_{o\in S_1\setminus\{g\}}v(o).
\end{equation*}

Consider the for-loop iteration at Line~3 where item $g$ is in consideration.
Since $v_2(S_2)\geq v_2(S_1\setminus\{g\})$, we have $\sum_{o\in S_1\setminus\{g\}}v_{2o}\leq\frac12(v_2([m])-v_2(g))$, so $S_1\setminus\{g\}\subseteq O_1$ is a valid solution to the Knapsack problem at Line~6.
By the nature of FPTAS, $A_1'$ output at Line~6 must satisfy $\sum_{o\in A_1'}v(o)\geq (1-\epsilon)\sum_{o\in S_1\setminus\{g\}}v(o)$.
The social welfare for the allocation $(A_1' \cup \{g\},A_2')$ satisfies
\begin{align*}
\SW(A_1'\cup\{g\},A_2') & =\sum_{o\in A_1'\cup\{g\}}v_{1o}+\sum_{o\in A_2'}v_{2o}=\sum_{o\in A_1'\cup\{g\}}v(o)+\sum_{o=1}^mv_{2o} \\
& =v_2([m])+v(g)+\sum_{o\in A_1'}v(o) \\
& \geq v_2([m])+v(g)+ (1-\epsilon)\sum_{o\in S_1\setminus\{g\}}v(o)\\
& >(1-\epsilon)\OPT.
\end{align*}

Finally, \Cref{lem:localsearch} and our choice of the allocation with the largest social welfare (Line~\ref{line:choose_largest}) ensure that the final allocation has social welfare that is at least $\SW(A_1'\cup\{g\}, A_2')$.
\end{proof}

\subsection{A Polynomial-Time Approximation Scheme for \mswx}
\label{sect:PTAS}
Our polynomial-time approximation scheme for \mswx is built upon the algorithm in Section~\ref{sect:FPTAS}.
However, additional techniques are needed because the local search subroutine in Algorithm~\ref{alg:fptas} may not preserve the EFX property.
For an allocation $(A_1,A_2)$, it is possible that i) agent 1 EFX-envies agent 2 and ii) moving \emph{any} item from $A_2$ to $A_1$ makes agent~$2$ EFX-envy agent~$1$.
Notice that the operation at Line~\ref{line:add_g_to_agent1} of Algorithm~\ref{alg:fptas}, while preserving the EF1 property, may not preserve the EFX property.
This is where \Cref{prop:monotonicity} comes into play.

\begin{proposition}\label{prop:monotonicity}
    Consider two agents with utility functions $v_1$ and $v_2$.
    Let $A,B\subseteq M$ and $g\in M$ satisfy $A\cap B=\emptyset$ and $g\notin (A\cup B)$.
    If agent $1$ envies agent $2$ in the allocation $(A,B\cup\{g\})$ and agent $2$ envies agent $1$ in the allocation $(A\cup\{g\},B)$, then there exists an allocation $(A',B')$ with $A'\cup B'=A\cup B\cup\{g\}$ such that (i) $(A',B')$ is EFX, and (ii) $v_1(A')\geq v_1(A)$ and $v_2(B')\geq v_2(B)$.
    In addition, given $A, B$, and $g$, the allocation $(A',B')$ can be computed in polynomial time.
\end{proposition}
\begin{proof}
We first show that a bi-partition $(X,Y)$ of $A\cup B\cup\{g\}$ satisfying the following two properties exists, and it can be computed in polynomial time.
\begin{enumerate}
    \item[(1)] $(X,Y)$ is an EFX allocation if both agents' utility functions are identically $v_2$.
    \item[(2)] $\min\{v_2(X),v_2(Y)\}\geq v_2(B)$.
\end{enumerate}
If such a bi-partition $(X,Y)$ exists, we can obtain an allocation $(A',B')$ that satisfies the two conditions in the proposition. 
The allocation is defined as follows: let agent $1$ pick one of $X$ or $Y$ with a higher value, and let agent $2$ pick the other bundle.
It is clear that the new allocation is EFX (this is similar to the analysis of the \emph{I-cut-you-choose} algorithm by viewing agent $2$ as the cutter: agent $1$ does not envy agent $2$ since she picks first, and agent $2$ does not EFX-envy agent $1$ due to (1) above).
Besides, both agents receive weakly higher values than they would have received in the allocation $(A,B)$.
This is obvious for agent $2$ due to (2) above.
For agent $1$, since agent $1$ envies agent $2$ in the allocation $(A,B\cup\{g\})$, $v_1(A)<\frac12v_1(A\cup B\cup\{g\})$.
Now, by receiving a bundle in $X,Y$ with a higher value, agent $1$ receives at least $\frac12v_1(A\cup B\cup\{g\})$.

It then remains to show such a bi-partition $(X,Y)$ exists and can be computed in polynomial time.
We describe a simple algorithm to compute such a bi-partition.
Starting with $X=A\cup\{g\}$ and $Y=B$.
Perform the following until $(X,Y)$ satisfies (1) above: if $v_2(X)>v_2(Y)$, pick an item with minimum value in $X$ and move it to $Y$; if $v_2(X)<v_2(Y)$, pick an item with minimum value in $Y$ and move it to $X$.

We first show that (2) holds when the algorithm terminates. 
At the beginning of the algorithm, we have $\min\{v_2(X),v_2(Y)\}=v_2(B)$ since agent $2$ envies agent $1$ in the allocation $(X=A\cup\{g\},Y=B)$.
We define the potential function $\phi(X,Y) = \abs{v_2(X) - v_2(Y)}$, which is the same as $\phi(X,Y) 
 = v_2(X\cup Y)-2\min\{v_2(X), v_2(Y)\}$, where $v_2(X\cup Y)$ is a constant.
 Thus, it suffices to show $\phi(X,Y)$ is non-decreasing throughout the algorithm.
After moving an item $g$ from one bundle to the other, if the direction of the inequality between $v_2(X)$ and $v_2(Y)$ is unchanged, $\phi(X,Y)$ clearly does not increase.
If the direction of the inequality changes, we have $\phi(X,Y)>v_2(g)$ before moving $g$ (for otherwise (1) is already satisfied and the algorithm should have stopped before moving $g$) and $\phi(X, Y) <v_2(g)$ after the move, in which case $\phi(X,Y)$ decreases.

Second, the algorithm terminates in $O(m^2)$ iterations.
For each time the direction of the inequality between $v_2(X)$ and $v_2(Y)$ changed, the last item $g$ moved must satisfy $\phi(X,Y)<v_2(g)$.
Moreover, $g$ will no longer be moved in any later iterations.
Otherwise, if $g$ is moved in a future iteration, it must be that $\phi(X,Y)>v_2(g)$ before this move, which contradicts $\phi(X,Y)<v_2(g)$.
Therefore, each change in the inequality direction identifies an item that will never be moved in the future. Thus, the direction of the inequality can be changed at most $m$ times.
The total number of moves is at most $O(m^2)$ as at most $m$ items can be moved between two changes of the inequality direction.
\end{proof}

The next issue is that, although \Cref{prop:monotonicity} guarantees that we can find an EFX allocation $(A_1',A_2')$ with a weakly higher social welfare than the allocation $(A_1,A_2\setminus\{g\})$, the social welfare may be reduced by up to $v_2(g)$ if we update the allocation from $(A_1,A_2)$ to $(A_1',A_2')$.
This is again different from the case with \msw where the social welfare is non-decreasing throughout the local search subroutine.
To ensure that we do not lose too much in the social welfare by applying \Cref{prop:monotonicity}, we need to make sure $v_2(g)$ is small compared to the optimal social welfare.
To accomplish this, we first identify all the ``large items'' for each of which at least one of the agents has a value higher than $\epsilon\cdot\SW(\A^\ast)$. 
There can only be a constant number of large items (i.e., no more than $1/\epsilon$) if the PTAS parameter $\epsilon$ is a constant.
We can just enumerate all possible allocations of these large items before applying Algorithm~\ref{alg:fptas}.
This initial enumeration step makes our algorithm a PTAS instead of an FPTAS.

\begin{algorithm}[htbp]
\caption{PTAS for \mswx with {\color{blue} blue codes} highlighting the differences from the FPTAS.}
\label{alg:ptas}
\KwInput{two valuation functions $v_1(\cdot)$ and $v_2(\cdot)$, and the parameter $\epsilon>0$} 
\KwOutput{an EFX allocation}
Let $\Pi\leftarrow \emptyset$ be the set of all the candidate allocations\;
Let $\Gamma = \frac12\cdot \max\{v_1([m]),v_2([m])\}$\tcp*{a lower bound for the optimal social welfare}
{\color{blue} Let $L=\{g\in [m]\mid\exists i\in\{1,2\}:v_i(g)\geq \frac\epsilon2\cdot \Gamma\}$\tcp*{large items}
Let $S=[m]\setminus L$\tcp*{small items}}
\For{each allocation $(L_1,L_2)$ of $L$}{
    Let $O_1 \leftarrow \{g\in S\mid  v_{1}(g) \geq v_{2}(g) \}$ and $O_2 \leftarrow \{g \in S\mid v_{1}(g) < v_{2}(g) \}$\;
    \If{$(L_2\cup O_2,L_1\cup O_1)$ is envy-free}{
        $\Pi\leftarrow\Pi\cup\{ (L_2\cup O_2,L_1\cup O_1)\}$\;
        \textbf{break}\;
    }
    \If{$(L_1\cup O_1,L_2\cup O_2)$ is EFX}{
        $\Pi\leftarrow\Pi\cup\{ (L_1\cup O_1,L_2\cup O_2)\}$\;
        \textbf{break}\;
    }
    Suppose w.l.o.g. agent $2$ EFX-envies agent $1$ in the allocation $(L_1\cup O_1,L_2\cup O_2)$ \label{line:assume_agent2_envy_agent1}\;
    {\color{blue} $\ell\leftarrow \arg\min_{\ell'\in L_1}v_2(\ell')$\;
    $G\leftarrow\{g\in L_1\cup O_1\mid v_2(g)\leq v_2(\ell)\}$\;}
    \For{each item $g\in G$}{
        {\color{blue} $H\leftarrow \{h\in O_1\mid v_2(h)<v_2(g)\}$\;}
        \For{each item $o \in O_1\setminus (\{g\}\cup H)$}{
            $v(o) \leftarrow v_{1o} - v_{2o}$ and $w(o) \leftarrow v_{2o}$\tcp*{values and weights of items in Knapsack}
        }
        $C\leftarrow \frac12v_2([m]\setminus\{g\})-v_2(L_1\setminus\{g\})$\tcp*{capacity for Knapsack}
        \If{$C<0$}{
            \textbf{continue}\;
        }
        Run the classical FPTAS with parameter $\frac\epsilon2$ for Knapsack with item set {\color{blue} $O_1\setminus (\{g\}\cup H)$}, value function $v$, weight function $w$ and capacity constraint $C$, and let $S_1$ be the output\;
        $A_1'\leftarrow L_1\cup S_1\cup\{g\}$ and
        $A_2'\leftarrow [m]\setminus A_1'$\;
        $(A_1,A_2)\leftarrow$ {\sc LocalSearchEFX}$(A_1',A_2',O_1,O_2)$
        $\Pi\leftarrow\Pi\cup\{(A_1,A_2)\}$\;
    }
}
\Return{the allocation with the largest social welfare in $\Pi$}\;

\SetKwProg{Fn}{Function}{:}{}
\Fn{\FLocalSearchEFX{$A_1$, $A_2$, $O_1$, $O_2$}}{
\While{agent $1$ envies agent $2$ in the allocation $(A_1,A_2)$}{
Find an arbitrary item $g \in A_2\cap O_1$ and $A_2\leftarrow A_2\setminus\{g\}$\;
\eIf{agent $2$ envies agent $1$ in the allocation $(A_1\cup\{g\},A_2)$}{
    Apply the algorithm in \Cref{prop:monotonicity} to obtain allocation $(A_1',A_2')$ for inputs $A_1,A_2,g$\;
    \Return{$(A_1',A_2')$}
}
{ $(A_1, A_2)\leftarrow (A_1\cup\{g\}, A_2)$\;}
}
\Return{$(A_1, A_2)$} 
}
\end{algorithm}

The algorithm is described in Algorithm~\ref{alg:ptas}.
Firstly, $\Gamma=\frac12\cdot \max\{v_1([m]),v_2([m])\}$ at Line~2 is a lower bound to the optimal social welfare (\Cref{prop:triviallowerboundofOPT}).
Line~3 and Line~4 define the large and the small items based on the PTAS parameter $\epsilon$.
Then, we enumerate all possible allocations $(L_1,L_2)$ of the large items as a start-up (Line~5).
The set of the small items $S$ is then partitioned to $(O_1,O_2)$ in a similar way as we did for the EF1 case.
Lines~7-12 handle two trivial cases: the two agents envy each other in the allocation $(L_1\cup O_1,L_2\cup O_2)$ and the allocation $(L_1\cup O_1,L_2\cup O_2)$ is already EFX.
After these, the only possible case is agent~$i$ EFX-envies agent~$j$ and agent~$j$ does not envy agent~$i$, for $(i,j)$ being $(1,2)$ or $(2,1)$.
We assume agent~$2$ EFX-envies agent~$1$ and agent~$1$ does not envy agent~$2$ without loss of generality (Line~\ref{line:assume_agent2_envy_agent1}).

At the next step, we need to enumerate the item $g$ whose removal ensures that agent~$2$ does not envy agent~$1$.
This is trickier than the EF1 case in two aspects. 
Firstly, after $g$ is chosen, all the items with values less than $v_2(g)$ (based on agent~$2$'s valuation) must then be allocated to agent~$2$ (see \Cref{remark:EFX1}). 
Secondly, we need to consider the possibility that $g\in L_1$.
Moreover, there should not be any item $h\in L_1$ with $v_2(h)<v_2(g)$ since $L_1$ is fixed in agent~$1$'s bundle at this moment (again, see \Cref{remark:EFX1}).
Lines~14-17 handle this step.

After the enumeration of the item $g$, we solve the Knapsack problem as we did in the EF1 case (Lines~18-24).
Here, it is possible that the capacity constraint for the Knapsack problem is negative, in which case we just abort the mission.
For example, it is possible that agent~$1$ already receives too much for $L_1$, and the for-loop at Line~16 will do nothing in this case.

Finally, if agent $1$ envies agent $2$ in the allocation obtained from the Knapsack solution, we perform a local search algorithm.
The algorithm iteratively moves an item from agent~$2$'s bundle to agent~$1$'s, and we only move those small items where agent $1$ has higher values (i.e., items in $O_1$).
This keeps going until agent $1$ does not envy agent $2$, and the algorithm will terminate at some point since we have assumed agent $1$ does not envy agent $2$ in the allocation $(L_1\cup O_1, L_2\cup O_2)$.
If, at some middle stage, agent $2$ begins to envy agent $1$, the pre-condition of \Cref{prop:monotonicity} is met, and we can apply \Cref{prop:monotonicity} to finalize the allocation (notice that, at this point, we no longer fix $L_1$ and $L_2$ in the two agents' bundles).
The social welfare is non-decreasing except for the application of \Cref{prop:monotonicity}.
However, the application of \Cref{prop:monotonicity} only reduces the social welfare by at most $v_2(g)$ for some \emph{small} item $g$ in $O_1$, which is acceptable.
The correctness of Algorithm~\ref{alg:ptas} is shown in the following lemma.

\begin{restatable}{lemma}{thmPTASEFX}
\label{thm:ptas}
Algorithm~\ref{alg:ptas} is a PTAS for \mswx.    
\end{restatable}
\begin{proof}
Again, we let $\OPT$ for the value of the optimal solution to \mswx.
Let $\ALG$ be the social welfare of the allocation output by Algorithm~\ref{alg:ptas}.
We first show the following two propositions.

\begin{proposition}\label{prop:triviallowerboundofOPT}
    $\Gamma = \max\{v_1([m]),v_2([m])\}/2\leq\OPT$.
\end{proposition}
\begin{proof}
    Assume $v_1([m])\leq v_2([m])$ without loss of generality.
    Let $(X,Y)$ be an EFX allocation in the valuation profile where both agents' utility functions are $v_1$.
    Consider the EFX allocation $\A$ where agent $2$ gets one of the bundles $X$ and $Y$ with a higher value and agent $1$ gets the other bundle.
    Then $\SW(\A)\geq v_2([m])/2= \max\{v_1([m]),v_2([m])\}/2$.
\end{proof}

\begin{proposition}\label{prop:largeitems}
    For $L$ defined at Line~3 of Algorithm~\ref{alg:ptas}, we have $\abs{L}\leq 8/\epsilon$.
\end{proposition}
\begin{proof}
    Let $L^{(1)}=\{g\in[m]\mid v_1(g)\geq\frac\epsilon2\cdot\Gamma\}$ and $L^{(2)}=\{g\in[m]\mid v_2(g)\geq\frac\epsilon2\cdot\Gamma\}$.
    Then $L=L^{(1)}\cup L^{(2)}$.
    Suppose for the sake of contradiction that $|L|>\frac8\epsilon$.
    There must exist $i\in\{1,2\}$ with $|L^{(i)}|>\frac4\epsilon$.
    Then $v_i([m])\geq v_i(L^{(i)})>\frac4\epsilon\cdot\frac\epsilon2\cdot\Gamma=2\Gamma\geq v_i([m])$, which is a contradiction.
\end{proof}

Finally, we prove the approximation guarantee for Algorithm~\ref{alg:ptas}:    $\ALG\geq (1-\epsilon)\OPT$.
Let $(S_1,S_2)$ be the allocation corresponding to $\OPT$.
Let $L_1$ and $L_2$ be the sets of the large items in $S_1$ and $S_2$, respectively.
Consider the for-loop iteration at Line~5 where $(L_1,L_2)$ is in consideration.
The maximum social welfare is attained at the allocation $(L_1\cup O_1,L_2\cup O_2)$.
Therefore, if the for-loop is broken at Line~9 or Line~12, we have $\ALG=\OPT$.
We assume that agent $2$ EFX-envies agent $1$ in the allocation $(L_1\cup O_1,L_2\cup O_2)$ from now on.

Let $\OPT(C_2)$ be the maximum social welfare of the allocation $(S_1',S_2')$ where i) agent $2$ does not EFX-envy agent $1$ and ii) $L_1\subseteq S_1'$, $L_2\subseteq S_2'$.
We have $\OPT\leq\OPT(C_2)$ since we do not require that agent~$1$ does not EFX-envy agent~$2$ in regarding $\OPT(C_2)$.
It then suffices to show that $\ALG\geq(1-\epsilon)\OPT(C_2)$.

Since agent $2$ does not EFX-envy agent $1$ in $(S_1',S_2')$, we have $v_2(S_2')\geq v_2(S_1'\setminus\{g\})$ for $g\in S_1'$ with minimum $v_2(g)$ (\Cref{remark:EFX1}).
Consider the for-loop iteration at line~16 where $g$ is in consideration.
The social welfare of the allocation obtained by the Knapsack solution is at least $(1-\frac\epsilon2)\OPT(C_2)$.
As we have mentioned before, after the local search step, the social welfare can be decreased by at most $v_2(g)\leq\frac\epsilon2\cdot\Gamma$ for some $g\in O_1$, which is at most $\frac\epsilon2\cdot\OPT\leq\frac\epsilon2\cdot\OPT(C_2)$ by \Cref{prop:triviallowerboundofOPT}.
Therefore, $\ALG\geq (1-\epsilon)\OPT(C_2)$.
\end{proof}

\subsection{NP-Hardness for Two Agents with Normalized Valuations}
\label{sect:hardness2}

We complement our positive results with the following $\classNP$-hardness results. \cite{aziz2020computing} showed the $\classNP$-hardness of \msw with two agents. However, their result applies only to the case with unnormalized valuations. They also posed the corresponding problem with normalized valuations as an open question.
We resolve it in \Cref{thm:NPhardTwo}. 
The $\classNP$-hardness also holds for \mswx with two agents and normalized valuations, shown by~\cite{aziz2020computing}.

\begin{restatable}{lemma}{NPhardTwo}
\label{thm:NPhardTwo}
\msw is $\classNP$-hard for $n=2$ even under normalized valuations.
\end{restatable}

\begin{proof}
    We show a reduction from the partition problem. Fix a partition instance $S= \{e_1, \ldots, e_\ell \}$, and let $\sum_{i=1}^\ell e_i = 2x\in \mathbb{R}^+$. Without loss of generality, we assume $x=1$, and construct an instance as shown in the table below, where $K$ is a constant and $K>1.25\sum_{i=1}^\ell e_i = 2.5$. 
    Note that though $\sum_{o\in [m]}v_1(o)$ and $\sum_{o\in [m]}v_2(o)$ are not normalized, they are both equal to $2K+2$ and can be rescaled to $1$. Therefore, the normalization assumption is not violated.
    
    \begin{center}
    \begin{tabular}{c|cccc}
        \hline
        item &  $k$ ($1\le k\le \ell$) &  $\ell+1$ &  $\ell+2$ &  $\ell+3$  \\
        \hline
        $v_1$ & $e_k$ & $K$ & $K$ & $0$\\ 
        $v_2$ & $e_k/2$ & $(2K+1)/3$ & $(2K+1)/3$ & $(2K+1)/3$\\ 
        \hline
    \end{tabular}
    \end{center}
    
    If the partition instance is a \YES instance, suppose $A_1\subseteq[\ell]$ and $A_2\subseteq[\ell]$ correspond to the partition $(S_1,S_2)$ of $S$ with equal sum. 
    It is not hard to verify that the allocation $\mathcal{A}= ( A_1 \cup \{\ell+1,\ell+2\}, A_2\cup \{\ell+3\})$ satisfies EF1, and $\mathcal{SW}(\mathcal{A})=(16K+ 11)/6$. 
    
    If the partition instance is a \NO instance, we show that the maximum social welfare of an EF1 allocation is less than $(16K+ 11)/6$. 
    If there exists an allocation $\mathcal{A}'$ with the social welfare of at least $(16K + 11)/6$, then it is easy to see that $\ell+1$ and $\ell+2$ must be assigned to agent $1$ and $\ell+3$ must be assigned to agent $2$. 
    Since $\mathcal{SW}(\mathcal{A}')\ge (16K+ 11)/6$, agent $1$ should take a bundle $O_1\subseteq [\ell]$ with at least half of the value of the first $\ell$ items.
    Due to the EF1 constraint, agent $2$ should also take a bundle $O_2\subseteq [\ell]$ with at least half of the value of the first $\ell$ items.
    This would imply the partition instance is a yes-instance, which contradicts our assumption.    
    Thus, the partition problem is reduced to whether $(16K+11)/6$ is achievable for \msw, which leads to the $\classNP$-hardness of \msw.
\end{proof}

\section{$\text{\mswx}$ for More Than Two Agents}
\label{sec:mswxgeneraln}
In this section, we present our results for \mswx with a general number of agents.
We first go through some terms and notations.
We will use Algorithm~\ref{alg:replace} as a subroutine multiple times.
Given a partial allocation where agent $i$ gets $A_i$ and a set of items $B$ with $A_i\cap B=\emptyset$, if agent $i$ envies $B$, Algorithm~\ref{alg:replace} computes $X_i\subseteq B$ and $k_i\in\mathbb{Z}^+$ such that $k_i$ is the minimum integer with $v_i(X_i)>v_i(A_i)$, where $X_i\subseteq B$ is the set of the $k_i$ items in $B$ with the largest values with respect to agent $i$'s valuation $v_i$.
If $v_i(A_i)\geq v_i(B)$, we set $k_i=\infty$.

\begin{algorithm}[ht]
\caption{The replacement subroutine}\label{alg:replace}
\SetKwProg{Fn}{Function}{:}{}
\Fn{\FReplace{$v_i$, $A_i$, $B$}}{
Sort items in $B$ in descending order based on $v_i$\;
Let $B[k]$ be the set of the first $k$ items in $B$\;
$X_i\leftarrow B[k_i]$ where $k_i$ is the minimum integer $k$ such that $v_i(B[k])>v_i(A_i)$\;
\Return{$X_i$} 
}
\end{algorithm}

We will use the idea of \emph{the most envious agent} by~\cite{chaudhury2021little}.
Given a partial allocation $(A_1,\ldots,A_n)$ and a set of items $B$ such that $B\cap A_i=\emptyset$ for each $i=1,\ldots,n$, \emph{the most envious agent to set $B$} is an agent with minimum $k_i=|X_i|$.
The proof of the following proposition is mostly the same as the proof of Lemma~5 in~\cite{chaudhury2021little}, and we include it here for completeness.

\begin{proposition}\label{prop:mostenviousreplace}
    Consider a partial EFX allocation $(A_1,\ldots,A_n)$ and a set of items $B$ such that $B\cap A_i=\emptyset$ for each $i=1,\ldots,n$. Let $i$ be the most envious agent to $B$. Then $(A_1,\ldots,A_{i-1},X_i,A_{i+1},\ldots,A_n)$ is an EFX allocation, where $X_i$ is the output of {\sc Replace}$(v_i,A_i,B)$.
\end{proposition}
\begin{proof}
    Agent $i$ receives a strictly higher value by updating the bundle from $A_i$ to $X_i$, so she will not EFX-envy any other agent since the original allocation $(A_1,\ldots,A_n)$ is EFX.
    It remains to show that any agent $j\neq i$ will not EFX-envy the bundle $X_i$. 
    If $v_j(A_j)\geq v_j(B)$, $j$ will not envy $X_i$ as $X_i\subseteq B$. 
    Otherwise, let $k_j=|X_j|$ where $X_j$ is the output of {\sc Replace}$(v_j,A_j,B)$.
    Agent $j$ will not envy the subset of $B$ consisting of the $(k_j-1)$ items with the highest values to her and thus will not envy any subset of $B$ with at most $k_j-1$ items.
    We have $k_j\geq k_i$ by the definition of the most envious agent.
    Therefore, agent $j$ will not envy any subset of $B$ with at most $k_i-1$ items.
    Thus, agent $j$ will not envy $X_i$ after removing any item from $X_i$.
\end{proof}

\subsection{Approximation Algorithms for \mswx}
\label{sect:On}

\begin{restatable}{theorem}{thmmswxposoverall}
\label{thm:mswx_pos_overall}
For unnormalized valuations, Algorithm~\ref{alg:On} outputs an EFX allocation $\A$ with $\SW(\A)\geq\sum_{i=1}^nv_i([m])/(2n+1)$ in pseudo-polynomial time.
For normalized valuations, Algorithm~\ref{alg:Osqrtn} is a pseudo-polynomial time algorithm that outputs an EFX allocation $\A$ with $\SW(\A)\geq \MSW/ O(\sqrt{n})$.
\end{restatable}
\begin{proof}
We begin with the simpler case where valuations are unnormalized. 
Since $v_i(S) \leq v_i([m])$ for any $S \subseteq [m]$, a trivial upper bound for \(\MSW\) is \(\sum_{i=1}^n v_i([m])\) (i.e., $\MSW = \sum_{g \in [m]} \max v_i(g)$).  
We then demonstrate that an EFX (possibly partial) allocation achieving at least $\sum_{i=1}^n v_i([m]) / (2n + 1)$ can be computed in pseudo-polynomial time.  
Despite the estimate of the upper bound is rough, we will show that it is the best approximation ratio that any pseudo-polynomial time algorithm can achieve in the next subsection.

The algorithm for unnormalized valuations (Algorithm~\ref{alg:On}) works as follows: we create an initial allocation $\mathcal{A}$ with $\mathcal{A}$ maximizes the social welfare such that $\abs{A_i} = 1$ for any $i$, which can be achieved through finding a weighted perfect matching.
The initial allocation satisfies EFX and we set $B$ as the set of unallocated items.
Next, when there exists an agent that envies $B$, we apply the \REPLACE subroutine for the most envious agent and update the current allocation.
It terminates when no one envies the set of unallocated items.
Since every invoke of \REPLACE strictly increases the social welfare, the update process will terminate in pseudo-polynomial time. 
The non-trivial part is the proof for the approximation guarantee $(2n+1)$.
The proof of this is similar to~\cite[Lemma 1]{barman2020optimal}.
The high-level intuitions are described as follows.
Suppose first $|A_i|\geq 2$ for all $i$ in the allocation output by Algorithm~\ref{alg:On}.
The EFX property then ensures $v_i(A_j)$ is at most twice as much as $v_i(A_i)$.
In addition, $v_i(A_i)\geq v_i(B)$ for otherwise the while-loop of the algorithm should be carried on.
Therefore, among the $n+1$ bundles $A_1,\ldots,A_n,B$, agent $i$ gets a bundle with value at least $1/2n$ fraction of $v_i([m])$, which implies Algorithm~\ref{alg:On} is a $2n$-approximation.
The case where $|A_i|\leq 1$ for some agent $i$ is trickier and the maximum-matching initialization at Line~2 is aimed to handle this.
The full proof is deferred to \Cref{appendix:On}.

\begin{algorithm}[t]
\caption{$O(\sqrt{n})$-approximation for \mswx with normalized valuations}\label{alg:Osqrtn}
\KwInput{$(v_1,\ldots,v_n)$} 
\KwOutput{an EFX allocation (that is allowed to be partial)}
Compute an allocation $(O_1,\ldots,O_n)$ with social welfare $\MSW$\;
Initialize $(A_1,\ldots,A_n)$ such that $A_i$ contains one item $g$ with $\max_{g\in O_i}v_i(g)$ if $O_i\neq\emptyset$ and $A_i=\emptyset$ if $O_i=\emptyset$\;
For each $i=1,\ldots,n$, set $B_i\leftarrow O_i\setminus A_i$\;
\While{there exists $j$ such that an agent envies $B_j$}{
    Let $i$ be the most envious agent to $B_j$\tcp*{$i$ may or may not be $j$}
    $X_i\leftarrow${\sc Replace}$(v_i,A_i,B_j)$\;
    Release all items in $A_i$ such that each $g\in A_i$ is added to $B_k$ if $g\in O_k$\;
    $A_i\leftarrow X_i$\;
}
\Return{$(A_1,\ldots,A_n)$}
\end{algorithm}

\paragraph{Normalized Valuations.}
Next, we move to the more complicated case where the valuations are normalized: $v_i([m]) =1$ for every agent $i$.
We show that a ratio of $O(\sqrt{n})$ is achievable.
Our algorithm is presented in Algorithm~\ref{alg:Osqrtn}.
If the maximum social welfare is upper-bounded by $O(\sqrt{n})$, say, $\MSW\leq 10\sqrt{n}$, then we can directly apply the above procedure for unnormalized valuations, which outputs an EFX allocation with the social welfare of at least $n/(2n+1) > 1/3$, which has already achieved an $O(\sqrt{n})$-approximation.
Therefore, we can assume $\MSW > 10\sqrt{n}$ in the following proof.

Algorithm~\ref{alg:Osqrtn} starts by computing a social welfare maximizing allocation $(O_1,\ldots,O_n)$ where $O_i$ consists of those items where agent $i$ values the most (break tie arbitrarily).
The allocation $(A_1,\ldots,A_n)$ is initialized such that agent $i$ takes one most valuable item in $O_i$, or $A_i=\emptyset$ if $O_i=\emptyset$.
Each $B_i$ denotes the pool of the unallocated items in $O_i$.
The algorithm then enters a while-loop.
Whenever there is a pool of the unallocated items $B_j$ that some agent envies, we find the most envious agent $i$ to $B_j$ and replace $A_i$ by some $X_j$ in $B_j$.

Notice that, throughout the algorithm, we have $B_i\subseteq O_i$.
In addition, each $A_i$ is a subset to some bundle $O_j$, i.e., an agent cannot get items from more than one bundle of $O_1,\ldots,O_n$, although $A_i$ may or may not be contained in $O_i$.

Note that \Cref{prop:mostenviousreplace} ensures that the output allocation is EFX.  
Meanwhile, since each iteration of the while-loop strictly improves the social welfare by applying \REPLACE, the algorithm runs in pseudo-polynomial time.  
The approximation is guaranteed by the following proposition.  

\begin{proposition}
\label{prop:Osqrtn}
Suppose $\MSW>10\sqrt{n}$. 
The allocation $\A$ output by Algorithm~\ref{alg:Osqrtn} satisfies with $\SW(\A)\geq\MSW/\alpha$ with $\alpha = {20\sqrt{n}+10}/9 = O(\sqrt{n})$.
\end{proposition}
\begin{proof}
Recall that each $A_j$ is a subset of some $O_i$.
Let $n_i$ be the number of the agents $j$ with $A_j\subseteq O_i$.
Let $\mathcal{L}=\{i\in [n]\mid n_i>\sqrt{n}\}$ and $\mathcal{S}=\{i\in[n]\mid n_i\leq\sqrt{n}\}$.
We have $|\mathcal{L}|<\sqrt{n}$ for otherwise $\sum_{i\in\mathcal{L}}n_i>n$ which contradicts $\sum_{i=1}^nn_i=n$.

Now, consider arbitrary $i$ and $j$ with $A_j\subseteq O_i$.
We show that $v_i(A_i)\geq \frac12 v_i(A_j)$.
If $|A_j|\geq 2$, the inequality holds trivially by the EFX property (the removed item $g$ satisfies $v_i(g)\leq v_i(A_j\setminus\{g\})$ by \Cref{remark:EFX1}).
If $|A_j|=1$, we have $v_i(A_i)\geq v_i(\{g^\ast\})\geq v_i(A_j)$ for $g^\ast$ being the item in $O_i$ with the largest value to agent $i$, where the first inequality is due to that $A_i=\{g^\ast\}$ at the beginning of the algorithm and $v_i(A_i)$ is non-decreasing throughout the algorithm.

Next, for each $i\in\mathcal{S}$, we have $v_i(A_i)\geq\frac1{2\sqrt{n}+1}v_i(O_i)$.
This is because $O_i$ is the disjoint union of at most $\sqrt{n}+1$ bundles $\{A_j\mid A_j\subseteq O_i\}\cup\{B_i\}$ (by definition of $\mathcal{S}$), $v_i(A_i)\geq \frac12 v_i(A_j)$ (just proved), and $v_i(A_i)\geq v_i(B_i)$ (otherwise the while-loop should be carried on).

Finally,
$\SW(\A)=\sum_{i=1}^nv_i(A_i)\geq\frac1{2\sqrt{n}+1}\sum_{i\in\mathcal{S}}v_i(O_i)=\frac1{2\sqrt{n}+1}\left(\MSW-\sum_{i\in\mathcal{L}}v_i(O_i)\right).$
On the other hand, we have
$\sum_{i\in\mathcal{L}}v_i(O_i)\leq\sum_{i\in\mathcal{L}}v_i([m])=\sum_{i\in\mathcal{L}}1=|\mathcal{L}|<\sqrt{n}<\frac1{10}\MSW.$
Putting together, we have $\SW(\A)>\frac9{20\sqrt{n}+10}\cdot\MSW$.
\end{proof}
It is worth noting that, both the two algorithms essentially approximate $\MSW$, which is the optimal social welfare that can be achieved by any allocation.
As a side result, this further implies that the price of EFX is respectively $O(n)$ and $O(\sqrt{n})$ for the two cases. 
\end{proof}

\subsection{Asymptotically Tight Inapproximability Results}
\label{sect:hardnessEFX}
Although the above target values $(\MSW)$ that we attempt to approximate seem a bit ambitious, the above ratios are indeed the best approximation that can be achieved through pseudo-time polynomial algorithms under the assumption that $\classP \neq \classNP$.
Under that assumption, stronger inapproximability results hold for polynomial-time algorithms, even in a more restricted case -- a constant number of agents.
Our results can be summarized in the following theorem.

\begin{restatable}{theorem}{thmmswxnegoverall}
\label{thm:mswx_neg_overall}
It implies $\classP=\classNP$ if one of the following exists for \mswx:
\begin{enumerate}[leftmargin=0.5cm]
    \item A pseudo-polynomial time $n^{1-\epsilon}$-approximation algorithm for any constant $\epsilon>0$;
    \item A pseudo-polynomial time $n^{0.5-\epsilon}$-approximation algorithm for any constant $\epsilon>0$ when agents' valuations are normalized;
    \item A polynomial time algorithm that achieves better than $(k+1)$-approximation for some fixed odd number $n=2k+1$ (with $k\geq1$) of agents;
    \item A polynomial time algorithm that achieves better than $(k/2)$-approximation for some fixed odd number $n=k(2k+1)$ (with $k\geq1$) of agents when agents' valuations are normalized.
\end{enumerate}
\end{restatable}

\begin{proof}[Proof of (1)]
When the valuations are unnormalized, we reduce from the original version of the independent set problem to show the $n^{1-\epsilon}$ inapproximability. 
We present a reduction from the independent set problem.
Given an independent set instance $(G=(V,E),k)$ with $k\geq 3$, we construct an \mswx instance as follows.
The set of agents consists of a super agent $s$ and $t$ groups of normal agents $\{a_{i0},a_{i1},\ldots,a_{i|E|}\}_{i=1,\ldots,t}$, where $a_{i1},a_{i2},\ldots,a_{i|E|}$ in each group $i$ correspond to the $|E|$ edges in $G$.
Notice that $n=1+t(|E|+1)$, we can let $n$ be sufficiently large (but also of polynomial size with respect to $G$) such that $n^{1-\epsilon}<(t+1)/2$. In other words, we should choose a $t>4|E|^{1/\epsilon}$ to make the number of agents in each group small compared with $n$.
The set of items consists of $t+1$ ``super items'' $g_0,g_1,\ldots,g_t$ and $t$ groups of ``normal items'' $\{v_{i1},\ldots,v_{i|V|},e_{i1},\ldots,e_{i|E|}\}_{i=1,\ldots,t}$ such that each group of $|V|+|E|$ items correspond to the $|V|$ vertices and $|E|$ edges in $G$.
In each group $i$, agent $a_{i0}$ has value $k$ on the super item $g_i$, value $1$ on each of $v_{i1},\ldots,v_{i|V|}$.
For $j=1,\ldots,|E|$, each agent $a_{ij}$ only has positive values on the normal items $v_{i1},\ldots,v_{i|V|},e_{i1},\ldots,e_{i|E|}$ in group $i$.
In particular, $a_{ij}$ has value $1$ on $e_{ij}$ and on the two vertex items $v_{iu_1},v_{iu_2}$ where $u_1$ and $u_2$ are the two endpoints of the $j$-th edge.
The super agent has a value of $w$ for each super item, where $w$ is polynomial in $n$ and larger than the optimal social welfare for the instance without the super agent with no fairness constraints (say, $w=n^{100}$), and she has a value of $0$ for remaining items.

If the independent set instance is a \YES instance, we describe an EFX allocation with social welfare at least $(t+1)w$.
The super agent $s$ gets all the super items $g_0,g_1,\ldots,g_t$.
In each group $i$, agent $a_{i0}$ gets a set of $x$ items from $\{v_{i1},\ldots,v_{i|V|}\}$ corresponding to an independent set of size $x$, and agent $a_{ij}$ (for $j=1,\ldots,|E|$) gets the item $e_{ij}$.
The remaining items are discarded.
It is straightforward to check that the allocation is EFX and is, in fact, envy-free.
The optimal social welfare under EF$X$ constraint is lower bounded by $w(t+1)$.

If the independent set instance is a \NO instance, we will show that the super agent $s$ can get at most one super item in any EFX allocation.
Suppose this is not the case.
A super item $g_i$ with $i=1,\ldots,t$ must be allocated to the super agent $s$, and $s$ is allocated at least one more item.
By EFX, agent $a_{i0}$ cannot envy agent $s$, and must receive a value of at least $k$ from $v_{i1},\ldots,v_{i|V|}$.
This means at least $x$ items from $v_{i1},\ldots,v_{i|V|}$.
Since the independent set instance is a \NO instance, agent $a_{i0}$ must receive two items $v_{iu_1},v_{iu_2}$ such that $(u_1,u_2)$ is an edge.
Let $a_{ij}$ and $e_{ij}$ be the agent and the item in the $i$-th group corresponding to this edge, respectively.
Then $a_{ij}$ can receive a value of at most $1$ by getting $e_{ij}$, and the value she has on agent $a_{i0}$'s bundle is $2$.
To maintain EFX, agent $a_{i0}$ must not receive more than the two items $v_{iu_1},v_{iu_2}$.
This contradicts our assumption $k\geq 3$.
Since we have proved that agent $s$ can get at most one super item, the social welfare, in this case, is upper bounded by $2w$ (as the optimal social welfare of other agents is no more than $w$).

Putting the completeness and the soundness parts together, the inapproximability factor is $(t+1)/2$, which is more than $n^{1-\epsilon}$.
In addition, all the values of the items are bounded by $n^{100}$.
\end{proof}

Note that the first two sub-results provide asymptotically tight inapproximability results, matching our approximation algorithms in the last section. 
The last two sub-results provide inapproximability results with ratios of orders $n$ and $\sqrt{n}$ for unnormalized and normalized valuations, respectively, and hold even for a constant number of agents.
However, these results are not aligned with the previous algorithms, which run in pseudo-polynomial time.
Instead, these results stand in contrast to our bi-criteria algorithm provided below: if the EFX requirement is slightly relaxed, we can compute the optimal social welfare in polynomial time; otherwise, we have strong inapproximability results.
We defer the proofs of the remaining three cases to \Cref{appendix:hardness-EFX-overall}.

\section{$\text{\msw}$ for More Than Two Agents}
\label{sec:mswgeneraln}
This section studies the \msw problem for more than two agents.
When the number of agents is a fixed constant integer larger than $2$, in contrast to our previous results for $n=2$, where $\OPT$ can be approximated within any ratio, we establish increasingly stricter inapproximability bounds as $n$ grows when exact fairness constraint should be guaranteed.
Despite this, by slightly relaxing the fairness ratio with a factor of $(1-\epsilon)$ (i.e., \( v_i(A_i) \ge (1-\epsilon) \cdot v_i(A_j \setminus \{g\}) \) for some \( g \in A_j \)), we show that the optimal solution can be approximated within any ratio, which closes the gap of the approximability in terms of fairness constraint. The full proofs are deferred to Appendix~\ref{append:msw}.

\begin{restatable}{theorem}{thmmswnegoverall}
\label{thm:msw_neg_overall}
It is $\classNP$-hard to approximate \msw within a factor of 
\begin{itemize}[leftmargin=0.5cm]
    \item $4n/(3n+1)$ for any fixed agent number $n>2$, even with normalized valuations;
    \item $n^{1/4}/4.5$ for any fixed $n$, even with normalized valuations;
    \item $\lfloor (1+\sqrt{4n-3})/2\rfloor$ for any fixed $n\geq 2$;
    \item $n^{\frac13 -\epsilon}$ or $m^{\frac12-\epsilon}$ for any constant $\epsilon>0$, even with normalized valuations.
\end{itemize}
\end{restatable}

\subsection{Inapproximability for Constant Number of Agents}
As a warm-up, we first show that the problem of \msw has a constant inapproximability factor for any fixed number of agents $n>2$, in contrast to the case of $n=2$, where FPTAS exists.
\begin{restatable}{lemma}{NPhardGeqTwo}
\label{thm:NPhardGeqTwo}
For any fixed $n>2$, \msw is $\classNP$-hard to approximate to factor ${4n}/{(3n+1)}$, even under normalized valuations.
\end{restatable}
\begin{proof}
We show a reduction from the partition problem. 
Given a partition instance $\{e_1, \dots, e_\ell\}$ with $\sum_{i=1}^{\ell}e_i=2x$, we construct an instance for \msw with $n$ agents and $\ell+n-1$ items as follows.
For agent $1$, (s)he has value ${(nx+x)}/{2}$ for the item $\ell+1$ and $\ell+2$, and value $0$ for the remaining items. 
For agent $2$ to agent $n$, they have value $e_{o}$ for item $o\in[\ell]$ and value $x$ for each of the remaining items.
The values will sum up to $nx+x$ for all agents and it can be normalized to $1$, which satisfies the definition of normalized valuations.

If the partition instance is a \YES instance, we assume the partition is $(S_1, S_2)$. Let agent $2$ receive items corresponding to $S_1$, agent $3$ receive items corresponding to $S_2$, and each of the agents $4,\ldots,n$ receive one of the items from item $\ell+3$ to item $n+\ell-1$, and all of them receive items with total value $x$ in this case. Let agent $1$ receive item $\ell+1$ and item $\ell+2$. The allocation is EF1 because agent $1$ will not envy any other agent, and when removing one of the items in agent $1$'s bundle, other agents will not envy agent 1 either.
The social welfare under this case is $2nx$.

If the partition instance is a \NO instance, agent $1$ can only receive one of item $t+1$ and item $t+2$. 
Otherwise, if agent $1$ receives both items, all other agents should receive items whose values sum up to at least $x$, which is impossible in a \NO instance.
In this case, the social welfare will be at most ${(3nx+x)}/{2}$.
Hence, the inapproximability factor is ${4n}/{(3n+1)}$ for any fixed $n>2$.       
\end{proof}
    
\paragraph{Larger constant number of agents.}
Next, we present the inapproximability result for larger constant number of agents.
We show the ratio becomes larger as the number of agents increases.
The reduction is still from the partition problem, but we use a different gadget to construct the instance.

Below, we first provide an overview of the reduction.
Let $E=\{e_1, \ldots, e_\ell\}$ be a scaled partition instance such that $\sum_{i=1}^\ell e_i = 2\epsilon$ with $\epsilon$ as a relatively small number.
We suggest the reader interpret the partition of these numbers as the fair allocation of $\ell$ items to two agents $a_1, a_2$ -- each has a utility of $e_i$ for item $i$.
In addition to the two agents and these items, we introduce a super agent $s$ and construct two larger items $g_1^L$, $g_2^L$, and two medium items $g_1^M$, $g_2^M$.
Both agents $a_1$ and $a_2$ have a utility of $x$ for either of the two large items and a utility of $x-\epsilon$ for either of the two medium items, where $x$ is sufficiently larger than $2\epsilon$.
Meanwhile, the super agent has a utility of $y$ with $y>x$ for either of the two large items and a utility of $0$ for other items.
When the input partition instance is a \YES instance, we can equally divide the partition instance into two subsets with equal subset-sum.
Then every agent $a_1$ and $a_2$ is allocated one of the subsets, together with a medium item, while the two large items are allocated to the super agent $s$. 
The allocation satisfies EF1 and achieves social welfare of $2y + 2x$.
However, when the partition instance is a \NO instance, the super agent cannot receive both the two large items $g_1^L$ and $g_2^L$ in an EF1 allocation as the utility of one of $a_1$ and $a_2$ must be less than $\epsilon + x-\epsilon = x$. 
Thus, when $y \gg x$, this will lead to an inapproximability of about $2$ since the super agent can receive at most one large item in the \NO instance.

Notice that the above reduction is built upon the assumption that $y$ can be sufficiently larger than $x$, which no longer holds when the valuation functions are normalized. In such a case, $y$ and $x$ will be approximately equal.
To resolve this issue, we would like to create replicas of the above partition gadgets as follows, which also leads to a stronger inapproximability ratio. 

In the above gadget, the normal agent’s utilities are distributed across large items and medium items, while the super agent’s utilities only go to the two large items $g_1^L$ and $g_2^L$.
To make sure that $y$ can still be sufficiently larger than $x$ under the assumption of normalization, it is desirable that the number of large items can be much smaller than the number of medium items.
To achieve this, we design the relationship between super agents based on clique graph structures.
We create a set of $t$ identical cliques, in each of which there are $k$ vertices.
Clique $i$ corresponds to one unique super agent $s_i$, and every edge is associated with a set of items corresponding to the given partition instance and a pair of normal agents.
There is one unique large item (also referred to as a clique item) for every vertex.
Similarly, each edge corresponds to two medium items, and the pair of normal agents associated with this edge will have a utility of $x-\epsilon$ to either of them.
We collect all these medium items inside a pool and let all the normal agents regard them the same.
To make the reduction easier to understand, we provide a graphical explanation in Fig.~\ref{fig:valuations_inmain}.

\begin{figure}[t]
\centering
\begin{tikzpicture}[
 agent/.style={regular polygon, regular polygon sides=3, draw, thick, fill=green!30, minimum size=3mm}, 
 item/.style={circle, draw, thick, fill=blue!30, minimum size=5mm}]

\node[agent, scale=1] (a1) at (-2.5, 0) [label=left:\text{clique agent}] {};
\node[agent, scale=1] (s1) at (0, 2.5) [label=left:\text{super agent}] {};
\node[agent, scale=1] (s2) at (6, 2.5)  {};
\node at (3, 2.5) {$\cdots$};

\foreach \i in {1,...,5} {
    \node[item] (n\i) at ({72*(\i-1)}:1) {};
}

\node at (3, 0) {$\cdots$};
\node at (3,-2) {$t$ \text{groups clique items}};

\foreach \i in {1,...,5} {
    \node[item] (n1\i) at ($({72*(\i-1)}:1) + (6,0)$) {};
}

\foreach \i in {1,...,3} {
    \node[item, scale=0.8] (p\i) at (-\i*0.8-2, -1.5) {};
}
\node at (-3.6, -2) {\text{partition items}};

\foreach \i in {1,...,3} {
    \node[item, scale=1] (pp\i) at (\i*0.8-2, -3) {};
}
\node at (-0.4, -3.8) {\text{pool items}};

\foreach \i in {1,...,5} {
    \foreach \j in {\i,...,5} {
        \ifnum\i<\j
            \draw[thick] (n\i) -- (n\j);
        \fi
    }
}
\foreach \i in {1,...,5} {
    \foreach \j in {\i,...,5} {
        \ifnum\i<\j
            \draw[thick] (n1\i) -- (n1\j);
        \fi
    }
}

\draw[thick, dotted, ->, shorten >=1mm, shorten <= 1mm] (a1) to node[midway, above] {$x$} (n3);
\draw[thick, dotted, ->, shorten >=1mm, shorten <= 1mm] (a1) to node[midway, below] {$x$} (n4);
\draw[thick, dotted, ->, shorten >=1mm, shorten <= 1mm] (a1) to node[midway, left] {$e_1$} (p3);
\draw[thick, dotted, ->, shorten >=1mm, shorten <= 1mm] (a1) to node[midway, left] {$e_\ell$} (p1);
\draw[thick, dotted, ->, shorten >=1mm, shorten <= 1mm] (a1) to node[midway, below] {$x-\epsilon$} (pp1);

\draw[thick, dotted, ->, shorten >=1mm, shorten <= 1mm] (s1) to node[midway, left] {$y$} (n3);
\draw[thick, dotted, ->, shorten >=1mm, shorten <= 1mm] (s1) to node[midway, right] {$y$} (n2);

\begin{pgfonlayer}{background}
    \node[circle, fill=gray!5, draw=none, rounded corners, inner sep=0pt, fit=(n1) (n2) (n3) (n4) (n5)] {};
\end{pgfonlayer}
\begin{pgfonlayer}{background}
    \node[circle, fill=gray!5, draw=none, rounded corners, inner sep=0pt, fit=(n11) (n12) (n13) (n14) (n15)] {};
\end{pgfonlayer}

\end{tikzpicture}
\caption{Graphical explanation of the reduction, where nodes in blue and green respectively represent the items and agents. A directed arrow $i \stackrel{w}{\rightarrow}j$ represents the value of $v_i(j) = w$.}
\label{fig:valuations_inmain}
\end{figure}

Notice that the utilities of a super agent $i$ are distributed across $k$ clique items within clique $i$.
In contrast, the utilities of every normal agent are mainly distributed across the two large items associated with the edge she belonged to and all items in the pool -- whose size is equal to the number of normal agents, $t\cdot {k \choose 2} = t(k^2-k)$, which is much larger than $k$ when $k$ is sufficiently large. 
This further leads to an inapproximability of $n^{1/5}/4.4144$ through a delicate analysis~(\Cref{thm:NPhard1/7}). 
Furthermore, we find that, by using the structure of multi-edge (i.e., there are multiple edges connecting every pair of vertices), the inapproximability ratio can be improved to $n^{1/4}/4.5$~(\Cref{thm:NPhardOneDivFour}).
Below, we provide the complete description of the reduction.

\begin{restatable}{lemma}{NPhardOneDivSeven}
\label{thm:NPhard1/7}
For any fixed $n > 2$, \msw is $\classNP$-hard to approximate to factor $n^{1/5}/{4.4144}$, even under normalized valuations.
\end{restatable}    
\begin{proof}
\begin{figure}[t]
\centering
\begin{tikzpicture}[
 agent/.style={regular polygon, regular polygon sides=3, draw, thick, fill=green!30, minimum size=3mm}, 
 item/.style={circle, draw, thick, fill=blue!30, minimum size=5mm}]

\node[agent, scale=1] (a1) at (-2.5, 0) [label=left:$c_{1,3}^{(1)}$] {};
\node[agent, scale=1] (s1) at (0, 2.5) [label=left:$s_1$] {};
\node[agent, scale=1] (s2) at (5, 2.5) [label=left:$s_2$] {};

\foreach \i in {1,...,3} {
    \node[item] (n\i) at ({120*(\i-1)}:1) {};
}
\foreach \i in {1,...,3} {
    \node[item] (m\i) at ($({120*(\i-1)}:1) + (5,0)$) {};
}

\node at (0.5, -2) {\text{clique items}};

\foreach \i in {1,...,3} {
    \node[item, scale=0.8] (p\i) at (-\i*0.8-2, -1.5) {};
}
\node at (-3.6, -2) {\text{partition items}};

\foreach \i in {1,...,12} {
    \node[item, scale=1, inner sep=0pt] (pp\i) at (\i*0.8-2, -3) {\scriptsize $q_{\i}$};
}
\node at (2.8, -3.8) {\text{pool items}};

\foreach \i in {1,...,3} {
    \foreach \j in {\i,...,3} {
        \ifnum\i<\j
            \draw[thick] (n\i) -- (n\j);
        \fi
    }
}
\foreach \i in {1,...,3} {
    \foreach \j in {\i,...,3} {
        \ifnum\i<\j
            \draw[thick] (m\i) -- (m\j);
        \fi
    }
}

\draw[thick, dotted, ->, shorten >=1mm, shorten <= 1mm] (a1) to node[midway, above] {$x$} (n3);
\draw[thick, dotted, ->, shorten >=1mm, shorten <= 1mm] (a1) to node[midway, below] {$x$} (n2);
\draw[thick, dotted, ->, shorten >=1mm, shorten <= 1mm] (a1) to node[midway, left] {$e_1$} (p3);
\draw[thick, dotted, ->, shorten >=1mm, shorten <= 1mm] (a1) to node[midway, left] {$e_\ell$} (p1);
\draw[thick, dotted, ->, shorten >=1mm, shorten <= 1mm] (a1) to node[midway, below] {$x-\epsilon$} (pp1);

\draw[thick, dotted, ->, shorten >=1mm, shorten <= 1mm] (s1) to node[midway, left] {$\frac13$} (n2);
\draw[thick, dotted, ->, shorten >=1mm, shorten <= 1mm] (s1) to node[midway, right] {$\frac13$} (n1);
\draw[thick, dotted, ->, shorten >=1mm, shorten <= 1mm] (s1) to node[midway, right] {$\frac13$} (n3);

\draw[thick, dotted, ->, shorten >=1mm, shorten <= 1mm] (s2) to node[midway, left] {$\frac13$} (m2);
\draw[thick, dotted, ->, shorten >=1mm, shorten <= 1mm] (s2) to node[midway, right] {$\frac13$} (m1);
\draw[thick, dotted, ->, shorten >=1mm, shorten <= 1mm] (s2) to node[midway, right] {$\frac13$} (m3);

\begin{pgfonlayer}{background}
    \node[circle, fill=gray!5, draw=none, rounded corners, inner sep=0pt, fit=(n1) (n2) (n3)] {};
\end{pgfonlayer}

\begin{pgfonlayer}{background}
    \node[circle, fill=gray!5, draw=none, rounded corners, inner sep=0pt, fit=(m1) (m2) (m3)] {};
\end{pgfonlayer}

\end{tikzpicture}
\caption{Agents and items in our construction when $n=14$, $t=2$ and $k=3$ in \Cref{thm:NPhard1/7}.
Nodes in blue and green represent items and agents respectively.
There are two cliques, and each has three corresponding clique items.
Each edge corresponds to two normal agents and a set of partition items. 
}
\label{fig:n=14}
\end{figure}

Fix a partition instance $S = \{e_1, \ldots, e_\ell \}$ such that $\sum_{i=1}^\ell e_i = 2\epsilon \in \mathbb{R}^+$, where $\epsilon$ is a relatively small number (with the partition instance rescaled). 
We construct a fair division instance as follows.
Let $t = \lceil n^{\frac{1}{5}} \rceil$, $k =\lfloor \frac{1}{2} + \frac{1}{2} \sqrt{\frac{4n}{t} -3}\ \rfloor$, and $x = \frac{1+ (t(k^2-k)-2)\epsilon}{t(k^2-k)}$. 
To help the reader better understand the reduction, we give a graphical illustration of $n=14$ in Fig.~\ref{fig:n=14}.
We construct four types of items, including
\begin{itemize}[leftmargin=0.5cm]
    \item Clique item: there are $t$ groups of clique items (the $i$-th group is named $\mathcal{C}_i$). For each group, consider a clique $G= (V, E)$ with $k$ vertices, each vertex in the clique corresponds to a clique item. The items in group $\mathcal{C}_i$ are named $c_1^{(i)},\ldots, c_k^{(i)}$. As shown in Fig.~\ref{fig:n=14}, when $n=14$, there are $t=2$ clique item groups, each of which contains $k=3$ items.   
    \item Partition item: there are $t(k^2-k)/2$ groups of partition items, which are denoted by $\mathcal{P}^{(i),\{u,v\}}, 1\le i\le t, 1\le u< v\le k$.
    Each group corresponds to an edge ($\mathcal{P}^{(i),\{u,v\}}$ corresponds to the edge $(u,v)$ in the $i$-th clique) and contains $\ell$ items, and those $\ell$ items, named as $p_1^{(i),\{u,v\}},\ldots,p_\ell^{(i),\{u,v\}}$, correspond to the $\ell$ numbers in the partition instance. 
    \item Pool item: there are $t(k^2-k)$ pool items, denoted by $\mathcal{Q} = \{q_i$: $1\le i \le t(k^2-k)\}$. In Fig.~\ref{fig:n=14}, there are $12$ pool items, denoted by $q_1,\ldots,q_{12}$.  
    \item Dummy item: there are $n-t(k^2-k+1)$ dummy items. In the example with $n=14$, since $14-2(3^2-3+1) = 0$, there are no dummy items in Fig.~\ref{fig:n=14}. 
\end{itemize}

Moreover, there are three types of agents, including
\begin{itemize}[leftmargin=0.5cm]
    \item Super agent: there are $t$ super agents named $s_1,\ldots, s_t$ corresponding to the $t$ cliques.
    \item Normal agent: there are also $t$ groups of normal agents corresponding to the $t$ cliques, and each group contains $k^2-k$ normal agents. For the $i$-th group, each edge $(u,v)$ in the clique corresponds to two normal agents $a_{(u,v)}^{(i)}$ and $a_{(v,u)}^{(i)}$. 
    In the example of $n=14$, since there are $6$ edges for the two cliques in Fig.~\ref{fig:n=14}, they correspond to $12$ normal agents in total. 
    \item Dummy agent: there are $n-t(k^2-k+1)$ dummy agents. For the example with $n=14$, since $14-2(3^2-3+1) = 0$, there is no dummy agent in Fig.~\ref{fig:n=14}.  
\end{itemize}
\par 

When $n$ is generalized, we define the valuations as follows: \textbf{1)} Each super agent $s_i$ has value $\frac1k$ for each item in $\mathcal{C}_i$ and value $0$ for other items.
For example, in Fig.~\ref{fig:n=14}, super agent $s_1$ has positive utilities only to items in $\mathcal{C}_1$;
\textbf{2)} For each $j=1,\ldots, t$ and each of the two normal agents $a_{(u,v)}^{(j)}$ and $a_{(v,u)}^{(j)}$, they both have utilities of $x$ to the clique items $c_u^{(j)}$ and $c_v^{(j)}$. 
Moreover, they both have value $e_w$ to the partition item $p_w^{(j),\{u,v\}}$ for each $w=1,\ldots,\ell$ and zero utility to other partition items. 
Other than these items, for each normal agent, (s)he has utility $x-\epsilon$ to each pool item. As shown in Fig.~\ref{fig:n=14}, each normal agent has the same utility $x-\epsilon$ to those pool items $q_1,\ldots,q_{12}$. 
\textbf{3)} For each dummy agent, (s)he values $\frac{1}{n-t(k^2-k+1)}$ to all dummy items (if existing).

If the partition instance is a \YES instance, consider the following allocation:
For each $j=1,\ldots, t$, we allocate super agent $s_j$ all the $k$ clique items in group $\mathcal{C}_j$. For each normal agent, we allocate one pool item with value $x-\epsilon$ and a set of partition items with total value $\epsilon$ (the partition instance is a yes-instance).
Each dummy agent receives exactly one dummy item.
Hence, each normal agent receives $x$ in total. It is not hard to verify that this allocation is EF1. 
In particular, a normal agent will not EF1-envy the super agent in the same group because, according to her utility function, the super agent receives two items with value $x$ and all the other items have value 0.
In this allocation, the social welfare is given by 
\begin{gather*}
    \mathcal{SW}(\mathcal{A}) \geq \underbrace{t\times 1}_{\text{clique\ items}} + \underbrace{t(k^2-k)\times ( x - \epsilon)}_{\text{pool\ items}}+\underbrace{t(k^2-k)\times \epsilon}_{\rm partition\ items} =  t + t(k^2-k)x.  
\end{gather*}

If the partition instance is a \NO instance, we then give an upper bound on social welfare. 
Consider an EF1 allocation $\A'$. 
We first estimate the total number of clique items received by all super agents.
\begin{proposition}\label{prop:ub_of_sum_of_li}
Suppose $s_j$ takes $L_j$ items in $\mathcal{C}_j$ for each $j\in[t]$. 
Then $\sum_{i=1}^t L_i \le t\sqrt{2k}$.
\end{proposition}
\begin{proof}
For each $i$, since each item corresponds to a unique vertex and there are $({L_i^2-L_i})/{2}$ edges in the clique induced by the $L_i$ vertices, then there are $L_i^2-L_i$ normal agents in group $j$ who can not receive any value from $\mathcal{C}_j$. 
Since the partition instance is a \NO instance, half of these normal agents can only receive less than $\epsilon$ value from the partition items. 
To ensure they do not EF1-envy $s_j$, each of these normal agents should receive at least two pool items.
Thus, the number of normal agents receiving at least two pool items is at least
$\sum_{i=1}^t \frac{L_i^2 - L_i}{2}$.
Hence, since the number of normal agents is equal to the number of pool items, at least $\sum_{i=1}^t \frac{L_i^2 - L_i}{2}$ normal agents will not receive any pool items.
To ensure they do not envy other normal agents, each of them should receive at least one clique item.
As there are $t\cdot k$ clique items in total, we have
\begin{align}\label{eqn:agents_not_receiving_pool_items}
\sum_{i=1}^t L_i + \sum_{i=1}^t \frac{L_i^2 - L_i}{2}\le t\cdot k \implies \sum_{i=1}^t L_i^2  < 2t\cdot k
\end{align}
Therefore, by the Cauchy-Schwarz (C-S) Inequality, the sum of $L_i$ can be bounded by $t\sqrt{2k}$, which concludes the proposition.
\end{proof}

Observe that the social welfare attained on the clique items is upper-bounded by the sum of $1/k\cdot \sum_{i=1}^t L_i$ and $tk\cdot x$.
Meanwhile, as there are $t(k^2-k)$ pool items, the total value of the pool items is upper-bounded by $t(k^2-k)\cdot (x-\epsilon)$.
Therefore, the social welfare of allocation $\A'$ is less than
\begin{align*}
    \mathcal{SW}(\A') &< \sum_{i=1}^tL_i\cdot \frac1k + tk\cdot x + t(k^2-k)\cdot (x-\epsilon) + t(k^2-k)\cdot \epsilon + \frac{n-t(k^2-k+1)}{n-t(k^2-k+1)}  \\
    &= \frac1k\sum_{i=1}^tL_i + tkx + t(k^2-k)x+1  \\
    &\le \frac1k\sum_{i=1}^tL_i + 1  + 1 + 1  \tag{$x$ is defined as $ \frac{1+(t(k^2-k)-2)\epsilon}{t(k^2-k)}$, and $\epsilon$ is small} \\ 
    &\le \frac{1}{k}\cdot t \sqrt{2k} + 3 \tag{By \Cref{prop:ub_of_sum_of_li}}
\end{align*}
For the last second inequality, it is straightforward to verify $tkx\leq1$ and $t(k^2-k)x\le 1$ by definitions of $t$, $k$, and $x$.
Finally, the inapproximability ratio is no less than
\begin{equation}\label{eqn:lb_of_inapprox}
\frac{\SW(\A)}{\SW(\A')} \ge 
\frac{t + t(k^2-k)x}{t\cdot \sqrt{2k}/k + 3} \ge 
\frac{t}{t\sqrt{2/k} + 3} =
    \frac{1}{\sqrt{\frac2k} + \frac{3}{t}} \ge \frac{1}{\sqrt{\frac2k}  + \frac{3}{n^{\frac{1}{5}}}}. 
\end{equation}
Observe that
\begin{equation}\label{eqn:ub_of_k}
k \ge \left(\frac{1}{2} + \frac{1}{2} \sqrt{\frac{4n}{t} -3}\right) -  1 \ge \sqrt{\frac{n}t} - 1 \ge \sqrt{\frac{n}{n^{\frac15} + 1}} -  1,
\end{equation}
where the second inequality can be obtained by the fact that $\sqrt{n/t}\geq1$.
By combining Inequality~\ref{eqn:lb_of_inapprox} and Inequality~\ref{eqn:ub_of_k}, we can find $\SW(\A)/\SW(\A')\ge n^{\frac15}/4.4144$.

The lemma concludes as we have shown that the corresponding social welfare under an EF1 allocation is lower-bounded by $t+t(k^2-k)x$ in a \YES instance of partition, and is upper-bounded by $t(\sqrt{k}+1)/k+3$ in a \NO instance.
The ratio is at least $n^{\frac15}/4.4144$.
\end{proof}

Recall that the main idea behind the above construction is to force the super agent to receive a large number of clique items when the partition instance is a \NO instance.
To make the consequence of super agent receiving clique items more intense, we introduce multiple edges, so that taking two adjacent clique item will lead to a larger number of normal agents receiving no pool items.
We can then improve the inapproximability ratio to ${n^{1/4}}/{4.5}$.
Notice that the ratio obtained below is stronger: the range of $n$ is larger, and the ratio is larger for all $n$ within the range.
The proof is deferred to \Cref{app:inapprox_const}.

For unnormalized valuations, we have the following result that \msw admits a $\Theta(\sqrt{n})$ inapproximability factor as $n$ grows (while still being a constant).

\begin{restatable}{lemma}{NPhardsqrtn}
\label{thm:NPhardsqrtn}
For any fixed constant $n\geq 2$, \msw is $\classNP$-hard to approximate to within any factor that is smaller than $\lfloor{\frac{1+\sqrt{4n-3}}{2}}\rfloor$.
\end{restatable}
\begin{proof}
We will present a reduction from the partition problem. 
When given a partition instance $S = \{e_1,\ldots, e_\ell \}$ such that $\sum_{i=1}^\ell e_i = 2\epsilon \in \mathbb{R}^+$ (the partition instance is scaled such that $\epsilon$ is sufficiently small), we construct an \msw instance with $n$ agents as follows.

The items are divided into two categories: \textit{clique items} and \textit{partition items}.
Consider a clique with $k=\lfloor{\frac{1+\sqrt{4n-3}}{2}}\rfloor$ vertices.
Each vertex in the clique corresponds to a clique item.
For partition items, the partition instance is copied for $k(k-1)/2$ times, corresponding to those $k(k-1)/2$ edges. 
For each edge $(u,v)$, construct a set of items $P^{(u,v)}=\{p_1^{(u,v)},\ldots,p_t^{(u,v)}\}$ that corresponds to the partition instance.

The agents are divided into three categories: one \textit{super agent}, $k(k-1)$ \textit{normal agents}, and $n-k(k-1)-1$ \textit{dummy agents}.
For each edge $(u,v)$ in the clique, it corresponds to two normal agents $a_{{(u,v)}_1}$ and $a_{{(u,v)}_2}$.

The valuations are defined as follows: for the super agent, she has a value of $1$ for each clique item and $0$ for each partition item. For each dummy agent, (s)he has a value of $0$ for all items. 
For each normal agent $a_{{(u,v)}_i}, i\in \{1, 2\}$, (s)he has value $\epsilon$ to the clique items $u$ and $v$, and value $e_i$ to the partition item $p_i^{(u,v)}$ for each $i\in\{1,\ldots,\ell \}$. 
(S)he has a value of $0$ on each of the remaining items.

If the partition instance is a \YES instance, it is possible to allocate the partition items to the normal agents such that each normal agent receives a value exactly $\epsilon$.
We will then allocate all the clique items to the super agent and let each dummy agent receive the empty set.
It is easy to see that the allocation is EF1.
In particular, a normal agent will not EF1-envy the super agent because, according to her utility function, the super agent receives two items with value $\epsilon$ and all the other items have value 0. 
In this case, the social welfare is at least $k$ by only accounting for the super agent's utility.

If the partition instance is a \NO instance, the super agent can receive only one clique item. 
Otherwise, if the super agent receives both item $u$ and item $v$, then one of the two normal agents in $a_{(u,v)_1}, a_{(u,v)_2}$ will EF1-envy the super agent: both agents think the super agent receives two items with value $\epsilon$, yet, at least one of them can only get a value less than $\epsilon$ due to that the partition instance is a \NO instance.
In this case, the social welfare is at most $1+2\epsilon\cdot(k+k(k-1)/2)$: the super agent receives at most one clique item with value $1$, the contribution to the social welfare is at most $\epsilon$ for each clique item that is not given to the super agent, and the contribution to the social welfare is at most $2\epsilon$ for each set of items $P^{(u,v)}=\{p_1^{(u,v)},\ldots,p_t^{(u,v)}\}$.
Notice that this social welfare can be made arbitrarily close to $1$ by making $\epsilon$ sufficiently small.
Hence, the inapproximability factor is $\frac{1}{k}$, and the lemma concludes by noticing $k=\lfloor{\frac{1+\sqrt{4n-3}}{2}}\rfloor$.
\end{proof}

\subsection{Approximability for General Number of Agents}
\label{sect:generaln}
When the number of agents and the number of items are given as input (not a fixed number), we demonstrate that \msw is $\classNP$-hard to approximate within a factor polynomial in $n$ or $m$.
In particular, for any given $\epsilon > 0$, \Cref{thm:NP-hard-n} shows \msw is $\classNP$-hard to approximate within a factor of $n^{1/3 -\epsilon}$ or $m^{1/2-\epsilon}$, even with normalized valuations.
The proof is deferred to~\Cref{app:approx_general_agents}, which relies on the maximum independent set problem and the inapproximability in~\Cref{thm:indset}.
On the other hand, we remark that a variant of the round-robin algorithm, \emph{the greedy-based round-robin algorithm}, can achieve an approximation of $O(n)$, with proof deferred to~\Cref{app:approx_general_agents}.

\section{Bi-Criteria Optimization}
\label{sec:bi-critera}
The last two sections show that a strong inapproximability holds for \mswx and \msw, even when the number of agents is fixed.
One may wonder whether it stems from the restriction imposed by the fairness constraint.
In this section, we consider the bi-criteria optimization version of \mswx and \msw, where the EFX and EF1 constraints are relaxed in the following natural way.

\begin{definition}\label{def:EFX_approx}
Given a real number $\beta\in(0,1]$, an allocation is \emph{$\beta$-approximately EFX}, if for any two agents $i$ and $j$, $v_{i}(A_{i})\geq \beta\cdot v_{i}(A_{j}\setminus\{g\})$ holds for any item $g\in A_{j}$.
\end{definition}

\begin{definition}\label{def:EF1_approx}
Given a real number $\beta\in(0,1]$, an allocation is \emph{$\beta$-approximately EF1}, if for any two agents $i$ and $j$, there exists an item $g\in A_{j}$ such that $v_{i}(A_{i})\geq \beta\cdot v_{i}(A_{j}\setminus\{g\})$.
\end{definition}

For $\alpha\geq1$ and $\beta\in(0,1]$, an algorithm is an $(\alpha,\beta)$-bi-criteria approximation algorithm for \mswx if it always outputs an allocation $\A$ such that
\begin{itemize}[leftmargin=0.5cm]
    \item $\alpha\cdot \SW(\A)\geq \SW(\A^*)$, and
    \item $\A$ is $\beta$-approximately EFX,
\end{itemize}
where $\A^\ast$ is the optimal solution to \mswx, i.e., $\A^\ast$ is an EFX allocation with the highest social welfare.

An $(\alpha,\beta)$-bi-criteria approximation algorithm for \msw is defined similarly.

The remaining part of this section is organized as follows.
In Section~\ref{sect:BCresults}, we consider constant numbers of agents, and we state our results of $(1,1-\epsilon)$-bi-criteria optimization for both \msw and \mswx.
These results are proved in Section~\ref{sect:BCef1} and Section~\ref{sect:BCefx}.
Finally, when the number of the agents is not a constant, we show in Section~\ref{sect:BChardness} that bi-criteria optimization fails assuming $\classP\neq\classNP$.

\subsection{Bi-Criteria Optimization for Constant Number of Agents}
\label{sect:BCresults}

For a constant number of agents, we will describe a $(1,1-\epsilon)$-bi-criteria algorithm for each of \msw and \mswx.
Our algorithm's running time is polynomial in terms of $m$ and $1/\epsilon$ (notice that $n$ is a constant).

\begin{restatable}{theorem}{bicriteria}
\label{thm:bicriteria}
Fix an arbitrary value of algorithm parameter $\epsilon>0$ and let the number of agents $n$ be a constant.
There exists a $(1,1-\epsilon)$-bi-criteria approximation algorithm for \msw whose running time is polynomial in terms of $m$ and $1/\epsilon$,
and a $(1,1-\epsilon)$-bi-criteria approximation algorithm for \mswx whose running time is polynomial in terms of $m$ and $1/\epsilon$.
\end{restatable}

\cite{aziz2020computing} provide a pseudo-polynomial time algorithm for \msw with a constant number of agents.
One may expect that standard rounding techniques can achieve the above-mentioned bi-criteria optimization.
However, straightforward rounding techniques fail.
It is possible that an agent's valuations to all items (except for a few items allocated to someone else; recall that EF1/EFX allows envy for up to one item) are extremely small, so a significant loss of precision occurs after rounding.
An EF1/EFX allocation in the rounded instance may be very far from being EF1/EFX in the original instance.
For example, consider an agent who has a very large value on one item $g$ and different small values on the remaining items.
The rounding may round the values of the items in $[m]\setminus \{g\}$ to $0$.
The EF1/EFX allocation may allocate $g$ to another agent.
In this case, any allocation after rounding is EF1 (and it is EFX if $g$ is the only item in another agent's bundle), but the allocation may be far from approximately EF1/EFX before rounding.

To achieve bi-criteria optimization, we will use careful individualized rounding together with some extra enumeration techniques.
The bi-criteria optimization algorithms for \msw and \mswx are mostly the same.

\subsection{Overview of Bi-Criteria Approximate Algorithms}
\label{sect:BCef1}
We will first show our high-level designs, and then give the complete algorithms (shown in Algorithm~\ref{bicriteria} and Algorithm~\ref{bicriteriaDP}). The high-level idea of our algorithm is shown as follows.

The algorithm proceeds in three steps. In the first step, we fix some items in each agent's bundle. 
In particular, for each pair of agents $i$ and $j$, we fix an item $g_{ij}$ in agent $i$'s bundle $A_i$.
This item will be the item $g$ in \Cref{def:EF1} such that, after removing it from $A_i$, agent $j$ will not envy agent $i$.
The first step will enumerate all possible sets for $\{g_{ij}\}$.
In the second step, we will apply an individualized rounding technique so that each agent's valuations to all the \emph{remaining items that are not fixed in the first step} can only take values from a set of numbers whose cardinality is polynomial in $m$ and $1/\epsilon$.
In the third step, we will use a dynamic programming method to solve our problem.
It is crucial that the first two steps cannot be swapped.
We will see the reason later.

\paragraph{Step 1: items fixing.}
For each agent $i$, we fix $X_i=\{g_{i1},g_{i2},\ldots,g_{i(i-1)},g_{i(i+1)},\ldots,g_{in}\}$ with $X_i\subseteq A_i$.
For each $g_{ij}$, it is expected that agent $j$ does not envy agent $i$ after removing $g_{ij}$.
Notice that it is possible that $g_{ij_1}=g_{ij_2}$.
We enumerate all possible sets for $\{X_1,\ldots,X_n\}$.
For each fixed $\{X_1,\ldots,X_n\}$, we proceed to Step~2 and 3 and compute the remaining part of the allocation.
Notice that the total number of the sets is at most $m^{n(n-1)}$ (each $g_{ij}$ can be one of the $m$ items), and this is a polynomial of $m$ given that $n$ is a constant.

\paragraph{Step 2: individualized rounding.}
We first fix an adjustable precision parameter $K$ that is a polynomial of $m$ and $1/\epsilon$.
Let $Y_i=\{g_{1i},\ldots,g_{(i-1)i},g_{(i+1)i},\ldots,g_{ni}\}$.
Notice that, for agent $i$, each item in $Y_i$ has been fixed to another agent.
Let $V_i=v_i([m]\setminus Y_i)$ be agent $i$'s total value for all the items that are not in $Y_i$.
Let $\tau_i=V_i/K$.
We will round down the valuations $v_i$ of agent $i$ to the items in $[m]\setminus Y_i$ such that the rounded valuations $\overline{v}_i$ can only take values from $\{0,\tau_i,2\tau_i,\ldots,K\tau_i\}$.
Notice that valuations of items in $Y_i$ are not rounded, and they can be significantly larger than $K\tau_i$.

In the next step, we will solve one problem similar to \msw on the instance after the rounding process above, where we aim to find an allocation that maximizes social welfare in the original valuation, subject to the EF1 constraint in the rounded valuation.
To show that we can obtain an optimal social welfare with a nearly EF1 allocation, we need to achieve that 1) an EF1 allocation in the original instance is still EF1 in the rounded instance (so that no ``high-quality'' allocation with high social welfare is ruled out after rounding) and 2) an EF1 allocation in the rounded instance is approximate EF1 in the original instance.
Notice that the EF1 here is in terms of the items $g_{ij}$ fixed in Step~1.

To achieve 1), we add $n$ dummy items $d_1,\ldots,d_n$ such that agent $i$ has value $m\tau_i$ on item $d_i$ and value $0$ on the remaining $n-1$ items.
Each item $d_i$ is allocated to agent $i$.
If the allocation satisfies that $v_i(A_i)\geq v_i(A_j\setminus\{g_{ji}\})$ in the original instance, we must also have $\overline{v}_i(A_i\cup\{d_i\})\geq \overline{v}_i(A_j\cup\{d_j\}\setminus\{g_{ji}\})$ in the rounded instance.
This is because $v_i(A_i)$ can only be reduced by at most $m\tau_i$ after rounding.
We can set $K$ to be significantly larger than $m$ such that the value of each dummy item is negligible.

To show 2), a crucial observation is that, in an allocation $\A=(A_1,\ldots, A_n)$ (with the $n$ dummy items added) that is EF1 in the rounded instance, we must have $\overline{v}_i(A_i)\geq \frac1nV_i$.
Notice that agent $i$ should not envy any other agent after removing the corresponding item in $Y_i$ from the other agent's bundle.
Therefore, agent $i$ should receive at least the \emph{average} value of $[m]\cup\{d_1,\ldots,d_n\}\setminus Y_i$.
Since the value of $[m]\cup\{d_1,\ldots,d_n\}\setminus Y_i$ is at least $V_i$ (notice that the value of $[m]\setminus Y_i$ is at least $V_i-m\tau_i$ after rounding), we have $\overline{v}_i(A_i)\geq\frac1nV_i$.
Suppose, when considering $\A$ in the original instance (with dummy items removed), agent $i$ envies agent $j$ even after removing $g_{ji}$ from agent $j$'s bundle.
We aim to show that the amount of envy is small compared with agent $i$'s value on $j$'s bundle (i.e., the amount of envy is an $\epsilon$ fraction of $i$'s value on $j$'s bundle).
Firstly, since the allocation is EF1 in the rounded instance, the amount of envy in the original instance is at most $2m\tau_i$ ($m\tau_i$ for the dummy item, and at most $m\tau_i$ for the loss of the precision in the rounding).
We can make this amount considerably smaller than $\frac1nV_i$ by setting $K$ large enough.
Secondly, agent $i$'s value on $j$'s bundle (with $g_{ji}$ removed) in the original instance should be at least $\frac1nV_i-2m\tau_i\approx\frac1nV_i$ (in order to make $i$ possibly envy $j$).
The amount of envy, which is at most $2m\tau_i$, is indeed very small.

We now remark that Step~1 and Step~2 cannot be swapped.
If $Y_i$ participates in the rounding, it is possible that items in $Y_i$ have much larger values to agent $i$ compared with the remaining items, such that the remaining items have value $0$ after rounding.
In this case, $V_i=v_i([m]\setminus Y_i)$ may even be rounded to $0$, and the amount of envy $m\tau_i$ can be large compared with agent $i$'s value on $j$'s bundle.

\paragraph{Step 3: dynamic programming.}
The remaining part of the algorithm is a standard dynamic program.
Let $H[\{u_{ij}\}_{i=1,\ldots,n;j=1,\ldots,n}]$ be a Boolean function which takes $n^2$ values $\{u_{ij}\}$ as inputs and outputs $\true$ if and only if there exists an allocation $\A=(A_1,\ldots,A_n)$ such that $\overline{v}_i(A_j\setminus\{g_{ji}\})=u_{ij}$.
After the rounding and adding dummy items in Step~2, we can assume each $u_{ij}$ only takes values from $\{0,\tau_i,2\tau_i,\ldots,(K+m)\tau_i\}$.
In addition, the total number of possible inputs to $H[\cdot]$ is $(K+m+1)^{n^2}$, which is a polynomial in $m$ and $1/\epsilon$ (as $K$ is a polynomial in $m$ and $1/\epsilon$, and $n$ is a constant).
We can use a standard dynamic program to evaluate $H[\cdot]$ for all inputs.
We can check each of the corresponding allocations to see if it satisfies EF1, and find out the one with maximum social welfare.
In addition, we can get the exact optimal social welfare instead of the $(1-\epsilon)$-approximation by comparing the social welfare of the stored allocations in terms of their actual values (instead of the rounded values).

\subsection{Bi-Criteria Optimization for EF1}

\begin{algorithm}[H]
\caption{Bi-criteria optimization of \msw}
\label{bicriteria}
\KwInput{utility functions $v_1,\ldots,v_n$, item set $M=[m]$, and the parameter $\epsilon>0$}
\KwOutput{an $(1-\epsilon)$-approximate EF1 allocation.} 
Set $K\leftarrow\lceil \frac{3mn}{\epsilon} \rceil$\;
Initialize $\Pi\leftarrow\varnothing$\;
\tcc{ $\Pi$ stores candidate allocations} 
\For{each $X=\{X_i=\{g_{i1},\ldots,g_{i(i-1)},g_{i(i+1)},\ldots,g_{in}\}\mid i=1,\ldots,n\}$}{
\tcc{elements in $X_i$ may be repeated, but $X_i\cap X_j=\emptyset$ for any $i,j$} \label{bicriteria-enumerate}
\For{each $i=1,\ldots,n$}{
Set $Y_i\leftarrow\{g_{1i},\ldots,g_{(i-1)i},g_{(i+1)i},\ldots,g_{ni}\}$\;
Set $V_i\leftarrow v_i([m]\setminus Y_i)$\;\label{bicriteria-step2b}
Set $\tau_i\leftarrow V_i/K$\;
for each $o\in [m]\setminus Y_i$, set $\overline{v}_i(o)\leftarrow\max_{k:k\tau_i\leq v_i(o)}k\tau_i$\;\label{bicriteria-step2e}
for each $o\in Y_i$, set $\overline{v}_i(o)\leftarrow 0$\;\label{bicriteria-conve}
Add a new item $d_i$ to $M$ such that $\overline{v}_i(d_i)=m\tau_i$ and $\overline{v}_j(d_i)=0$ for each $j\neq i$\ and $v_j(d_i)=0$ for each $j\in[n]$\; \label{bicriteria-dummyb}
$X_i\leftarrow X_i\cup d_i$\;\label{bicriteria-dummye}
}
$\A\leftarrow${\sc DynamicProgram}($M,\overline{v}_1,\ldots,\overline{v}_n,v_1,\ldots,v_n,\tau_1,\ldots,\tau_n,X,K$)\tcp*{see Algorithm~\ref{bicriteriaDP}} \label{bicriteria-retb}
Include $\A$ in $\Pi$\;
}
\Return{the allocation in $\Pi$ with the largest social welfare with respect to $\{v_1,\ldots,v_n\}$} \label{bicriteria-rete}
\end{algorithm}
\begin{algorithm}[H]
\caption{The dynamic programming subroutine for \msw}
\label{bicriteriaDP}
\SetKwProg{Fn}{Function}{:}{}
\Fn{\FDynamicProgram{$M=[m+n],\overline{v}_1,\ldots,\overline{v}_n,v_1,\ldots,v_n,\tau_1,\ldots,\tau_n,X,K$}}{
\tcc{$X=\{X_i=\{g_{i1},\ldots,g_{i(i-1)},g_{i(i+1)},\ldots,g_{in},d_i\}\mid i=1,\ldots,n\}$}
\tcc{for each $i$, dummy item $d_i$ is the $(m+i)$-th item }
\tcc{for each $i$, $\overline{v}_i(S)/\tau_i\in \{0,1,\ldots,K+m\}$ holds for any $S\subseteq M$}
\For{each $\chi\in\{0,1,\ldots,K+m\}^{n^2}$ and each $t=0,1,\ldots,m+n$}{ \label{bicriteriaDP-initb}
Initialize $H[\chi,t]\leftarrow\nil$\;
}
$H[0^{n^2},0]\leftarrow 0$\;\label{bicriteriaDP-inite}
\For{each $t=0,\ldots,m+n-1$}{ \label{bicriteriaDP-transferb}
\For{each $\chi$ in the dictionary ascending order such that \emph{$H[\chi,t]\neq \nil$}}{
\If{item $t+1$ belongs to $X$}{  \label{bicriteriaDP-transferoneb}
Suppose  $t\in X_{i^\ast}$\;
{\sc Update}($\chi,t,i^\ast$)\; \label{bicriteriaDP-transferonee}
}
\Else{ \label{bicriteriaDP-transfertwob}
\For{each $i=1,\ldots,n$}{
{\sc Update}($\chi,t,i$)\; \label{bicriteriaDP-transfere}
}
}
}
}
\For{each $\chi\in\{0,1,\ldots,K+m\}^{n^2}$}{ \label{bicriteriaDP-retb}
    Set $H[\chi,m+n]\leftarrow\nil$ if the allocation stored in $H[\chi,m+n]$ is not envy-free w.r.t. $\overline{v}_1,\ldots,\overline{v}_n$\;
}
\If{\emph{$H[\chi,m+n]\neq\nil$} for some $\chi$}{
    \Return{the allocation in \emph{$\{H[\chi,m+n]\mid H[\chi,m+n]\neq\nil,\chi\in\{0,1,\ldots,K+m\}^{n^2}\}$} with the largest social welfare w.r.t. $v_1,\ldots,v_n$, i.e. the one with the largest value $H[\chi,m+n]$} %
}
\Else{
    \Return{\emph{$\nil$}} \label{bicriteriaDP-rete}
}
}
\Fn{\Update{$\chi,t,i$}\label{bicriteriaDP-upb}}
{
            for each $j$, set $\chi_{ji}'\leftarrow\chi_{ji}+\overline{v}_j(t+1)/\tau_j$\; 
            for each $(i',j)$ with $i'\neq i$, set $\chi_{ji'}'\leftarrow\chi_{ji}$\;
            \If{\emph{$H[\chi',t+1]==\nil$} or $H[\chi',t+1]<H[\chi,t]+v_i(t+1)$}
            {
            $H[\chi',t+1]\leftarrow H[\chi,t]+v_i(t+1)$\; \label{bicriteriaDP-upe}
            }
}
\end{algorithm}
    
In Algorithm~\ref{bicriteria}, we set the precision parameter $K$ as $\lceil \frac{3mn}{\epsilon} \rceil$ which is a polynomial of $m$ and $1/\epsilon$. We enumerate all possible removing item set $X=\{X_i=\{g_{i1},\ldots,g_{i(i-1)},g_{i(i+1)},\ldots,g_{in}\}\mid i=1,\ldots,n\}$ in Line \ref{bicriteria-enumerate}, where each $g_{ij}$ is expected that agent $j$ will not envy agent $i$ after removing $g_{ij}$ in agent $i$'s bundle (\textbf{Step 1: items fixing}). Lines \ref{bicriteria-step2b}-\ref{bicriteria-step2e} perform the \textbf{individualized rounding} step and set the rounded valuation $\overline{v}_i$ such that the valuations can only take values from $\{0,\tau_i,2\tau_i,\ldots,K\tau_i\}$. For convenience, for each agent $i$, Line \ref{bicriteria-conve} sets \textbf{the valuation of all items in $Y_i$ as $0$} so that we can turn the original EF1 condition with regard to $Y_i$ to the envy-freeness condition. 

Lines \ref{bicriteria-dummyb}-\ref{bicriteria-dummye} add the dummy items $\{d_1,\ldots,d_n\}$ to reach the goal that if the allocation satisfies that $v_i(A_i)\geq v_i(A_j\backslash \{g_{ji}\})$ in the original instance, we must also have $\overline{v}_i(A_i\cup \{d_i\})\geq \overline{v}_i(A_j\cup \{d_j\})$, and update the fixed item set $X_i$ for each agent $i$. Finally, Lines \ref{bicriteria-retb}-\ref{bicriteria-rete} call Algorithm~\ref{bicriteriaDP} to find the optimal envy-free allocation after fixing set $X$ and return the one with the largest social welfare with respect to original valuation among all possible sets $X$.

For Algorithm~\ref{bicriteriaDP}, it uses the dynamic program to calculate the optimal envy-free allocation after fixing the assigned item set $X=\{X_i=\{g_{i1},\ldots,g_{i(i-1)},g_{i(i+1)},\ldots,g_{in},d_i\}\mid i=1,\ldots,n\}$, where all items in $X_i$ should be assigned to agent $i$ (\textbf{Step 3: dynamic programming}). Since we add the dummy items, for any subset $S\subseteq M$, we have $\overline{v}_i(S)/\tau_i\in\{0,1,\ldots,K+m\}$ for each $i$. In this algorithm, we use the state $H[\chi,t]$ to store the maximum social welfare value with regard to the original valuation (and record the corresponding allocation simultaneously), where $\chi_{ij}$ represents the value of agent $j$'s bundle with regard to $\overline{v}_i$ when allocating only the first $t$ items.

Lines \ref{bicriteriaDP-initb}-\ref{bicriteriaDP-inite} initialize the values $H[\chi,t]$ for all possible $\chi$ and $t$. Lines \ref{bicriteriaDP-transferb}-\ref{bicriteriaDP-transfere} transfer the state, where Lines \ref{bicriteriaDP-transferoneb}-\ref{bicriteriaDP-transferonee} is for the case where item $t+1$ has been assigned to agent $i$ before and Lines \ref{bicriteriaDP-transfertwob}-\ref{bicriteriaDP-transfere} is for the case where item $t+1$ has not been assigned and we need to enumerate the allocated agent. Next, Lines \ref{bicriteriaDP-retb}-\ref{bicriteriaDP-rete} return the envy-free allocation with the largest social welfare with regard to the original valuation. Lines \ref{bicriteriaDP-upb}-\ref{bicriteriaDP-upe} describe the detailed state transfer if we assign item $t+1$ to agent $i$ from the state $H[\chi,t]$.

After describing the algorithm formally, we now prove \Cref{thm:bicriteria}. The proof contains two parts: 1) an EF1 allocation in the original instance can be transferred to an EF allocation in the modified instance (after rounding and adding dummy items); 2) an EF allocation in the modified instance can be transferred to an $(1-\epsilon)$-approximate EF1 allocation in the original instance.
We show the two points by the following two lemmas.

\begin{lemma}
\label{lem:oritoround}
For an EF1 allocation $\A=\{A_1,\ldots,A_n\}$ with regard to $\{v_1,\ldots,v_n\}$ in the original instance, the corresponding allocation $\A'=\{A'_1,\ldots,A'_n\}$ where $A'_i=A_i\cup\{d_i\}$ for each $i$, is an EF allocation with regard to $\{\overline{v}_1,\ldots,\overline{v}_n\}$ after fixing some set $X=\{X_i=\{g_{i1},\ldots,g_{i(i-1)},g_{i(i+1)},\ldots,g_{in},d_i\}\mid i=1,\ldots,n\}$ in the modified instance.
\end{lemma}

\begin{proof}
For an EF1 allocation $A=\{A_1,\ldots,A_n\}$ with regard to $\{v_1,\ldots,v_n\}$, for each $(i,j)$ where $i\neq j$, there exists one $g_{ji}\in A_j$ such that $v_i(A_i)\geq v_i(A_j\backslash\{g_{ji}\})$ by the definition of EF1. We use all these $g_{ji}$ to get the set $X$. We want to show $\overline{v}_i(A_i\cup\{d_i\})\geq \overline{v}_i(A_j\cup\{d_j\})$ for each $(i,j)$ where $i\neq j$.

We consider each $(i,j)$ where $i\neq j$.  We have $\overline{v}_i(A_i\cup\{d_i\})=\overline{v}_i(A_i)+m\tau_i\geq v_i(A_i)-m\tau_i+m\tau_i=v_i(A_i)$ because $A_i$ contains at most $m$ items and the value of each is rounded down by at most $\tau_i$. 
On the other hand, since $g_{ji}\in Y_i$, we have $\overline{v}_i(g_{ji})=0$ by Line~\ref{bicriteria-conve} of Algorithm~\ref{bicriteria}.
By EF1, $v_i(A_i)$ is at least $v_i(A_j\backslash\{g_{ji}\})\geq \overline{v}_i(A_j\backslash\{g_{ji}\})=\overline{v}_i(A_j)=\overline{v}_i(A_j\cup\{d_j\})$.
Therefore, $\overline{v}_i(A_i\cup\{d_i\})\geq\overline{v}_i(A_j\cup\{d_j\})$.
\end{proof}

\begin{lemma}
\label{lem:roundtoori}
For an EF allocation $\A'=\{A'_1,\ldots,A'_n\}$ where $A'_i=A_i\cup\{d_i\}$ for each $i$ with regard to $\{\overline{v}_1,\ldots,\overline{v}_n\}$ after fixing the set $X=\{X_i=\{g_{i1},\ldots,g_{i(i-1)},g_{i(i+1)},\ldots,g_{in},d_i\}\mid i=1,\ldots,n\}$ in the modified instance, the corresponding allocation $\A=\{A_1,\ldots,A_n\}$ is an $(1-\epsilon)$-approximate EF1 allocation with regard to $\{v_1,\ldots,v_n\}$ in the original instance.
\end{lemma}

\begin{proof}
By definition, for each pair $(i,j)$ where $i\neq j$,  we have $\overline{v}_i(A_i\cup\{d_i\})\geq \overline{v}_i(A_j\cup\{d_j\})$, then it suffices to show $v_i(A_i)\geq (1-\epsilon)v_i(A_j\backslash\{g_{ji}\})$.

We consider each pair $(i,j)$ where $i\neq j$ such that agent $i$ envies agent $j$ even after removing item $g_{ji}$ from $j$'s bundle: $v_i(A_i)< v_i(A_j\backslash\{g_{ji}\})$ (otherwise, if no such pair exists, the allocation is EF1 and we are done).
We have
\begin{align*}
    v_i(A_i)&\geq \overline{v}_i(A_i)\tag{values of items are rounded down}\\
    &=\overline{v}_i(A_i\cup\{d_i\})-m\tau_i\tag{$\overline{v}_i(d_i)=m\tau_i$}\\
    &\geq \overline{v}_i(A_j\cup\{d_j\})-m\tau_i\tag{$\A'$ is envy-free}\\
    &=\overline{v}_i(A_j)-m\tau_i\tag{$\overline{v}_i(d_j)=0$}\\
    &=\overline{v}_i(A_j\backslash\{g_{ji}\})-m\tau_i\tag{$\overline{v}_i(g_{ji})=0$ from $g_{ji}\in Y_i$ (Line~\ref{bicriteria-conve} of Algorithm~\ref{bicriteria})}\\
    &\geq v_i(A_j\backslash\{g_{ji}\})-2m\tau_i,
\end{align*}
where the last inequality is because $A_j\backslash\{g_{ji}\}$ contains at most $m$ items and the value of each is rounded down by at most $\tau_i$.
Thus, if we can show $2m\tau_i\leq \epsilon v_i(A_j\backslash\{g_{ji}\})$, then we have $v_i(A_i)\geq (1-\epsilon)v_i(A_j\backslash\{g_{ji}\})$, which finishes the proof.

So the remaining is to show $2m\tau_i\leq \epsilon v_i(A_j\backslash\{g_{ji}\})$. Since we have $\overline{v}_i(A_i)+m\tau_i=\overline{v}_i(A_i\cup\{d_i\})\geq \overline{v}_i(A_{j'}\cup\{d_{j'}\})=\overline{v}_i(A_{j'}\backslash\{g_{j'i}\})$ for each $j'\neq i$, and we also have $\overline{v}_i(A_i)+m\tau_i\geq \overline{v}_i(A_i)+m\tau_i$, we sum all these $n$ terms up and we have
\begin{equation*}
    \overline{v}_i(A_i)+m\tau_i\geq\frac1n\left(\overline{v}_i(A_i)+m\tau_i+\sum_{j'\neq i}\overline{v}_i(A_{j'}\backslash\{g_{j'i}\})\right)=\frac{1}{n}(\overline{v}_i([m]\backslash Y_i)+m\tau_i)\geq \frac{1}{n}v_i([m]\backslash Y_i)
\end{equation*}
where the last inequality is due to that the value of each item is rounded down by at most $\tau_i$ and there are at most $m$ items. 
Recalling that we have set $V_i=v_i([m]\setminus Y_i)$,
this means $\overline{v}_i(A_i)\geq \frac{V_i}{n}-m\tau_i$.

Next, 
\begin{align*}
    \epsilon v_i(A_j\backslash\{g_{ji}\})&>\epsilon v_i(A_i)\tag{we have assumed $v_i(A_i)<v_j(A_j\setminus\{g_{ji}\})$ at the beginning}\\
    &\geq \epsilon \overline{v}_i(A_i)\\
    &\geq \epsilon\left(\frac{V_i}{n}-m\tau_i\right)\tag{we have just proved this}\\
    &\geq \epsilon \frac{V_i}{n}-m\tau_i.
\end{align*}
To show $2m\tau_i\leq \epsilon v_i(A_j\backslash\{g_{ji}\})$, it remains to show $\epsilon \frac{V_i}{n}-m\tau_i\geq 2m\tau_i$, that is $\tau_i\leq \frac{\epsilon}{3mn}V_i$. Since $K=\lceil \frac{3mn}{\epsilon} \rceil \geq \frac{3mn}{\epsilon}$, we have $\tau_i=V_i/K\leq \frac{\epsilon}{3mn}V_i$, and the lemma concludes.
\end{proof}

\begin{proof}[Proof of the first part of \Cref{thm:bicriteria}.]
Assume the allocation $\A^\ast=(A_1,\ldots,A_n)$ is the EF1 allocation with the largest social welfare, from \Cref{lem:oritoround}, it can be transferred to one EF allocation $\A'=\{A'_1,\ldots,A'_n\}$ where $A'_i=A_i\cup\{d_i\}$ for each $i$ after fixing some set $X$. Then, since we choose the allocation with the largest social welfare with regard to original valuation for all possible sets $X$, the allocation output by Algorithm \ref{bicriteria} must have at least the same social welfare with regard to original valuation after removing the dummy items. Furthermore, because of \Cref{lem:roundtoori}, this allocation is an $(1-\epsilon)$-approximate EF1 allocation with regard to the original valuation. To sum up, the output allocation $\A$ should be $(1-\epsilon)$-approximate EF1 and $\SW(\A)\geq \SW(\A^\ast)$.
\end{proof}

\subsection{Bi-Criteria Optimization for EFX}
\label{sect:BCefx}

We can follow the same ideas as above.
Thus, we only discuss the differences in the algorithms. Our algorithm is shown in Algorithm~\ref{bicriteriax} and Algorithm~\ref{bicriteriaxDP}.

The main difference is that the $g_{ij}$ we enumerate now should ensure such $g_{ij}$ is the smallest item in agent $i$'s bundle from agent $j$'s perspective.
To ensure this, at Step~\ref{bicriteriaxcond} in Algorithm~\ref{bicriteriax}, we only need the feasible choices of the set $X$ such that $v_i(g_{ji})\leq v_i(g_{jk})$, for all $i,j,k$. This means $g_{ji}$ is the smallest item among all items that have been allocated to agent $j$, from agent $i$'s perspective.

The second difference is at Step~\ref{bicriteriaxDP-newif} in Algorithm~\ref{bicriteriaxDP}.
Here, when we allocate the items not in the set $X$ to some agent $i$, we also need to ensure that the allocated item cannot have a smaller value than any $g_{ij}$ from agent $j$'s perspective.

\addtocontents{toc}{\protect\setcounter{tocdepth}{1}}
\subsection{Intractability Results for Bi-Criteria Optimization}
\label{sect:BChardness}
\addtocontents{toc}{\protect\setcounter{tocdepth}{2}}
In the previous sections, for a constant number of agents, we show that \msw and \mswx become tractable if we slightly relax EF1 and EFX.
In this section, we show in Theroem~\ref{thm:hardness-bicriteria-overall} that the problem becomes largely intractable for a general number of agents even if EF1 and EFX are relaxed substantially, whose proof is broken down into \Cref{thm:hardness-bicriteria-n} and \Cref{thm:hardness-bicriteria-n-efx} below.

\begin{theorem}\label{thm:hardness-bicriteria-overall}
Fix any small $\epsilon >0$.
It implies $\classP=\classNP$ if there exists
\begin{itemize}[leftmargin=0.5cm]
    \item a polynomial-time $(n^{0.5-\epsilon},\epsilon)$-bi-criteria approximation algorithm or a polynomial-time $(m^{1-\epsilon},\epsilon)$-bi-criteria approximation algorithm for \msw, or
    \item a polynomial-time $(n^{1-\epsilon},0.5+\epsilon)$-bi-criteria approximation algorithm for \mswx.
\end{itemize}
\end{theorem}

\begin{lemma}\label{thm:hardness-bicriteria-n}
Fix small $\epsilon>0$. 
If there exists a $(n^{0.5-\epsilon},\epsilon)$-bi-criteria or a $(m^{1-\epsilon},\epsilon)$-bi-criteria polynomial-time approximation algorithm for \msw, then $\classP=\classNP$.
\end{lemma}
\begin{proof}
The proof is similar to the proof of~\citep[Appendix~A.2]{10.5555/3306127.3331927}.
We present a reduction from the maximum independent set problem in Definition~\ref{def:indset} based on \Cref{thm:indset}.

Given a maximum independent set instance $G=(V,E)$, we construct a fair division instance with $n$ items and $m+1$ agents, where $n=|V|$ and $m=|E|$.
Those $n$ items correspond to the $n$ vertices in $G$.
Those $m+1$ agents consist of a \emph{super agent} and $m$ \emph{edge agents} that correspond to the $m$ edges in $G$.

Fix a small number $\delta>0$.
The super agent has a value of $1$ for each item.
For the edge agent representing the edge $(u,v)$, (s)he has value $\delta$ on the two items representing $u$ and $v$, and (s)he has value $0$ on the remaining items.

To guarantee $\epsilon$-approximate EF1, it is required that the super agent cannot take any pair of items representing vertices $u$ and $v$ for an edge $(u,v)$.
For otherwise, the agent representing the edge $(u,v)$ will receive a value of $0$, and $\epsilon$-approximate EF1 fails to hold for any $\epsilon>0$.
As a result, the set of items taken by the super agent must correspond to an independent set in $G$.

By making $\delta$ sufficiently small, the social welfare almost exclusively depends on the super agent's utility.
\Cref{thm:hardness-bicriteria-n} follows easily from \Cref{thm:indset} and the fact $m=O(n^2)$.
\end{proof}

\begin{lemma}\label{thm:hardness-bicriteria-n-efx}
Fix any small $\epsilon>0$. 
If there exists a $(n^{1-\epsilon},0.5+\epsilon)$-bi-criteria polynomial-time approximation algorithm for \mswx, then $\classP=\classNP$.
\end{lemma}
\begin{proof}
    The same reduction in the proof of Part (1) of \Cref{thm:mswx_neg_overall} can be used here.
    The analysis of the yes-instance is exactly the same.
    The analysis of the \NO instance is almost the same except for the following differences.
    Firstly, if agent $a_{i0}$ gets both $v_{iu_1}$ and $v_{iu_2}$ for an edge $j=(u_1,u_2)$, then the fact agent $a_{i0}$ cannot get more than these two items (that we have proved in Part (1) of \Cref{thm:mswx_neg_overall}) still follows from $(0.5+\epsilon)$-EFX requirement.
    This is where the parameter $0.5+\epsilon$ in the lemma statement comes from.
    Secondly, we will use the inapproximability result of the independent set problem described in \Cref{thm:indset}.
    We need that, for a yes-instance, $G$ has an independent set of size $x$, and, for a \NO instance, $G$ has an independent set of size less than $0.5x$.
    We need a factor $2$ inapproximability result of the independent set problem, while \Cref{thm:indset} says that the inapproximability ratio is more than $2$.
\end{proof}

\section{Prices of EFX and EF1}
\label{append:priceoffairness}

Under normalized valuations, it is known that the price of EF1 is $8/7$ for $n=2$~\citep{bei2021price,LiLiLu24} and is $\Theta(\sqrt{n})$ for general $n$~\citep{bei2021price,barman2020optimal}.
For the price of EFX under normalized valuations, it is $1.5$ for two agents~\citep{bei2021price} and is unknown for general $n$.

\begin{restatable}{theorem}{thmpriceoffairness}
\label{thm:price_of_fairness}
The price of EF1 is exactly $n$ for any number of agents.
The price of EFX is $\Theta(n)$ for any number of agents, and $\Theta(\sqrt{n})$ for normalized valuations.
\end{restatable}

\begin{proof}
We first show that, under unnormalized valuations, the prices of both EF1 and EFX are large: $n$ is a trivial lower bound.
For an arbitrary $n$, consider $n$ agents and $n$ items where agent $1$ has value $1$ on all the items and each of agents $2,\ldots,n$ has value $\frac\epsilon{n-1}$ on all the items. $\MSW=n$ for the allocation where agent $1$ gets all the items. On the other hand, an EF1/EFX allocation must allocate each agent one item, which has social welfare $1+\epsilon$.

Thus, \Cref{thm:1/n} and the above argument immediately imply the price of EF1 is $n$.
Our results in Section~\ref{sec:mswxgeneraln} also make the price of EFX clear.
\Cref{thm:mswx_pos_overall} implies the price of EFX is $O(n)$, as $\sum_{i=1}^nv_i([m])$ is a trivial upper bound to $\MSW$. The initial argument then implies the first part of this theorem.
For normalized valuations, \Cref{thm:mswx_pos_overall} implies the price of EFX is $O(\sqrt{n})$.
Since EFX is a stronger notion than EF1, the lower bound $\Omega(\sqrt{n})$ for the price of EF1 by \cite{bei2021price} can be directly applied here.
\end{proof}

\section{Conclusion and Future Work}
\label{sect:conclusion}
In this work, we provided a complete landscape on the complexity and approximability of maximizing social welfare subject to the EF1/EFX constraint.
Our results also provide asymptotically tight ratios for the price of EFX, which is missing in the previous literature.

Given that this problem has been well understood for the cake-cutting problem and is now well understood for the indivisible item allocation problem, an interesting future direction is to study this problem for the mixed divisible and indivisible goods.
The allocation of mixed goods was first considered by~\cite{bei2021fair}, in which a fairness notion ``EFM'' that adapts EF is proposed.
Studying maximizing efficiency while guaranteeing fairness is thus a compelling future direction.

\section*{Acknowledgments}
The research of Shengxin Liu is supported by the National Natural Science Foundation of China (No. 62102117), by the Shenzhen Science and Technology Program (No. GXWD20231129111306002), and by the Key Laboratory of Interdisciplinary Research of Computation and Economics (Shanghai University of Finance and Economics), Ministry of Education.
The research of Biaoshuai Tao is supported by the National Natural Science Foundation of China (No. 62472271) and the Key Laboratory of Interdisciplinary Research of Computation and Economics (Shanghai University of Finance and Economics), Ministry of Education.
We are extremely grateful to Prof. Bhaskar Ray Chaudhury, Prof. Ruta Mehta, and Dr. Aniket Murhekar of the UIUC game theory group for their invaluable discussions and insightful suggestions on writing.

\bibliographystyle{alpha}
\bibliography{reference}

\appendix

\newpage
\appendix

\section{Resource Monotonicity for EFX}
\label{append:monotonicityEFX}
\efxResourceMono*
\begin{proof}
When $n=2$, the resource monotonicity for EFX allocations with two agents follows straightforwardly from \Cref{prop:monotonicity}.
Given a partial EFX allocation $(A,B)$ and an unallocated item $g$, we can compute in polynomial time an EFX allocation $(A',B')$ with $A'\cup B'=A\cup B\cup\{g\}$ such that $v_1(A')\geq v_1(A)$ and $v_2(B')\geq v_2(B)$.
To see this, if $(A\cup\{g\},B)$ is EFX or $(A,B\cup\{g\})$ is EFX, we can directly update the allocation.
Otherwise, we can apply \Cref{prop:monotonicity} to update the allocation.
This implies the lemma as the social welfare is clearly non-decreasing and the update from $(A,B)$ to $(A',B')$ can be done for at most $m$ times.

When $n>2$, we first consider the following instance with three agents and seven items, introduced by~\cite{ChaudhuryGaMe24}.

\begin{center}
    \begin{tabular}{c|ccccccc}
    \hline
    & $g_1$ & $g_2$ & $g_3$ & $g_4$ & $g_5$ & $g_6$ & $g_7$\\
    \hline
    $v_1$ & 8 & 2 & 12 & 2 & 0 & 17 & 1 \\
    $v_2$ & 5 & 0 & 9 & 4 & 10 & 0 & 3\\
    $v_3$ & 0 & 0 & 0 & 0 & 9 & 10 & 2\\
    \hline
\end{tabular}
\end{center}

\cite{ChaudhuryGaMe24} proved the following proposition.

\begin{proposition}[\cite{ChaudhuryGaMe24}]\label{prop:POmonotonicityexample}
    The allocation $(A_1,A_2,A_3)$ with
    $$A_1=\{g_2,g_3,g_4\},\qquad A_2=\{g_1,g_5\},\qquad \mbox{and}\qquad A_3=\{g_6\},$$
    is a partial EFX allocation (with $g_7$ unallocated) such that no complete EFX allocation $(B_1,B_2,B_3)$ satisfies
    $$v_1(B_1)\geq v_1(A_1)=16, \qquad v_2(B_2)\geq v_2(A_2)=15, \qquad \mbox{and}\qquad v_3(B_3)\geq v_3(A_3)=10.$$
\end{proposition}

Let $w$ be a large number.
Based on Chaudhury et al.'s instance, we construct one more agent and two more items with the new valuation profile shown below.

\begin{center}
    \begin{tabular}{c|ccccccccc}
    \hline
    & $g_1$ & $g_2$ & $g_3$ & $g_4$ & $g_5$ & $g_6$ & $g_7$ & $h_1$ & $h_2$\\
    \hline
    $v_1$ & 8 & 2 & 12 & 2 & 0 & 17 & 1 & 16 & 16\\
    $v_2$ & 5 & 0 & 9 & 4 & 10 & 0 & 3 & 15 & 15\\
    $v_3$ & 0 & 0 & 0 & 0 & 9 & 10 & 2 & 10 & 10\\
    $v_4$ & 0 & 0 & 0 & 0 & 0 & 0 & 0 & $w$ & $w$ \\
    \hline
\end{tabular}
\end{center}

If we are allowed partial allocation, with $g_7$ unallocated, an EFX allocation can be
\begin{equation}\label{eqn:partialEFX}
A_1=\{g_2,g_3,g_4\},\qquad A_2=\{g_1,g_5\},\qquad A_3=\{g_6\},\qquad \mbox{and} \qquad A_4=\{h_1,h_2\}.
\end{equation}

This allocation has social welfare, which is slightly more than $2w$.
We will show that agent $4$ can only get one of $h_1$ and $h_2$ in any complete EFX allocation, in which case the social welfare is slightly more than $w$, which is less than $2w$.

Suppose for the sake of contradiction that agent $4$ gets both $h_1$ and $h_2$.
If $A_4=\{h_1,h_2\}$, then \Cref{prop:POmonotonicityexample} suggests that one of the first three agents will get a value of less than $16$, $15$, and $10$ respectively to keep the EFX property among them.
In this case, the said agent will EFX-envy agent $4$.
If $\{h_1,h_2\}\subsetneq A_4$, then, after removing an item in $A_4\setminus\{h_1,h_2\}$, $A_4$ is worth at least $32$, $30$, and $20$ for the first three agents respectively.
Therefore, there should be an allocation of the first seven items to the first three agents such that the three agents get values of $32$, $30$, and $20$ respectively.
It is easy to see that this is impossible: $g_3$ and $g_6$ must be in agent $1$'s bundle, then agent $3$ must get $g_5$ and $g_7$; agent $2$ cannot get a value of $30$ from $\{g_1,g_2,g_4\}$.
Therefore, we have proved \Cref{thm:monotonicityEFX} with unnormalized valuations.

Finally, we show that our example also works if the valuation is normalized, where the numbers in the first three rows of the table are divided by $74$, $61$, and $41$ respectively, and $w=0.5$.
The partial EFX allocation (\ref{eqn:partialEFX}) has social welfare $\frac{16}{74}+\frac{15}{61}+\frac{10}{41}+1>1.70$.
The same arguments above show that at most one of $h_1$ and $h_2$ can be allocated to agent $4$ in any complete EFX allocation.
To find an upper bound to the social welfare of an EFX allocation, suppose $h_2$ is allocated to agent $4$ and each item is allocated to an agent with the highest value, except that $h_1$ cannot be given to agent $4$.
The upper bound is
$$\frac8{74}+\frac2{74}+\frac{12}{74}+\frac4{61}+\frac9{41}+\frac{10}{41}+\frac3{61}+\frac{15}{61}+\frac12<1.63.$$
The social welfare is strictly less than that of (\ref{eqn:partialEFX}).
\end{proof}

\section{Omitted Proofs in Section~\ref{sec:mswxgeneraln}}
\label{append:mswx}

\addtocontents{toc}{\protect\setcounter{tocdepth}{1}}
\subsection{Proof of First Part of \Cref{thm:mswx_pos_overall}}
\label{appendix:On}
\addtocontents{toc}{\protect\setcounter{tocdepth}{2}}

\thmmswxposoverall*

\begin{algorithm}[H]
    \caption{An $O(n)$-approximation algorithm for \mswx}\label{alg:On}
    \KwInput{$(v_1,\ldots,v_n)$} 
    \KwOutput{an EFX allocation (that is allowed to be partial)}
    {\bf if} $m<n$ {\bf then} Add dummy items with value $0$ to all agents so that $m=n$\;
    Initialize $\A=(A_1,\ldots,A_n)$ that maximizes $\SW(\A)$ subject to $|A_1|=\cdots=|A_n|=1$\;%
    $B\leftarrow [m]\setminus\bigcup_{i=1}^nA_i$\;
    \While{there exists an agent that envies $B$}{
        Let $i$ be the most envious agent to $B$\;
        $X_i\leftarrow${\sc Replace}$(v_i,A_i,B)$\tcp*{see Algorithm~\ref{alg:replace}}
        $B\leftarrow B\cup A_i\setminus X_i$\;
        $A_i\leftarrow X_i$\;
    }
    \Return{$(A_1,\ldots,A_n)$}
    \end{algorithm}
    
\begin{proof}
We have proved that the algorithm runs in pseudo-polynomial time and always outputs EFX allocations in \Cref{sect:On}.
It remains to show the approximation guarantee $(2n+1)\cdot\SW(\A)\geq \sum_{i=1}^nv_i([m])$.
The proof is similar to Lemma~1 in~\cite{barman2020optimal}.
The algorithm by~\cite{barman2020optimal} also starts with an allocation $\A=(A_1,\ldots,A_n)$ that maximizes $\SW(\A)$ subject to $|A_1|=\cdots=|A_n|=1$.\footnote{Barman et al.'s algorithm is used to compute an EF1 allocation, whereas ours is used for EFX. Both algorithms start by allocating each agent one item.}
The following fact is proved by~\cite[Lemma 1]{barman2020optimal}.

\begin{proposition}[\cite{barman2020optimal}]
    Let $\A'$ be the allocation after Line~2 of Algorithm~\ref{alg:On}. We have $\SW(\A')\geq\frac1n\sum_{i=1}^n\sum_{g\in G_i}v_i(g)$, where $G_i$ is the set of the $n$ items with the largest values to agent $i$.
\end{proposition}

Let $\A=(A_1,\ldots,A_n)$ be the output of Algorithm~\ref{alg:On}.
Since {\sc Replace} subroutine always increases an agent's utility, the above proposition implies 
\begin{equation}\label{eqn:On}
    \SW(\A)\geq\frac1n\sum_{i=1}^n\sum_{g\in G_i}v_i(g).
\end{equation}

Next, we find a lower bound for each $v_i(A_i)$.
For each $j\neq i$, by EFX property we have proved (in fact, EF1 suffices here), there exists an item $g_j\in A_j$ such that $v_i(A_i)\geq v_i(A_j\setminus\{g_j\})$.
By the stopping condition for the while-loop, we have $v_i(A_i)\geq v_i(B)$.
Therefore, by summing over the $n+1$ bundles $A_1,\ldots,A_n,B$, we have
$$(n+1)\cdot v_i(A_i)\geq \sum_{i=1}^nv_i(A_j\setminus\{g_j\})+v_i(B)=\sum_{j=1}^nv_i(A_j)+v_i(B)-\sum_{j\neq i}v_i(\{g_j\}).$$
Since $[m]=B\cup\bigcup_{j=1}^nA_j$ and $\sum_{j\neq i}v_i(\{g_j\})\leq \sum_{g\in G_i}v_i(g)$, this implies
$$(n+1)\cdot v_i(A_i)\geq v_i([m])-\sum_{g\in G_i}v_i(g).$$
Summing over $i=1,\ldots,n$, we have
$$(n+1)\cdot \sum_{i=1}^nv_i(A_i)\geq \sum_{i=1}^nv_i([m])-\sum_{i=1}^n\sum_{g\in G_i}v_i(g).$$
By Inequality~(\ref{eqn:On}) and $\SW(\A)=\sum_{i=1}^nv_i(A_i)$ we have
$$(n+1)\cdot\SW(\A)\geq \sum_{i=1}^nv_i([m])-n\cdot\SW(\A),$$
which implies the desired result $(2n+1)\cdot\SW(\A)\geq \sum_{i=1}^nv_i([m])$.
\end{proof}

\addtocontents{toc}{\protect\setcounter{tocdepth}{1}}
\subsection{Proof of \Cref{thm:mswx_neg_overall}}\label{appendix:hardness-EFX-overall}
\addtocontents{toc}{\protect\setcounter{tocdepth}{2}}

We prove the remaining parts in \Cref{thm:mswx_neg_overall} by the following three lemmas.

\begin{restatable}{lemma}{HardnessEFXPseudoNormalized}
\label{thm:hardness-EFX-pseudo-normalized}
For any constant $\epsilon>0$, a pseudo-polynomial time $n^{0.5-\epsilon}$-approximation algorithm to \mswx with normalized valuations implies $\classP=\classNP$.
\end{restatable}

\begin{proof}
We present a reduction from the independent set problem.
Given an independent set instance $(G=(V,E),x)$ with $x\geq 3$, we construct an \mswx instance as follows.
The set of agents consists of a super agent $s$ and $k$ groups of normal agents $\{a_{i0},a_{i1},\ldots,a_{i|E|}\}_{i=1,\ldots,k}$, where $a_{i1},a_{i2},\ldots,a_{i|E|}$ in each group $i$ correspond to the $|E|$ edges in $G$.
Notice that $n=1+k(|E|+1)$, we can let $n$ be sufficiently large (but also of polynomial size with respect to $G$) such that $n^{1-\epsilon}<(k+1)/2$. In other words, we should choose a $k>4|E|^{1/\epsilon}$ to make the number of the agents in each group small compared with $n$.
The set of items consists of $k+1$ ``super items'' $g_0,g_1,\ldots,g_k$ and $k$ groups of ``normal items'' $\{v_{i1},\ldots,v_{i|V|},e_{i1},\ldots,e_{i|E|}\}_{i=1,\ldots,k}$ such that each group of $|V|+|E|$ items corresponds to the $|V|$ vertices and $|E|$ edges in $G$.
In each group $i$, agent $a_{i0}$ has value $x$ on the super item $g_i$, value $1$ on each of $v_{i1},\ldots,v_{i|V|}$.
For $j=1,\ldots,|E|$, each agent $a_{ij}$ only has positive values on the normal items $v_{i1},\ldots,v_{i|V|},e_{i1},\ldots,e_{i|E|}$ in group $i$.
In particular, $a_{ij}$ has value $1$ on $e_{ij}$ and on the two vertex items $v_{iu_1},v_{iu_2}$ where $u_1$ and $u_2$ are the two endpoints of the $j$-th edge.
The super agent has a value of $w$ for each super item, where $w$ is polynomial in $n$ and larger than the optimal social welfare for the instance without the super agent with no fairness constraints (say, $w=n^{100}$), and she has a value of $0$ for remaining items.

If the independent set instance is a \YES instance, we describe an EFX allocation with social welfare at least $(k+1)w$.
The super agent $s$ gets all the super items $g_0,g_1,\ldots,g_k$.
In each group $i$, agent $a_{i0}$ gets a set of $x$ items from $\{v_{i1},\ldots,v_{i|V|}\}$ corresponding to an independent set of size $x$, and agent $a_{ij}$ (for $j=1,\ldots,|E|$) gets the item $e_{ij}$.
The remaining items are discarded.
It is straightforward to check that the allocation is EFX and is, in fact, envy-free.
The optimal social welfare under EF$X$ constraint is lower bounded by $w(k+1)$.

If the independent set instance is a \NO instance, we will show that the super agent $s$ can get at most one super item in any EFX allocation.
Suppose this is not the case.
A super item $g_i$ with $i=1,\ldots,k$ must be allocated to the super agent $s$, and $s$ is allocated at least one more item.
By EFX, agent $a_{i0}$ cannot envy agent $s$, and must receive a value of at least $x$ from $v_{i1},\ldots,v_{i|V|}$.
This means at least $x$ items from $v_{i1},\ldots,v_{i|V|}$.
Since the independent set instance is a no-instance, agent $a_{i0}$ must receive two items $v_{iu_1},v_{iu_2}$ such that $(u_1,u_2)$ is an edge.
Let $a_{ij}$ and $e_{ij}$ be the agent and the item in the $i$-th group corresponding to this edge respectively.
Then $a_{ij}$ can receive a value of at most $1$ by getting $e_{ij}$, and the value she has on agent $a_{i0}$'s bundle is $2$.
To maintain EFX, agent $a_{i0}$ must not receive more than the two items $v_{iu_1},v_{iu_2}$.
This contradicts our assumption $x\geq 3$.
Since we have proved agent $s$ can get at most one super item, the social welfare, in this case, is upper bounded by $2w$ (as the optimal social welfare of other agents is no more than $w$).

Putting the completeness and the soundness parts together, the inapproximability factor is $(k+1)/2$, which is more than $n^{1-\epsilon}$.
In addition, all the values of the items are bounded by $n^{100}$.
A pseudo-polynomial time algorithm is no more powerful than a polynomial-time algorithm.
\end{proof}

\begin{lemma}\label{lem:hardness-EFX}
For any odd number $n=2k+1$ of agents with $k\geq 1$, it is $\classNP$-hard to approximate \mswx to a factor smaller than $(k+1)$.
\end{lemma}
\begin{proof}
    We present a reduction from the partition problem. Given a partition instance $S=\{e_1,\ldots,e_\ell\}$ such that $\sum_{i=1}^\ell e_i =2x$, we construct a \mswx instance as follows.
    Construct $2k+1$ agents named $\{s,a_1,b_1,\ldots,a_k,b_k\}$ and $m=(k+1)+k\ell$ items named $g_0,g_1,\ldots,g_{k},\{h_{ij}\}_{i=1,j=1}^{k,\ell}$.
    The utility functions of the $n$ agents are defined in the table below, where $w$ is a very large number.
    Let $s$ be the ``super agent'', and the social welfare mostly depends on the value that agent $s$ receives for large $w$.

    \begin{center}
    \begin{tabular}{c|ccccccc}
        \hline
         &  $g_0$ &  $g_1$ & $\cdots$ & $g_k$ &  $h_{1j}$ & $\cdots$ & $h_{kj}$\\
        \hline
        $v_s$     &  $w$ &  $w$ & $\cdots$& $w$ &         &          & \\
        $v_{a_1}$ &      &  $x$ &         &     &  $e_j$  &          &  \\
        $v_{b_1}$ &      &  $x$ &         &     &  $e_j$  &          &  \\
        $\vdots$  &      &      &$\ddots$ &     &         & $\ddots$ &  \\
        $v_{a_k}$ &      &      &         & $x$ &         &          & $e_j$ \\
        $v_{b_k}$ &      &      &         & $x$ &         &          & $e_j$ \\
        \hline
    \end{tabular}
    \end{center}

    If the partition instance is a \YES instance, for each $i\in[k]$, the two agents $a_i$ and $b_i$ can get a value of exactly $x$ from the item set $\{h_{i1},h_{i2},\ldots,h_{i\ell}\}$.
    In this case, the super agent $s$ can get the bundle $\{g_0,g_1,\ldots,g_k\}$, which has value $w(k+1)$, and the social welfare is at least $w(k+1)+2kx$.

    If the partition instance is a \NO instance, the super agent $s$ can get at most one item from $\{g_0,g_1,\ldots,g_k\}$.
    Otherwise, there exists $i\in\{1,\ldots,k\}$ such that $g_i$ is allocated to agent $s$.
    Since the partition instance is a no-instance, no matter how we allocate the remaining items, one of $a_i$ and $b_i$ will receive a value of less than $x$ and so will envy agent $s$.
    Since the super agent receives at least two items, removing an item $g\neq g_i$ from agent $s$'s bundle does not remove the envy from agent $a_i$/$b_i$ to agent $s$, which violates the EFX condition.
    In this case, the social welfare is at most $w+3kx$.

    The lemma concludes as we have shown that, the corresponding social welfare under an EFX allocation is lower-bounded by $w(k+1)+2kx$ in a yes-instance of partition, and is upper-bounded by $w+3kx$ in a no-instance.
    Thus, for any constant $c\in[1,k+1)$, by making $w>\frac{kx(3c-2)}{k+1-c}$, the ratio becomes $\frac{w(k+1)+2kx}{w+3kx}\ge c$, which leads to the $\classNP$-hardness of approximating \mswx to the factor $c$.
\end{proof}

\begin{restatable}{lemma}{HardnessEFXNormalized}
\label{thm:hardness-EFX-normalized}
    For any number $n=k(2k+1)$ of agents with $k\geq 1$, it is $\classNP$-hard to approximate \mswx to a factor smaller than $k/2$ even when agents' valuations are normalized.
\end{restatable}
\begin{proof}
We present a reduction from the partition problem.
Given a partition instance $S=\{e_1,\ldots,e_\ell\}$ with $\sum_{i=1}^\ell e_i=2x$, we construct an \mswx instance as follows.
The $n=k(2k+1)$ agents are partitioned into $k$ groups each of which consists of $2k+1$ agents.
Agents are indexed by 
$$\left\{s^{(t)},a_1^{(t)},b_1^{(t)},a_2^{(t)},b_2^{(t)},\ldots,a_k^{(t)},b_k^{(t)}\right\}_{t=1,\ldots,k},$$ 
where the superscript denotes the group number.
There are $m=1+k(k+1+k\ell)$ items which consist of one item named $f$ and $k$ groups of $(k+1+k\ell)$ items indexed below
$$\left\{g_0^{(t)},g_1^{(t)},\ldots,g_{k}^{(t)},\{h_{ij}^{(t)}\}_{i=1,\ldots,k;j=1,\ldots,\ell}\right\}_{t=1,\ldots,k}.$$
For each group $t=1,\ldots,k$, the utility functions of the agents in group $t$ over the items in group $t$ are defined in the same way as they are in the table in the proof of \Cref{lem:hardness-EFX}.
Agents from one group have a value of $0$ for items in another group.
To make the valuations normalized, we set $w=\frac1{k+1}$, and let the value of item $f$ be $1-3x$ for each agent in $\{a_i^{(t)},b_i^{(t)}\}_{i=1,\ldots,k;t=1,\ldots,k}$ where we set $x$ to be a very small positive number by rescaling the partition instance (the $k$ super agents $s^{(1)},\ldots,s^{(k)}$ have value $0$ on $f$).

If the partition instance is a yes-instance, we describe an EFX allocation with social welfare at least $k$.
In each group $t$, the items $g_0^{(t)},g_1^{(t)},\ldots,g_{k}^{(t)}$ are allocated to agent $s^{(t)}$, and the items in $\{h_{ij}^{(t)}\}_{i=1,\ldots,k;j=1,\ldots,\ell}$ are allocated to the agents in $a_1^{(t)},b_1^{(t)},a_2^{(t)},b_2^{(t)},\ldots,a_k^{(t)},b_k^{(t)}$ such that each of them receives a value of exactly $x$.
The item $f$ is discarded.
It is straightforward to check that the allocation is EFX and has social welfare of at least $k$.

If the partition instance is a no-instance, we will show that the social welfare is at most $2$ for sufficiently small $x$.
A ``lucky'' agent $a_{i^\ast}^{(t^\ast)}$ or $b_{i^\ast}^{(t^\ast)}$ can be allocated the item $f$ with the most value.
However, for most pairs of the ``normal agents'' $(a_i^{(t)},b_i^{(t)})$, if $g_i^{(t)}$ is not allocated to one of them and is allocated to $s^{(t)}$ instead, one of the two agents will envy $s^{(t)}$.
This keeps agent $s^{(t)}$ from getting a bundle more than the single item $g_i^{(t)}$.
Therefore, for the $t$ super agents $s^{(1)},\ldots,s^{(k)}$, at least $k-1$ of them can receive at most one item from $g_0^{(t)},g_1^{(t)},\ldots,g_{k}^{(t)}$, and at most one of them is allowed to receive two items from $g_0^{(t)},g_1^{(t)},\ldots,g_{k}^{(t)}$.
By including the values of the item $f$ and the items $h_{ij}^{(t)}$, the social welfare is at most
$$(1-3x)+3x\cdot k^2+\left((k-1)\cdot w+2w\right).$$
Since $w=\frac1{k+1}$, the social welfare can be arbitrarily closed to $2$ by having $x\rightarrow 0$.
\end{proof}

\section{Omitted Proofs in Section~\ref{sec:mswgeneraln}}
\label{append:msw}

\addtocontents{toc}{\protect\setcounter{tocdepth}{1}}
\subsection{Inapproximability for Constant Number of Agents}
\label{app:inapprox_const}

\begin{restatable}{lemma}{NPhardOneDivFour}
\label{thm:NPhardOneDivFour}
For any $n>2$, \msw is $\classNP$-hard to approximate to factor $n^{1/4}/4.5$, even under normalized valuations.
\end{restatable}
\begin{proof}
First, let $t= k = \lceil n^{\frac{1}{4}}\rceil$ and $M = \lceil n^{\frac{1}{4}}/2\rceil $. It can be verified that when $n\ge 2$, $t+ t \cdot (k^2 - k) \cdot M \le n$ and $k \le 2M$.
Besides, let $x= \frac{1+ (t(k^2-k)M -2)\epsilon}{t(k^2-k)M}$.
Similar to the above reduction, we construct a fair division instance as follows.
``$+$'' means that there is a difference from the previous reduction.
\begin{itemize}[leftmargin=0.5cm]
    \item Clique item: the same as the above reduction. There are $t$ groups of clique items $\mathcal{C}_1, \ldots, \mathcal{C}_k$;
    \item Partition item$^+$: the main difference from the previous reduction is that we create $M$ multiple edges for each pair of vertices of each clique.
    For each multiple-edge, we create a set of partition items by using the given partition instance;
    \item Pool item$^+$: there are $t\cdot (k^2-k)\cdot M$ pool items;
    \item Dummy item$^+$: there are $n-t - t(k^2-k)M$ dummy items.
\end{itemize}
In addition, we change construction of agents correspondingly,
\begin{itemize}[leftmargin=0.5cm]
    \item Super agent: the same as the above reduction, there are $t$ super agents;
    \item Normal agent$^+$: similar to the above reduction, there are $t$ groups of normal agents;
    Each multiple-edge corresponds to two normal agents.
    \item Dummy agent$^+$: there are $n-t - t(k^2-k)M$ dummy agents.
\end{itemize}

The valuation functions are the same as the previous reduction.
When the partition is a \YES instance, we can similarly allocate the clique items $\mathcal{C}_i$ to super agent $s_i$ for each $i\in [t]$.
Meanwhile, for each normal agent, we similarly allocate half of the partition items (with total value $\epsilon$) and one pool item (with value $x-\epsilon$).
Besides, each dummy agent receives exactly one dummy item.
It is not hard to check such an allocation satisfies EF1, and by ignoring the utilities of dummy agents, the total social welfare is no less than
\begin{align}
\SW(\A) \ge t + t(k^2-k)\cdot M \cdot  x
\end{align}

On the other hand, if the partition instance is a \NO instance, we also estimate an upper bound of the number of clique items received by super agents.
Consider an EF1 allocation $\A'$.
Similar to Inequality~\ref{eqn:agents_not_receiving_pool_items}, we have
\begin{align}\label{eqn:ub_of_sum_of_li_withlige1_improved}
\sum_{i=1}^t L_i + \sum_{i=1}^t \frac{L_i^2 - L_i}{2}\cdot M \le t\cdot k \Longrightarrow \sum_{i=1}^t (L_i - 0.5)^2   < \frac{2t\cdot k}{M}
\end{align}
Thus, by C-S Inequality and Inequality~\ref{eqn:ub_of_sum_of_li_withlige1_improved}, the sum of $L_i$ is upper bounded by
\begin{align*}
\sum_{i=1}^t L_i & \le t\sqrt{\frac{2k}{M}} +0.5t \le 1.5t\,. \tag{By $2k\le M$}
\end{align*}
Therefore, the social welfare of allocation $\A'$ is at most
\begin{align*}
\mathcal{SW}(\A') &<\sum_{i=1}^tL_i\cdot \frac1k + tk\cdot x + t(k^2-k)\cdot M \cdot (x-\epsilon)+t(k^2-k)\cdot M \cdot \epsilon + \frac{n-t - t(k^2-k)M}{n-t - t(k^2-k)M}  \\
&= \frac1k\cdot \sum_{i=1}^tL_i + tkx + t(k^2-k)M x+1  \\
&\le \frac1k\cdot \sum_{i=1}^tL_i + 1  + 1 + 1  \tag{$x$ is defined as $ \frac{1+(t(k^2-k)M-2)\epsilon}{t(k^2-k)M}$, and $\epsilon$ is small} \\ 
&\le 1.5 + 3 = 4.5\tag{$\sum_{i=1}^tL_i \le 1.5t$ and $t=k$}
\end{align*}
Therefore, the inapproximability ratio is at least
$\frac{\SW(\A)}{\SW(\A')} > \frac{t}{4.5} \ge \frac{n^{\frac14}}{4.5}$.
\end{proof}
    
\subsection{Approximability for General Number of Agents}
\label{app:approx_general_agents}

\begin{restatable}{lemma}{NPHardN}
\label{thm:NP-hard-n}
For any $\epsilon>0$, \msw is $\classNP$-hard to approximate to within a factor of $n^{\frac13-\epsilon}$, or within a factor of $m^{\frac12-\epsilon}$, even with normalized valuations.
\end{restatable}
We will present a reduction from the maximum independent set problem.
We begin by describing the construction, which will be used for proving the lemma.

\paragraph{The construction.}
Fix a maximum independent set instance $G=(V,E)$, and let $k=|V|$ and $\ell=|E|$.
We will construct a fair division instance with $n=k\ell+k$ agents and $m=2k^2$ items.
We will assume $\ell>2k$ (notice that we can add multi-edges without changing the nature of the independent set problem).

The agents are partitioned into $k$ groups.
Agents are named by $\{a_1^{(j)},\ldots,a_{\ell}^{(j)},s^{(j)}\}_{j=1,\ldots,k}$.
Group $j$ consists of $\ell+1$ agents $a_1^{(j)},\ldots,a_\ell^{(j)},s^{(j)}$.
Each edge $e_i\in E$ in the maximum independent set instance corresponds to $k$ agents $a_i^{(1)},\ldots,a_i^{(k)}$.
In particular, for each group $j$, the $\ell$ agents $a_1^{(j)},\ldots,a_\ell^{(j)}$ represent the $\ell$ edges in $G$, and there is a special agent $s^{(j)}$.
We will use $A^{(j)}$ to denote the set of agents in group $j$.

The items are partitioned into $k+1$ groups defined as follows.
For each $j=1,\ldots,k$, group $j$ contains $k$ items $b_1^{(j)},\ldots,b_k^{(j)}$. Each vertex $u_i\in V$ in the maximum independent set instance corresponds to $k$ items $b_i^{(1)},\ldots,b_i^{(k)}$.
The $(k+1)$-th group contains $k^2$ items $c_1,\ldots,c_{k^2}$.
We will use $B^{(j)}$ to denote the set of items in group $j$ for each $j=1,\ldots,k$ and $C$ to denote the set of items in group $k+1$.

The valuations are defined as follows.
For each special agent $s^{(j)}$, (s)he has value $1/k$ for each item in $B^{(j)}$, and value $0$ for each of the remaining items.
For each agent $a_i^{(j)}$, (s)he has value $\tau/2$ for each of the two items $b_{i_1}^{(j)}$ and $b_{i_2}^{(j)}$ representing the two vertices $u_{i_1}$ and $u_{i_2}$ of the edge $e_i$, (s)he has value $(1-\tau)/k^2$ for each item in $C$, and (s)he has value $0$ for each of the remaining items, where we set $\tau=2/k^2$ so that item $b_{i_1}^{(j)}$ and $b_{i_2}^{(j)}$ have a slightly higher value than each item in $C$.

\paragraph{Proof sketch.}
Notice that, for each $j=1,\ldots,k$, the agent group $A^{(j)}\setminus\{s^{(j)}\}$ and the item group $B^{(j)}$ resembles the graph $G=(V,E)$ in that each edge is represented by exactly one agent in $A^{(j)}\setminus\{s^{(j)}\}$ and each vertex is represented by exactly one item in $B^{(j)}$.
The $k$ groups can be viewed as $k$ copies of the maximum independent set instance.

For each group $j$, the items in $B^{(j)}$ have much higher values to agent $s^{(j)}$ than to any other agents.
We need to maximize the number of the items in $B^{(j)}$ that are allocated to $s^{(j)}$.

On the other hand, the number of the ``non-special'' agents, $k\ell$, is more than the number of items $2k^2$.
This implies some non-special agents will receive no item at all.
Since all the non-special agents have a common valuation on items in $C$, to guarantee EF1, no agent can receive more than one item in $C$.

It is then easy to see that, in each group $j$, the items allocated to agent $s^{(j)}$ must correspond to an independent set in $G$.
For otherwise, if $s^{(j)}$ receives both items $b_{i_1}^{(j)}$ and $b_{i_2}^{(j)}$ for an edge $e_i=(u_{i_1},u_{i_2})$, agent $a_{i}^{(j)}$ will have to receive at least two items in $C$ to guarantee EF1 (recall that, for agent $a_i^{(j)}$, the value of $b_{i_1}^{(j)}$ or $b_{i_2}^{(j)}$ is higher than the value of any item in $C$), and we have seen that this is infeasible.

Since we would like to maximize the number of items in $B^{(j)}$ that are allocated to agent $s^{(j)}$, our problem now naturally becomes the problem of maximizing the size of the independent set in $G$.
Notice that the $k$ groups simulate the $k$ \emph{identical} copies of the maximum independent set instance.

Since $\ell=O(k^2)$, we have $n=O(k^3)$ and $m=\Theta(k^2)$.
\Cref{thm:NP-hard-n} holds due to \Cref{thm:indset}.

\paragraph{Formal proof.}
For the completeness part, we prove the following proposition.
\begin{proposition}
If $G$ has an independent set of size $t$, then there exists an allocation with social welfare at least $t$.
\end{proposition}
\begin{proof}
Let $I\subseteq V$ be an independent set of $G$ with $|I|=t$.
Consider the following allocation.
For each $j=1,\ldots,k$, allocate $\{b_i^{(j)}\mid u_i\in I\}$ to agent $s^{(j)}$.
Allocate remaining items arbitrarily subject to that each remaining agent in $\{A^{(j)}\setminus\{s^{(j)}\}\mid j=1,\ldots,k\}$ receive at most one item.

It is straightforward to see that the allocation is EF1.
Since each agent in $\{A^{(j)}\setminus\{s^{(j)}\}\mid j=1,\ldots,k\}$ receives at most one item, no agent will EF1-envy any of them.
It remains to show that no one will envy $s^{(j)}$ for each $j$.
Firstly, those special agents $\{s^{(j)}\}$ will not envy each other.
This is because each $s^{(j)}$ receives items from the group $B^{(j)}$ only, and (s)he only has positive values for items in $B^{(j)}$.
Secondly, each non-special agent $a_i^{(j')}$ will not envy each special agent $s^{(j)}$.
If $j\neq j'$, agent $a_i^{(j')}$ has $0$ value for every item in $B^{(j)}$, and thus will not envy $s^{(j)}$.
If $j=j'$, the only two items in $B^{(j)}$ that agent $a_i^{(j')}$ has non-zero values are $b_{i_1}^{(j)}$ and $b_{i_2}^{(j)}$, where $e_i=(u_{i_1},u_{i_2})$ is the edge corresponding to $a_i^{(j')}$.
Since $s^{(j)}$ receives items that correspond to an independent set, $s^{(j)}$ receives at most one of $b_{i_1}^{(j)}$ and $b_{i_2}^{(j)}$.
Thus, agent $a_i^{(j')}$ does not envy $s^{(j)}$.

Finally, we find a lower bound on social welfare by only considering special agents.
Each $s^{(j)}$ receives $t$ items, and each of them has value $1/k$.
Since there are $k$ special agents, the social welfare is at least $k\cdot t\cdot \frac1k=t$.
\end{proof}

For the soundness part, we prove the following proposition.
\begin{proposition}
If $G$ does not have an independent set of size larger than $t$, then the social welfare of any allocation is at most $t+2$.
\end{proposition}
\begin{proof}
Consider an arbitrary allocation $\A$.
The crucial observation is that $A_{s^{(j)}}\cap B^{(j)}$ must correspond to an independent set.
To see this, suppose for an edge $e_i=(u_{i_1},u_{i_2})$ some agent $s^{(j)}$ receives both items $b_{i_1}^{(j)}$ and $b_{i_2}^{(j)}$.
Recall that agent $a_i^{(j)}$ has value $\tau/2=1/k^2$ on each of both items and (s)he has value $(1-\tau)/k^2$ on each item in $C$ which is less than $1/k^2$
To guarantee EF1, agent $a_i^{(j)}$ must receive at least two items in $C$.
To guarantee each remaining non-special agent does not envy agent $a_i^{(j)}$, each non-special agent should receive at least one item.
The total number of non-special agents is $k\ell$, which is more than the total number of items $2k^2$.
Thus, EF1 cannot be guaranteed if $A_{s^{(j)}}\cap B^{(j)}$ does not correspond to an independent set.

With this observation and the assumption that the maximum independent set has a size no more than $t$, each special agent $s^{(j)}$ receives a bundle with value at most $t/k$.
Thus, the overall utility for all the special agents is at most $t$.
For the non-special agent, each item is worth at most $\tau/2=1/k^2$.
Even if all the $2k^2$ items are allocated to the non-special agents, the overall utility for all the non-special agents is bounded by $2k^2\cdot\frac1{k^2}=2$.
Therefore, the social welfare of $\A$ is at most $t+2$.
\end{proof}

To conclude the proof, \Cref{thm:indset} and the two propositions above imply it is $\classNP$-hard to approximate $\SW(\A)$ to within a factor of $k^{1-\epsilon}$ for any $\epsilon>0$.
Since $m=2k^2$, the inapproximability factor can be written as $\frac{1}{(\sqrt{2})^{1-\epsilon}}m^{\frac12-\frac12\epsilon}$, which is more than $m^{\frac12-\epsilon}$ by choosing $\epsilon$ appropriately.
This concludes \Cref{thm:NP-hard-n} for the part with $m$.
Since $\ell=O(k^2)$, we have $n=O(k^3)$.
We can see \Cref{thm:NP-hard-n} for the part with $n$ holds by rewriting $k^{1-\epsilon}$ in a similar way.

\begin{restatable}{theorem}{OneDivN}
\label{thm:1/n}
There exists a polynomial-time $n$-approximation algorithm for \msw.            
\end{restatable}

\begin{proof}
We present a variant of the round-robin algorithm.
Like the round-robin algorithm, our algorithm consists of $\lceil m/n\rceil$ iterations.
In each iteration except for the last one, exactly $n$ items are allocated such that each agent receives exactly one item.
However, in each iteration, unlike the standard round-robin algorithm where agents receive items in an arbitrary order, our ``greedy-based round-robin'' algorithm greedily chooses the agent with the highest value for a remaining item who has not yet received an item in the current iteration.
In particular, we first find a tuple $(i,g)$ with maximum $v_i(g)$ such that agent $i$ has not been allocated an item yet in the current iteration and item $g$ has not been allocated, and we allocate item $g$ to agent $i$.
We do this $n$ times until each agent receives exactly one item in this iteration.

Firstly, the allocation output by our algorithm is EF1, for the same reason that the standard round-robin algorithm is EF1.
The crucial observation here is that, for every agent, the item (s)he receives at a particular iteration has a (weakly) larger value than the value of any item that is allocated in the later iterations.
Therefore, for any pair of agents $i$ and $j$, in agent $i$'s valuation, the item allocated to agent $i$ at the $t$-th iteration has a (weakly) larger value than the item allocated to agent $j$ at the $(t+1)$-th iteration.
Thus, if removing the item allocated to agent $j$ in the first iteration from agent $j$'s bundle, agent $i$ will not envy agent $j$.

We will then show that this is a $n$-approximation algorithm for \msw.
Let $I_t$ be the set of items allocated in the $t$-th iterations.
For each item $g$, let $u^\ast_g=\max_{i=1,\ldots,n}v_i(g)$.
It is obvious that
\begin{equation*}
   \sum_{g=1}^mu^\ast_g=\sum_{t=1}^{\lceil m/n\rceil}\sum_{g\in I_t}u^\ast_g 
\end{equation*}
is an upper bound to optimal social welfare.
On the other hand, let $o_t\in I_t$ be the first item allocated in the $t$-th iteration, and let $a_t$ be the agent who receives $o_t$.
By the nature of our algorithm, we have $v_{a_t}(o_t)=u^\ast_{o_t}\geq u^\ast_g$ for any $g\in I_t$.
By only accounting for the items $o_1,\ldots,o_{\lceil m/n\rceil}$, the social welfare for the allocation $\A$ output by our algorithm satisfies
\begin{equation*}
\SW(\A)\geq \sum_{t=1}^{\lceil m/n\rceil}u^\ast_{o_t}\geq\sum_{t=1}^{\lceil m/n\rceil}\left(\frac1n \sum_{g\in I_t}u^\ast_g\right)=\frac1n\sum_{g=1}^mu^\ast_g,
\end{equation*}
which shows that our algorithm is a $n$-approximation algorithm for \msw.
\end{proof}

\section{Omitted Pseudo-codes for Bi-Criteria Optimization}
\begin{algorithm}[H]
\caption{Bi-criteria optimization of \mswx}
\label{bicriteriax}
\KwInput{utility functions $v_1,\ldots,v_n$, item set $M=[m]$, and the parameter $\epsilon>0$}
\KwOutput{an $(1-\epsilon)$-approximate EFX allocation.} 
Set $K\leftarrow\lceil \frac{3mn}{\epsilon} \rceil$\;
Initialize $\Pi\leftarrow\varnothing$\;
\tcc{ $\Pi$ stores candidate allocations} 
\For{each \textbf{feasible} $X=\{X_i=\{g_{i1},\ldots,g_{i(i-1)},g_{i(i+1)},\ldots,g_{in}\}\mid i=1,\ldots,n\}$}{
\tcc{elements in $X_i$ may be repeated, but $X_i\cap X_j=\emptyset$ for any $i,j$} 
\tcc{for any $i,j,k$, $v_i(g_{ji})\leq v_i(g_{jk})$} \label{bicriteriaxcond}
\For{each $i=1,\ldots,n$}{
Set $Y_i\leftarrow\{g_{1i},\ldots,g_{(i-1)i},g_{(i+1)i},\ldots,g_{ni}\}$\;
Set $V_i\leftarrow v_i([m]\setminus Y_i)$\;
Set $\tau_i\leftarrow V_i/K$\;
for each $o\in [m]\setminus Y_i$, set $\overline{v}_i(o)\leftarrow\max_{k:k\tau_i\leq v_i(o)}k\tau_i$\;
for each $o\in Y_i$, set $\overline{v}_i(o)\leftarrow 0$\;
Add a new item $d_i$ to $M$ such that $\overline{v}_i(d_i)=m\tau_i$ and $\overline{v}_j(d_i)=0$ for each $j\neq i$\ and $v_j(d_i)=0$ for each $j\in[n]$\; 
$X_i\leftarrow X_i\cup d_i$\;
}
$\A\leftarrow${\FDynamicProgramEfx}($M,\overline{v}_1,\ldots,\overline{v}_n,v_1,\ldots,v_n,\tau_1,\ldots,\tau_n,X,K$)\tcp*{see Algorithm~\ref{bicriteriaxDP}} 
Include $\A$ in $\Pi$\;
}
\Return{the allocation in $\Pi$ with the largest social welfare with respect to $\{v_1,\ldots,v_n\}$}  %
\end{algorithm}

\begin{algorithm}[ht]
\caption{The dynamic programming subroutine for \mswx}
\label{bicriteriaxDP}
\SetKwProg{Fn}{Function}{:}{}
\Fn{\FDynamicProgramEfx{$M=[m+n],\overline{v}_1,\ldots,\overline{v}_n,v_1,\ldots,v_n,\tau_1,\ldots,\tau_n,X,K$}}{
\tcc{$X=\{X_i=\{g_{i1},\ldots,g_{i(i-1)},g_{i(i+1)},\ldots,g_{in},d_i\}\mid i=1,\ldots,n\}$}
\tcc{for each $i$, dummy item $d_i$ is the $(m+i)$-th item }
\tcc{for each $i$, $\overline{v}_i(S)/\tau_i\in \{0,1,\ldots,K+m\}$ holds for any $S\subseteq M$}
\For{each $\chi\in\{0,1,\ldots,K+m\}^{n^2}$ and each $t=0,1,\ldots,m+n$}{ 
Initialize $H[\chi,t]\leftarrow\nil$\;
}
$H[0^{n^2},0]\leftarrow 0$\;
\For{each $t=0,\ldots,m+n-1$}{ 
\For{each $\chi$ in the dictionary ascending order such that \emph{$H[\chi,t]\neq \nil$}}{
\If{item $t+1$ belongs to $X$}{  
Suppose  $t\in X_{i^\ast}$\;
{\sc Update}($\chi,t,i^\ast$)\; 
}
\Else{ 
\For{each $i=1,\ldots,n$}{
\If{$v_j(g_{ij})\le v_j(t+1)$ for all $j\neq i$ \label{bicriteriaxDP-newif}}{
{\sc Update}($\chi,t,i$)\; 
}
}
}
}
}
\For{each $\chi\in\{0,1,\ldots,K+m\}^{n^2}$}{ 
    Set $H[\chi,m+n]\leftarrow\nil$ if the allocation stored in $H[\chi,m+n]$ is not envy-free w.r.t. $\overline{v}_1,\ldots,\overline{v}_n$\;
}
\If{\emph{$H[\chi,m+n]\neq\nil$} for some $\chi$}{
    \Return{the allocation in \emph{$\{H[\chi,m+n]\mid H[\chi,m+n]\neq\nil,\chi\in\{0,1,\ldots,K+m\}^{n^2}\}$} with the largest social welfare w.r.t. $v_1,\ldots,v_n$, i.e. the one with the largest value $H[\chi,m+n]$} 
}
\Else{
    \Return{\emph{$\nil$}} 
}
}
\Fn{\Update{$\chi,t,i$}\label{bicriteriaxDP-upb}}
{
            for each $j$, set $\chi_{ji}'\leftarrow\chi_{ji}+\overline{v}_j(t+1)/\tau_j$\; 
            for each $(i',j)$ with $i'\neq i$, set $\chi_{ji'}'\leftarrow\chi_{ji}$\;
            \If{\emph{$H[\chi',t+1]==\nil$} or $H[\chi',t+1]<H[\chi,t]+v_i(t+1)$}
            {
            $H[\chi',t+1]\leftarrow H[\chi,t]+v_i(t+1)$\; 
            }
}
\end{algorithm}

\end{document}